\newif\ifabstract
\abstracttrue
\newif\iffull
\ifabstract \fullfalse \else \fulltrue \fi
\documentclass[11pt]{article}
\usepackage{amsfonts}
\usepackage{amssymb}
\usepackage{amstext}
\usepackage{amsmath}
\usepackage{xspace}
\usepackage{ntheorem}
\usepackage{graphicx}
\usepackage{url}
\usepackage{graphics}
\usepackage{colordvi}
\usepackage{hyperref}
\usepackage{subfigure}
\usepackage{xcolor}
\usepackage{cleveref}

\advance \oddsidemargin by -0.8in
\newcommand{\myparskip}{3pt}
\parskip \myparskip

\newcommand{\td}{\tilde d}

\newcommand{\cCMG}{c_{\mathsf{CMG}}}
\newcommand{\Elong}{E^{\mathsf{long}}}

\newcommand{\otilde}{\widetilde O}
\newcommand{\ohat}{\widehat O}

\newcommand{\MBM}{\mbox{\sf{maximum bipartite matching}}\xspace}
\newcommand{\MWU}{\mbox{\sf{MWU}}\xspace}

\newcommand{\SSSP}{\mbox{\sf{SSSP}}\xspace}
\newcommand{\APSP}{\mbox{\sf{APSP}}\xspace}

\newcommand{\algembedorcut}{\ensuremath{\operatorname{AlgEmbedOrCut}}\xspace}

\newcommand{\val}{\operatorname{val}}
\newcommand{\ceil}[1]{\ensuremath{\left\lceil#1\right\rceil}}
\newcommand{\floor}[1]{\ensuremath{\left\lfloor#1\right\rfloor}}

\newcommand{\opt}{\mathsf{OPT}}

\newcommand{\set}[1]{\left\{ #1 \right\}}

\newcommand{\tset}{{\mathcal T}}

\newcommand{\pset}{{\mathcal{P}}}
\newcommand{\qset}{{\mathcal{Q}}}

\newcommand{\xset}{{\mathcal{X}}}

\newcommand{\oset}{{\mathcal{O}}}

\newcommand{\be}{\begin{enumerate}}
\newcommand{\ee}{\end{enumerate}}
\newcommand{\bd}{\begin{description}}
\newcommand{\ed}{\end{description}}
\newcommand{\bi}{\begin{itemize}}
\newcommand{\ei}{\end{itemize}}

\newtheorem{theorem}{Theorem}[section]

\newtheorem{observation}[theorem]{Observation}
\newtheorem{corollary}[theorem]{Corollary}
\newtheorem{claim}[theorem]{Claim}

\newtheorem{definition}[theorem]{Definition}
\newenvironment{proof}{\par \smallskip{\bf Proof:}}{\hfill\stopproof}
\def\stopproof{\square}
\def\square{\vbox{\hrule height.2pt\hbox{\vrule width.2pt height5pt \kern5pt
\vrule width.2pt} \hrule height.2pt}}

\newenvironment{proofof}[1]{\noindent{\bf Proof of #1.}}%
        {\hfill\stopproof}

\newcommand{\algmwu}{{\sc ALG-MWU}\xspace}

\newenvironment{prog}[1]{
\begin{minipage}{5.8 in}
\begin{center}
{\sc #1}
\end{center}
}
{
\end{minipage}
}

\newcommand{\program}[3]{\begin{figure} \fbox{\vspace{2mm}\begin{prog}{#1} #3 \end{prog}\vspace{2mm}} 
			\caption{#1 \label{#2}} \end{figure}}
\renewcommand{\phi}{\varphi}
\newcommand{\eps}{\epsilon}

\newcommand{\poly}{\operatorname{poly}}
\newcommand{\dist}{\mbox{\sf dist}}
\newcommand{\tdist}{\widetilde{\mathsf {dist}}}

\newenvironment{properties}[2][0]
{
\begin{enumerate} \setcounter{enumi}{#1}}{\end{enumerate}}

\newcommand{\vol}{\operatorname{Vol}}

\newcommand{\attime}[1][\tau]{^{(#1)}}

\newcommand{\skipp}{\operatorname{skip}}
\newcommand{\spann}{\operatorname{span}}

\newcommand{\dmax}{\delta_{\mbox{\textup{\footnotesize{max}}}}}

\newcommand{\DLSSSP}{{\sf {DAG-like SSSP}}\xspace}
\newcommand{\ATO}{\ensuremath{\mathcal{ATO}}\xspace}

\newcommand{\shortpath}{\mbox{\sf{short-path}}\xspace}

\newcommand{\pquery}{\mbox{\sf{path-query}}\xspace}

\newcommand{\constructexpander}{\ensuremath{\mathsf{ConstructExpander}}\xspace}
\newcommand{\shortestpath}{\mbox{\sf{shortest-path}}\xspace}

\newcommand{\maintaincluster}{\ensuremath{\mathsf{MaintainCluster}}\xspace}

\newcommand{\pathquery}{\mbox{\sf{path-query}}\xspace}

\newcommand{\EST}{\mbox{\sf{ES-Tree}}\xspace}

\begin{document}

\begin{titlepage}
	
	\title{A Faster Combinatorial Algorithm for Maximum Bipartite Matching}
\author{Julia Chuzhoy\thanks{Toyota Technological Institute at Chicago, IL. Email: {\tt cjulia@ttic.edu}. Supported in part by NSF grant CCF-2006464 and NSF HDR TRIPODS award 2216899.}\and Sanjeev Khanna \thanks{University of Pennsylvania, Philadelphia, PA. Email: {\tt sanjeev@cis.upenn.edu}. Supported in part by NSF awards CCF-1934876 and CCF-2008305. }}
	\maketitle
\pagenumbering{gobble}
	
	\thispagestyle{empty}
\begin{abstract}
	The maximum bipartite matching problem is among the most fundamental and well-studied problems in combinatorial optimization. A beautiful and celebrated combinatorial algorithm of Hopcroft and Karp~\cite{HK73} shows that maximum bipartite matching can be solved in $O(m \sqrt{n})$ time on a graph with $n$ vertices and $m$ edges. 
For the case of very dense graphs, a different approach based on fast matrix multiplication was subsequently developed~\cite{IM81, MS04}, that achieves a running time of $O(n^{2.371})$. 
For the next several decades, these results represented the fastest known algorithms for the problem 
until in 2013, a ground-breaking work of Madry~\cite{Madry13} gave a significantly faster algorithm for sparse graphs.
Subsequently, a sequence of works developed increasingly faster algorithms for solving maximum bipartite matching, and more generally directed maximum flow, culminating in a spectacular recent breakthrough~\cite{ChenKLPGS22} that gives an $m^{1+o(1)}$ time algorithm for maximum bipartite matching (and more generally, for min cost flows). These more recent developments collectively represented a departure from earlier combinatorial approaches: they all utilized continuous techniques based on interior-point methods for solving linear programs. 

This raises a natural question: are continuous techniques essential to obtaining fast algorithms for the bipartite matching problem? Our work makes progress on this question by presenting a new, purely combinatorial algorithm for bipartite matching, that, on moderately dense graphs outperforms both Hopkroft-Karp and the fast matrix multiplication based algorithms. 
Similar to the classical algorithms for bipartite matching, our approach is based on iteratively augmenting a current matching using augmenting paths in the (directed) residual flow network. A common method for designing fast algorithms for directed flow problems is via the multiplicative weights update (MWU) framework, that effectively reduces the flow problem to decremental single-source shortest paths (SSSP) in directed graphs. Our main observation is that a slight modification of this reduction results in a {\em special case} of SSSP that appears significantly easier than general decremental directed SSSP. Our main technical contribution is an efficient algorithm for this special case of SSSP, that outperforms the current state of the art algorithms for general decremental SSSP with adaptive adversary, leading to a deterministic algorithm for bipartite matching, whose running time is $\tilde{O}(m^{1/3}n^{5/3})$. This new algorithm thus starts to outperform the Hopcroft-Karp algorithm in graphs with $m = \omega(n^{7/4})$, and it also outperforms the fast matrix multiplication based algorithms on dense graphs. We believe that this framework for obtaining faster combinatorial algorithms for bipartite matching by exploiting the special properties of the resulting decremental SSSP instances is one of the main conceptual contributions of our work.

Finally, using a standard reduction from the maximum vertex-capacitated $s$-$t$ flow problem in directed graphs to maximum bipartite matching, we also obtain an $\tilde{O}(m^{1/3}n^{5/3})$-time deterministic algorithm for maximum vertex-capacitated $s$-$t$ flow when all vertex capacities are identical. 
	\end{abstract}

\newpage

\tableofcontents{}
\end{titlepage}

\pagenumbering{arabic}

\section{Introduction}
In the \MBM problem, the input is a bipartite graph $G$  with $n$ vertices and $m$ edges, and the goal is compute a matching of maximum size in $G$. This problem is among the most extensively studied in computer science, combinatorial optimization, and operations research, and holds a central place in graph algorithms with 
intimate connections to many other fundamental graph optimization problems (see e.g.\cite{Sch03}). 

A standard textbook application of the maximum $s$-$t$ flow problem shows that the task of finding an optimal bipartite matching can be reduced to finding a maximum $s$-$t$ flow in a directed graph with unit edge capacities. This connection immediately gives an $O(mn)$-time algorithm for bipartite matching by simply invoking the basic Ford-Fulkerson algorithm~\cite{FF56} on the resulting flow network. The resulting algorithm is based on a simple idea: as long as the current matching is not optimal, it can be augmented by finding an $s$-$t$ path in the residual flow network in $O(m)$ time. A celebrated result of Hopcroft and Karp~\cite{HK73} showed an elegant way to speed up this simple algorithm. Instead of finding a single $s$-$t$ path at a time, the Hopcroft-Karp (HK) algorithm finds a {\em maximal collection} of internally disjoint augmenting  $s$-$t$ paths of shortest possible length in the residual flow network, using a combination of breadth-first search (BFS) and depth-first search (DFS). This modification improves the runtime to $O(m \sqrt{n})$. For the ensuing several decades, the HK algorithm remained the fastest algorithm for bipartite matching except for the case of very dense graphs. For dense graphs, a subsequent work based on fast matrix multiplication was shown to give an $O(n^{\omega})$ time algorithm for bipartite matching, beating the HK algorithm in the dense regime~\cite{IM81, MS04}. 

After many decades of no further progress on bipartite matching, starting in the 2000s, a new paradigm started to emerge for solving combinatorial problems like bipartite matching and maximum $s$-$t$ flows, using continuous techniques. A breakthrough result of Spielman and Teng~\cite{ST04}, achieving a near-linear time algorithm for solving Laplacian systems, paved the way for this new method, called the Laplacian paradigm. As a first illustration of this paradigm, Daitch and Spielman~\cite{DS08} showed that the task of computing maximum $s$-$t$ flow in a directed graph can be reduced to solving $\otilde(m^{1/2})$ Laplacian systems, thus obtaining an $\otilde(m^{3/2})$ time algorithm for directed max flow.
Later, Madry~\cite{Madry13}  showed that maximum flow in directed graphs, and hence bipartite matching, can be computed in $\tilde{O}(m^{10/7})$ time using this paradigm. Exactly 40 years after the discovery of the HK algorithm, this became the first algorithm to substantially outperform it at least on sparse graphs. The following decade saw a sequence of remarkable improvements~\cite{Madry16, LS19, CMSV17, LS20_stoc, AMV20} with another big milestone arriving when van den Brand {\em et al.}~\cite{BrandLNPSS0W20} gave an $\tilde{O}(m + n^{1.5})$ time algorithm for bipartite matching. Finally, in another major breakthrough, Chen {\em et. al.}~\cite{ChenKLPGS22} gave an  $m^{1+o(1)}$-time algorithm for solving maximum flow in directed graphs, thereby obtaining an almost linear time algorithm for bipartite matching in all edge-density regimes.
At a high level, these developments are all based on formulating a directed flow problem as a linear program, solving the resulting linear program by efficiently utilizing interior-point methods (IPM), whereby each iteration of IPM reduces the directed graph problem to either solving a Laplacian system, or another efficiently solvable problem on undirected graphs (e.g., min ratio cycle in~\cite{ChenKLPGS22}). This general approach is then further combined with dynamic graph data structures to avoid re-computation from scratch in each successive IPM iteration, resulting in highly efficient new algorithms for problems like bipartite matching and maximum $s$-$t$ flow. 

This history of algorithmic developments for bipartite matching and related problems raises a natural question: are continuous techniques essential to obtaining fast algorithms for \MBM? On one hand, the use of continuous techniques based on IPMs have led to an almost linear time algorithm for bipartite matching, strongly outperforming the HK algorithm in all edge density regimes. On the other hand, these improvements come at the expense of the simplicity and the transparent nature of the augmenting path based combinatorial algorithms. 
In fact, this question can be asked more broadly for many other graph problems, where significant recent advances have come by strongly utilizing continuous techniques.

We make progress on the above-mentioned question for bipartite matching by presenting a new, purely combinatorial, augmenting path based algorithm, which on moderately dense graphs outperforms the Hopcroft-Karp algorithm, and on dense graphs outperforms the fast matrix multiplication based algorithms, proving the following theorem. 

\begin{theorem}
\label{thm:main}
There is a deterministic combinatorial algorithm for \MBM that, on a graph $G$ with $n$ vertices and $m$ edges, has running time $\otilde(m^{1/3}n^{5/3})$.
\end{theorem}

We note that for moderately dense graphs with $m = \omega(n^{7/4})$, our algorithm  outperforms the 
Hopcroft-Karp algorithm~\cite{HK73}. Furthermore, since the best current upper bound on the matrix multiplication exponent $\omega$ is roughly $2.371866$~\cite{DWZ22}, on dense graphs, our algorithm also considerably outperforms the matrix multiplication based approach for bipartite matching. 

Using a standard reduction from vertex-capacitated flow in directed graphs to \MBM (see Theorem 16.12 in~\cite{Sch03}, for instance), we also obtain an $\otilde(m^{1/3}n^{5/3})$ time deterministic algorithm for maximum vertex-capacitated flow when all vertex capacities are identical. 

\begin{corollary}
\label{cor:main}
There is a deterministic combinatorial algorithm for directed maximum $s$-$t$ flow problem with uniform vertex capacities that runs in $\otilde(m^{1/3}n^{5/3})$ time on a graph $G$ with $n$ vertices and $m$ edges.
\end{corollary}

Similar to the classical algorithms for \MBM, our approach for proving Theorem~\ref{thm:main} is based on iteratively augmenting a current matching using augmenting paths in the residual flow network. We do so by using the multiplicative weights update (\MWU) framework, that effectively reduces the underlying flow problem to decremental single-source shortest paths (\SSSP) in directed graphs, a connection first observed by Madry~\cite{madry2010faster}. Our main observation is that a slight modification of this reduction results in a {\em special case} of \SSSP that appears significantly easier than general decremental directed \SSSP. As typical for such reductions, we need the algorithm for decremental directed \SSSP to work against an \emph{adaptive adversary}; in other words, the sequence of edge deletions from the graph may depend on past behavior of the algorithm, including its responses to queries.

While there is an extensive body of work on the decremental \SSSP problem in undirected graphs, with the best current algorithms even against an adaptive adversary achieving essentially optimal parameters (see \cite{bernstein2022deterministic}), our understanding of the same problem in directed graphs is still far from complete. For simplicity, in this discussion we assume that all edge lengths are integers that are bounded by $\poly(n)$. For the special case where the input graph $G$ is a {\em directed acyclic graph} (DAG), \cite{AlmostDAG2}, building on the results of \cite{Bernstein,gutenberg2020decremental}, provided a deterministic $(1+\eps)$-approximation algorithm with total update time $\otilde(n^2/\eps)$. They also introduced the Approximate Topological Order (\ATO) framework, whose purpose is to reduce the decremental \SSSP problem in general directed graphs to that in graphs that are ``close'' to DAG's\footnote{According to \cite{AlmostDAG2}, while the \ATO technique was first explicitly defined in  \cite{AlmostDAG2}, the work of \cite{gutenberg2020decremental} can be viewed as (implicitly) using this technique.}. Using this approach, they provided a $(1+\eps)$-approximation algorithm for decremental \SSSP in directed graphs with total update time $\otilde(n^2)$, but only for the \emph{oblivious adversary} setting, where the sequence of edge deletions is fixed in advance and may not depend on algorithm's past behavior. To the best of our knowledge, the best current algorithm for directed decremental \SSSP in the adaptive-adversary setting, due to Bernstein~{\em et al.}~\cite{SCC}, obtains a $(1+\eps)$-approximation with total update time  $n^{8/3+o(1)}$. This total update time however is prohibitively high for our purposes.

% give a {\em deterministic} decremental directed \SSSP algorithm with a total update time of $n^{8/3+o(1)}$ which works in general directed graphs but is not fast enough for our goal of outperforming previous best combinatorial algorithms for \MBM. We note here that in a related work, Bernstein~{\em et al.}~\cite{AlmostDAG2} give a {\em randomized} decremental \SSSP algorithm with a significantly better total update time of $\otilde(n^2)$. However, this algorithm only works against oblivious adversaries which is not sufficient for use with \MWU framework since the updates to the underlying graph are generated in response to the algorithms's output. 
Our main technical contribution is a deterministic algorithm for directed decremental \SSSP for the special family of instances that arise in the reduction from \MBM. Our work builds on the algorithm of \cite{SCC}, and exploits the special properties of the instances that we obtain in order to both significantly improve the running time, and simplify the algorithm itself and its analysis.

We believe that the general framework for obtaining faster combinatorial algorithms for bipartite matching presented in this paper is one of our main conceptual contributions, and we hope that it will lead for faster combinatorial algorithms for this and other related problems.

\noindent
{\bf Follow-up Work.} A very recent follow-up work \cite{SKnew} provides a randomized combinatorial $n^{2+o(1)}$-time algorithm for \MBM. This new result relies on the general approach presented here as a starting point. %We show that there is a randomized combinatorial algorithm for \MBM that runs in $n^{2+o(1)}$ time. 
On dense graphs, this new algorithm essentially matches the performance of the fastest known algorithms for \MBM.

We now provide a high-level overview of our algorithm.

\subsection{Our Techniques}

It is well known that the \MBM problem in a bipartite graph $G$ can be equivalently cast as the problem of computing an integral maximum $s$-$t$ flow in a directed flow network $G'$ with unit edge capacities. 
We can view any given matching $M$ in $G$ as defining an $s$-$t$ flow $f$ in $G'$ of value is $|M|$. We let $H=G'_f$ be the corresponding residual flow network, that we refer to as the \emph{residual flow network with respect to matching $M$}. The residual network $H$ contains an integral $s$-$t$ flow of value $\Delta=\opt-|M|$, where $\opt$ is the value of the maximum bipartite matching. Moreover, any integral $s$-$t$ flow $f'$ in $H$ that respects edge capacities, can be transformed into a collection $\qset$ of $\val(f')$ internally disjoint $s$-$t$ paths (augmenting paths), which can then in turn  be used in order to obtain a new matching $M'$ in $G$, of cardinality $|M|+\val(f')$. The task of augmenting $M$ to a larger matching is then equivalent to recovering a large flow in $H$.

One standard approach for obtaining fast algorithms for various flow problems, due to \cite{GK98, Fleischer00,AAP93}, is via the Multiplicative Weight Update (\MWU) method. It was further observed by Madry \cite{madry2010faster}  that this approach can be viewed as reducing a given flow problem to a variant of decremental \SSSP or \APSP.  While strong lower bounds are known for exact algorithms for decremental \SSSP and \APSP (see, e.g. \cite{DorHZ00,RodittyZ11,HenzingerKNS15,abboud2022hardness,DBLP:conf/stoc/AbboudBF23}), for our purposes, it is sufficient to only approximately solve the resulting decremental \SSSP problem, and recover an approximate max flow in $H$. Specifically, we will iteratively augment the current matching, where the goal of a single iteration is to augment $M$ to a new matching $M'$ that satisfies $|M'|\geq |M|+\Omega(\Delta/\poly\log n)$, where $\Delta=\opt-|M|$. This will ensure that the number of iterations is bounded by $O(\poly\log n)$. We now focus on the implementation of a single iteration. 

In a single iteration, we use the \MWU framework to find an $s$-$t$ flow in the residual flow network $H$ of value $\Omega(\Delta/\poly\log n)$ with $O(\log n)$ congestion. By scaling down, we may equivalently view this as giving a \emph{fractional} $s$-$t$  flow in the residual flow network $H$ of value $\Omega(\Delta/\poly\log n)$ with no congestion. 
The support of this flow is $\otilde(n)$ edges, and we can now use the blocking flow approach of the Hopcroft-Karp~\cite{HK73} algorithm to recover an \emph{integral} $s$-$t$ flow of value $\Omega(\Delta/\poly\log n)$ in deterministic $\otilde(n^{3/2})$ time.
We note that it is possible to round the fractional flow into an integral flow in $\otilde(n)$ time using the random walk technique of~\cite{GKK10, LRS13}. However, the improved runtime for this step does not improve our overall runtime and comes at the expense of making the algorithm randomized, so we will rely on the deterministic rounding scheme. The remaining challenge is in the efficient implementation of the \MWU framework, which in turn reduces to solving a special case of decremental \SSSP in directed graphs.

At a high level, we follow the approach of \cite{SCC}, who designed a $(1+\eps)$-approximation algorithm for directed decremental \SSSP, with total update time $O\left (\frac{n^{8/3+o(1)}}{\eps}\right )$ (assuming that all edge lengths are poly-bounded). This running time is prohibitively high for our purposes, so we cannot use their algorithm directly. Instead, we observe that we only need to solve a restricted special case of directed decremental \SSSP; by adapting their algorithm to this setting, we obtain the desired faster running time. Unlike the general \SSSP problem, the special case that we need to solve has some useful properties. First, all shortest-path queries are asked between a pre-specified pair $(s,t)$ of vertices. 
Additionally, we slightly modify the \MWU method, so that the task of responding to shortest-path queries is even easier. Specifically, we can assume that we are given two parameters $d$ and $\Delta$. In response to a shortest-path query, the algorithm needs to either return an $s$-$t$ path of length at most $d$, or return ``FAIL''. The latter can only be done if the largest number of internally disjoint $s$-$t$ paths of length at most $d/8$ in the current graph is bounded by $O(\Delta/\poly\log n)$. In other words, our algorithm only needs to support approximate shortest-path queries as long as there is a large number of short and internally disjoint $s$-$t$ paths in the current graph $H$. This property is crucially used by our algorithm to obtain a faster running time. Additionally, due to the specific setup provided by the \MWU method, the only edges that are deleted from $H$ are those appearing on the paths that the algorithm returns in response to queries. Moreover, unlike the standard setting of the dynamic \SSSP problem, where it is crucial that the time the algorithm takes to respond to each individual query is short, we only need to ensure that the total time the algorithm takes to respond to {\em all queries} is suitably bounded. We exploit all these advantages of the special case of the \SSSP problem that we obtain in order to design an algorithm for it that is significantly faster than the algorithm of \cite{SCC} for general decremental \SSSP.

We now provide additional details of our algorithm for restricted decremental \SSSP. We follow the high-level approach of \cite{SCC}, that consists of two parts. First, they maintain a partition $\xset$ of graph $H\setminus \set{s,t}$ into a collection of expander-like graphs; we abstract the problem of maintaining each such graph, that we call \maintaincluster problem, below. Second, they apply the Approximate Topological Order (\ATO) framework of \cite{AlmostDAG2} (which is in turn based on the works of \cite{gutenberg2020decremental} and \cite{Bernstein}), together with the algorithm of \cite{AlmostDAG2} for decremental \SSSP on ``almost'' DAG's to the graph $\hat H$, that is obtained from $H$ by contracting every almost-expander $X\in \xset$ into a single vertex. We now discuss each of these components in turn, and point out the similarities and the differences between our algorithm and that of \cite{SCC}.

We start by defining the \maintaincluster problem, that, intuitively, is used to manage a {\em single} expander-like graph. We will eventually apply it to every subgraph $X\in \xset$ of $H$. The input to the problem is a subgraph $X\subseteq H\setminus \set{s,t}$, and a parameters $d\geq \poly\log n$ and $\Delta>0$. We will treat all edges of $X$ as unweighted, so the length of each edge is $1$. The algorithm consists of at most $\Delta$ iterations. In each iteration, the algorithm is given a pair $x,y\in V(X)$ of vertices, and it needs to output a simple path $P$ of length at most $d$ connecting $x$ to $y$ in $X$. Some of the edges of $P$ may then be deleted from $H$, and the iteration ends. The algorithm can, at any time, compute a cut $(A,B)$ in $X$, whose sparsity is at most $O\left (\frac{\poly\log n}{d}\right )$. After that, if $|A|\ge |B|$, then vertices of $B$ are deleted from $X$, and otherwise the vertices of $A$ are deleted from $X$. In either case, the algorithm must continue with the new graph $X$. Our goal is to provide an algorithm for the \maintaincluster problem, whose total running time is small. %The running time includes both the total update time and the time required to respond to all queries. 
In particular, unlike a similar problem that was considered in \cite{SCC}, we do not need to ensure that the time that the algorithm takes in order to respond to \emph{each} query is low. Instead, we only need to ensure that the {\em total update time} of the algorithm, together with the {\em total time} spent on responding to queries, is low. The setting that is considered in \cite{SCC} is more general in several other ways: there is no bound on the number of queries, and edges may be deleted from $X$ as part of update sequence\footnote{Like \cite{SCC}, we distinguish between edges that are deleted from $X$ by ``adversary'', as part of the online update sequence, and edges that are deleted by the algorithm itself, when it computes a sparse cut $(A,B)$ in $X$, and deletes all vertices of the smaller side of the cut from $X$.} arbitrarily (in our setting, only edges that appear on the paths that the algorithm returns in response to queries may be deleted). The algorithm of \cite{SCC} for this more general version of the \maintaincluster problem has total update time $\ohat(|E(X)|\cdot d^2)$, and the time to respond to each query is roughly proportional to the number of edges on the path that the algorithm returns. We provide an algorithm for the \maintaincluster problem, whose total running time is bounded by $\otilde(|E(X)|\cdot d+|V(X)|^2)\cdot  \max\set{1,\frac{\Delta\cdot d^2}{|V(X)|}}$, which, for the setting of parameters that we employ becomes $\otilde(|E(X)|\cdot d+|V(X)|^2)$. Our algorithm for the \maintaincluster problem uses a rather standard approach: we embed a large $\Omega(1)$-expander $W$ into the graph $X$, and then use two Even-Shiloach tree (\EST) data structures in $X$, that are rooted at the vertices of $W$; in one of the resulting trees the edges are directed away to the root, and in the other they are directed towards the root. However, due to the specifics of the setting that we have described, we obtain an algorithm that is significantly simpler than that of \cite{SCC}, and other similar algorithms. For example, we use a straightforward (although somewhat slower) algorithm for the cut player in the Cut-Matching game, that is used in order to embed an expander into $X$, and our algorithm for the variant of \APSP in expanders that we obtain is quite elementary; it avoids both the recursion that was used in previous works (see \cite{APSP-old,detbalanced,SCC}), and the expander pruning technique of~\cite{expander-pruning,SCC}. While some parts of our algorithm can be optimized, for example, by using the faster algorithm for the Cut Player due to \cite{SCC}, this would not affect our final running time, and we believe that the simple self-contained proofs provided in this paper are of independent interest.

We now briefly describe the second component of our algorithm -- the use of the Approximate Topological Order (\ATO) framework of  \cite{AlmostDAG2}, that can be viewed as reducing the decremental \SSSP problem on general directed graphs to decremental \SSSP on ``almost DAG's''. This part is almost identical to the algorithms of \cite{AlmostDAG2,SCC}. The main difference is that we save significantly on the running time of the algorithm due to the fact that we only need to support approximate shortest-path queries between $s$ and $t$, as long as graph $H$ contains many short and disjoint $s$-$t$ paths. We start with the decremental graph $H$ that serves as the input to the restricted decremental \SSSP problem, and we maintain a partition $\xset$ of the vertices of $H\setminus\set{s,t}$ into subsets. At the beginning of the algorithm, we let $\xset$ contain a single set $V(H)\setminus\set{s,t}$, and, as the algorithm progresses, the only allowed updates to the collection $\xset$ of vertex subsets is a partition of a given set $X\in \xset$ into two subsets, $A$ and $B$. Intuitively, whenever a set $X$ of vertices is added to $\xset$, we view it as an instance of the \maintaincluster problem on the corresponding graph $H[X]$, with an appropriately chosen distance parameter $d_X$. We then apply the algorithm for the \maintaincluster problem to this instance. Whenever the algorithm produces a cut $(A,B)$ in $H[X]$, we let $Z$ be the smaller-cardinality set of vertices from among $A$, $B$.   Recall that the algorithm for the \maintaincluster problem then deletes the vertices of $Z$ from $X$. We add the set $Z$ as a new vertex set to $\xset$, and initialize the algorithm for the \maintaincluster problem on $H[Z]$. Additionally, we update the set $X$ of vertices in $\xset$, by deleting the vertices of $Z$ from it. For convenience, we will also add to $\xset$ two additional subsets: $S=\set{s}$, $T=\set{t}$, so that $\xset$ becomes a partition of the vertices of $H$ into subsets.

In order to define Approximate Topological Order (\ATO) on graph $H$, with the partition $\xset$ of the vertices of $V(H)\setminus\set{s,t}$ into subsets, we associate, with every set $X\in \xset$, a contiguous collection $I_X$ of $|X|$ integers from $\set{1,\ldots,n}$, that we refer to as an \emph{interval}. We always let $I_S=\set{1}$, $I_T=\set{n}$, and we require that all resulting intervals are mutually disjoint. Notice that the resulting collection $\set{I_X\mid X\in \xset}$ of intervals naturally defines an \emph{ordering} of the sets in $\xset$, that we denote by $\sigma$, and for convenience, we call it a \emph{left-to-right} ordering. Consider now some edge $e=(u,v)$, and let $X,Y\in \xset$ be the sets containing $u$ and $v$, respectively. Assume further that $X\neq Y$. If $X$ appears before $Y$ in the ordering $\sigma$, then we say that $e$ is a \emph{left-to-right} edge, and otherwise, we say that it is a \emph{right-to-left} edge.

Recall that, at the beginning of the algorithm, we let $\xset$ contain three vertex sets: $S=\set{s},T=\set{t}$, and $R=V(H)\setminus\set{s,t}$. We also set $I_S=\set{1}$, $I_T=\set{n}$, and $I_R=\set{2,\ldots,n-1}$. As the algorithm  progresses,
whenever a set $X$ of vertices is split into two sets $A$ and $B$, we partition the interval $I_X$ into two subintervals $I_A$ and $I_B$, of cardinalities $|A|$ and $|B|$, respectively. Recall that we only partition $X$ into sets $A$ and $B$ if $(A,B)$ defines a sparse cut in graph $H[X]$. This means that either $|E_H(A,B)|$, or $|E_H(B,A)|$ is  small relatively to $\min\set{|A|,|B|}$. In the former case, we ensure that interval $I_B$ appears before interval $I_A$, while in the latter case, we ensure that interval $I_A$ appears before interval $I_B$. This placement of intervals ensures that there are relatively few right-to-left edges.

Lastly, we consider the contracted graph $H'=H_{|\xset}$, that is obtained from $H$ by contracting every set $X\in \xset\setminus\set{S,T}$ into a single vertex. As the time progresses, graph $H'$ undergoes edge deletions and vertex-splitting updates. The latter type of updates reflects the partitioning of the sets $X\in \xset$ of vertices into smaller subsets. At all times, the ordering $\sigma$ of the intervals in $\set{I_X\mid X\in \xset}$ that is maintained by the \ATO naturally defines an ordering of the vertices of $H'$, that we denote by $\sigma'$. Notice that, if we could ensure that no right-to-left edges are present, then graph $H'$ would be a DAG, and $\sigma'$ would define a topological ordering of its vertices. We could then use the algorithm of \cite{AlmostDAG2} for decremental \SSSP in DAG's, whose total update time is $\otilde(n^2)$. In \cite{AlmostDAG2} it was further observed that, if we could ensure that the number of right-to-left edges with respect to the \ATO structure is low, and the average length of each such edge (with respect to the ordering $\sigma'$ of the vertices of $H'$) is small, one can apply an algorithm similar to that for \SSSP in DAG's to the contracted graph $H'$. The total update time of the resulting algorithm depends on the number of the right-to-left edges, and on their length with respect to $\sigma'$. Exactly this approach was used in \cite{SCC} in order to obtain an algorithm for decremental \SSSP in general directed graphs. We also employ the same approach, but we use the advantages of our setting, specifically, that all queries are between a pre-specified pair $s,t$ of vertices, and that we only need to support the queries as long as there are many short internally disjoint $s$-$t$ paths, in order to obtain a faster algorithm. Intuitively, since we can assume that the current graph $H$ contains $\tilde \Omega(\Delta)$ short internally disjoint $s$-$t$ paths, and since the right-to-left edges must be distributed among such paths, at least one short $s$-$t$ path must contain a very small number of right-to-left edges. It is this observation that leads to a faster running time of this part of our algorithm.

\subsection{Organization.}
We start with preliminaries in \Cref{sec: prelims}. We then provide a high-level overview of our algorithm in \Cref{sec: high level}, which includes the reduction of the \MBM problem to a special case of decremental \SSSP via a modification of the \MWU framework, and the statement of our main result for this special case of \SSSP. In \Cref{sec: ATO} we provide an overview of the Approximate Topological Order (\ATO) framework of \cite{AlmostDAG2}, and in  \Cref{sec: SSSP in almost DAGS} we discuss their algorithm for \SSSP in DAG-like graphs. In \Cref{sec: expander tools} we develop technical tools for maintaining a decomposition of the input graph $G$ into expander-like subgraphs, with the proof of the main result of this section -- an algorithm for the \maintaincluster problem -- deferred to \Cref{sec: expander algs}. Lastly, in \Cref{sec: SSSP alg} we combine all these tools to obtain the desired algorithm for the special case of decremental directed \SSSP that we consider, completing the algorithm for \MBM.

\section{Preliminaries}
\label{sec: prelims}

In this paper we work with both directed and undirected graphs.
By default graphs are allowed to contain parallel edges but they may not contain loops. Graphs with no parallel edges are explicitly referred to as \emph{simple} graphs.

We use standard graph theoretic notation. Given a graph $G$ (that may be directed or undirected), and two disjoint subsets $A,B$ of its vertices, we denote by $E_G(A)$ the set of all edges of $G$ with both endpoints in $A$, and by $E_G(A,B)$ the set of all edges connecting a vertex of $A$ to a vertex of $B$. 
The degree of a vertex $v$ in $G$, denoted by $\deg_G(v)$, is the number of edges incident to $v$ in $G$; if graph $G$ is directed, this includes both outgoing and incoming edges. Given a subset $A$ of vertices of $G$, its \emph{volume}, denoted by $\vol_G(A)$, is $\sum_{v\in A}\deg_G(v)$. When graph $G$ is clear from context, we may omit the subscript.

\paragraph{Congestion of Paths and Flows.}
Let $G$ be a graph with capacities $c(e)\geq 0$ on edges $e\in E(G)$, and let $\pset$ be a collection of simple paths in $G$. We say that the paths in $\pset$ cause \emph{edge-congestion} $\eta$, if every edge $e\in E(G)$ participates in at most $\eta\cdot c(e)$ paths in $\pset$. When edge capacities are not explicitly given, we assume that they are unit. 
If every edge of $G$ belongs to at most one path in $\pset$, then we say that the paths in $\pset$ are \emph{edge-disjoint}.

Similarly, if we are given a flow value $f(e)\geq 0$ for every edge $e\in E(G)$, we say that flow $f$ causes edge-congestion $\eta$, if for every edge $e\in E(G)$, $f(e)\leq \eta\cdot c(e)$ holds. If $f(e)\leq c(e)$ holds for every edge $e\in E(G)$, we may say that $f$ causes \emph{no edge-congestion}.

Assume now that we are given a graph $G$, a pair $s,t$ of its vertices, and a collection $\pset$ of simple paths connecting $s$ to $t$. We say that the paths in $\pset$ cause \emph{vertex-congestion} $\eta$, if every vertex $v\in V(G)\setminus\set{s,t}$ participates in at most $\eta$ paths. If the paths in $\pset$ cause vertex-congestion $1$, then we say that they are \emph{internally vertex-disjoint}.

\paragraph{Embedding of Graphs.} Let $H,G$ be two graphs with $V(H)\subseteq V(G)$. An \emph{embedding} of $H$ into $G$ is a collection $\pset=\set{P(e)\mid e\in E(H)}$ of paths in $G$, where for every edge $e=(u,v)\in E(H)$, path $P(e)$ connects $u$ to $v$ in $G$. The \emph{congestion} of the embedding is the maximum, over all edges $e'\in E(G)$, of the number of paths in $\pset$ containing $e'$; equivalently, it is the edge-congestion that the set $\pset$ of paths causes in $G$.

\paragraph{Cuts, Sparsity and Expanders for Directed Graphs.}
Let $G$ be a directed graph. A \emph{cut} in $G$ is a an ordered pair $(A,B)$ of subsets of vertices of $G$, such that $A\cap B=\emptyset$, $A\cup B=V(G)$, and $A,B\neq\emptyset$. Note that we do not require that $|A|\leq |B|$ in this definition. The \emph{sparsity} of the cut is:

\[ \Phi_G(A,B)=\frac{|E(A,B)|}{\min{\set{|A|,|B|}}}.\]

We say that a directed graph $G$ is a $\phi$-expander, for a given value $0<\phi<1$, if the sparsity of every cut in $G$ is at least $\phi$. Notice that, if $G$ is a $\phi$-expander, then for every cut $(A,B)$, both $|E(A,B)|\geq \phi\cdot \min\set{|A|,|B|}$ and $|E(B,A)|\geq \phi\cdot \min{\set{|A|,|B|}}$ hold.

We also use the following theorem that provides a fast algorithm for an explicit construction of an expander, and is based on the results of Margulis \cite{Margulis} and Gabber and Galil \cite{GabberG81}. The proof for undirected graphs was shown in \cite{detbalanced}.
\begin{theorem}[Theorem 2.4 in \cite{detbalanced}]
	%[Fast explicit expander construction]
	\label{thm:explicit expander}
	There is a constant $\alpha_0 > 0$ and a deterministic algorithm, that we call \constructexpander, that, given an integer $n>1$, in time $O(n)$ constructs a directed graph $H_n$ with $|V(H_n)|=n$, such that $H_n$ is an $\alpha_0$-expander, and every vertex in $H_n$ has at most $9$ incoming and at most $9$ outgoing edges.
\end{theorem}

The above theorem was proved in \cite{detbalanced} for undirected expanders, with maximum vertex degree in the resulting graph bounded by $9$, but it is easy to adapt it to directed expanders. This is since, if a graph $G$ is an undirected $\phi$-expander, and a directed graph $G'$ is obtained from $G$ by replacing every edge with a pair of bi-directed edges, then $G'$ is a directed $\phi$-expander.

The following simple observation follows from the standard ball-growing technique. We provide its proof for completeness in Section \ref{sec: appx ball growing} of Appendix.

\begin{observation}\label{obs: ball growing}
	There is a deterministic algorithm, that, given  a directed $n$-vertex graph $G=(V,E)$ with unit edge lengths and maximum vertex degree at most $\dmax$, and two vertices $x,y\in V$ with $\dist_G(x,y)\geq d$ for some parameter $d\geq 64\dmax\log n$, computes a cut $(A,B)$ of sparsity at most $\phi=\frac{32\dmax\log n}{d}$ in $G$. The running time of the algorithm is $O(\dmax\cdot\min\set{|A|,|B|})$.
\end{observation}

We will use the following simple observation a number of times. Intuitively, the observation allows us to convert a sequence of small sparse cuts into a single balanced sparse cut.

\begin{observation}\label{obs: from many sparse to one balanced}
	There is a deterministic algorithm, whose input consists of a directed $n$-vertex graph $G$, parameters $0<\alpha<1$ and $0<\phi<1$, and a sequence $\xset=(X_1,X_2,\ldots,X_k)$ of non-empty subsets of vertices of $G$, such that the following hold:
	\begin{itemize}
		\item Every vertex of $V(G)$ lies in exactly one set in $\set{X_1,\ldots,X_k}$; 
		\item $\max_i\set{|X_i|}= \alpha n$; and
		
		\item for all $1\leq i<k$, if we denote by $X'_i=\bigcup_{j=i+1}^kX_j$, then either $|E(X_i,X'_i)|\leq \phi\cdot |X_i|$, or $|E(X'_i,X_i)|\leq \phi\cdot |X_i|$.
	\end{itemize}
	The algorithm computes a cut $(A,B)$ in $G$ with $|A|,|B|\geq \min \set{\frac{1-\alpha}{2}\cdot n, \frac n 4}$, and $|E_G(A,B)|\leq \phi\cdot \sum_{i=1}^{k-1}|X_i|$. The running time of the algorithm is $O(|E(G)|)$.  
\end{observation}

\begin{proof}
	We will gradually construct an ordering $\oset$ of the sets in $\xset$, that will allow us to compute the cut. We compute the ordering in $k-1$ iterations, where for all $1\leq i< k$, in iteration $i$, we compute an ordering $\oset_i$ of the sets in $\set{X_1,\ldots,X_i,X'_i}$; notice that these sets of vertices partition $V(G)$. We will also gradually construct a set $E'$ of edges, will ensure that, at the end of every iteration $i$, if an edge $e\in E(G)\setminus E'$ connects a vertex in set $Z$ to a vertex in set $Z'$, with  $Z,Z'\in \set{X_1,\ldots,X_i,X'_i}$  and $Z\neq Z'$, then $Z$ appears before $Z'$ in the ordering $\oset_i$ (in other words, the edge is directed from left to right with respect to the ordering). We will also ensure that, at the end of each iteration $i$, $|E'|\leq \phi\cdot \sum_{j=1}^i|X_j|$.
	
	We start by constructing the first ordering $\oset_1$ of the sets $X_1,X'_1$. If  $|E(X_1,X'_1)|\leq \phi\cdot |X_1|$, then we place $X_1'$ before $X_1$ in the ordering, and we set $E'=E(X_1,X'_1)$. Otherwise, we place $X_1$ before $X_1'$ in the ordering (in this case, $|E(X'_1,X_1)|\leq \phi\cdot |X_1|$ must hold), and we set $E'=E(X_1',X_1)$. Notice that, after the first iteration, $|E'|\leq \phi\cdot |X_1|$ holds, and, if $e$ is an edge of $E(G)\setminus E'$ whose endpoints lie in different sets of $\set{X_1,X_1'}$, then $e$ is directed from left to right with respect to the ordering $\oset_1$.
	
	We now assume that the $i$th iteration successfully computed an ordering $\oset_i$ of sets $X_1,\ldots,X_i,X'_i$, together with a set $E'$ of edges, such that $|E'|\leq \phi\cdot \sum_{j=1}^i|X_j|$, and moreover, if an edge $e\in E(G)\setminus E'$ connects a vertex in set $Z$ to a vertex in set $Z'$, with  $Z,Z'\in \set{X_1,\ldots,X_i,X'_i}$  and $Z\neq Z'$, then $Z$ appears before $Z'$ in the ordering. We now consider the partition $(X_{i+1},X'_{i+1})$ of set $X'_{i}$. If $|E(X_{i+1},X'_{i+1})|\leq \phi\cdot |X_{i+1}|$,
	then we obtain ordering $\oset_{i+1}$ of $X_1,\ldots,X_{i+1},X'_{i+1}$ by replacing $X'_i$ in the ordering $\oset_i$ with the sets $X'_{i+1},X_{i+1}$ in this order. We then add the edges of $E(X_{i+1},X'_{i+1})$ to $E'$. Otherwise, $|E(X'_{i+1},X_{i+1})|\leq \phi\cdot |X_{i+1}|$ must hold, and we obtain ordering $\oset_{i+1}$ of $X_1,\ldots,X_{i+1},X'_{i+1}$ by replacing $X'_i$ in the ordering $\oset_i$ with the sets $X_{i+1},X'_{i+1}$ in this order. We then add the edges of $E(X'_{i+1},X_{i+1})$ to $E'$. Notice that, in either case, $|E'|\leq \phi\cdot \sum_{j=1}^{i+1}|X_j|$ now holds, and moreover, 
	if an edge $e\in E(G)\setminus E'$ connects a vertex in set $Z$ to a vertex in set $Z'$, with $Z\neq Z'$ and $Z,Z'\in \set{X_1,\ldots,X_i,X_{i+1},X'_{i+1}}$, then $Z$ appears before $Z'$ in the ordering $\oset_{i+1}$.

	We consider the ordering $\oset'$ of $\set{X_1,\ldots,X_k}$ obtained after the last iteration. For convenience, we denote the sets $X_1,\ldots,X_k$ by $Y_1,\ldots,Y_k$, so that their indices are consistent with the ordering $\oset'$. We are now guaranteed that $|E'|\leq \phi \cdot \sum_{i=1}^{k-1}|X_i|$ (since set $X_k$ did not contribute any new edges to $E'$), and that for every edge $e\in E(G)\setminus E'$, if $e$ connects a vertex of set $Y_j$ to a vertex of set $Y_{j'}$ with $j\neq j'$, then $j<j'$.

		Let $i$ be the index of largest-cardinality set $Y_i$. 
		Consider first the case where $\alpha\geq \frac{1}{4}$.
		Since $|Y_i|= \alpha\cdot n$, either $\left | \bigcup_{j=1}^{i-1}Y_j\right |\geq \frac{1-\alpha}{2}n$, or 
		$\left | \bigcup_{j=i+1}^{k}Y_j\right |\geq \frac{1-\alpha}{2}\cdot n$ must hold.
		
		In the former case, we define a cut $(A,B)$ by setting $B=\bigcup_{j=1}^{i-1}Y_j$ and $A=V(G)\setminus B$, while in the latter case, we define a cut $(A,B)$ by setting $A= \bigcup_{j=i+1}^{k}Y_j$ and $B=V(G)\setminus A$. Notice that in either case, one of the sides of the cut is guaranteed to contain at least $ \frac{1-\alpha}{2}\cdot n$ vertices, while the other side contains $Y_i$ and thus has cardinality at least $\alpha n\geq n/4$. Moreover, all edges in $E_G(A,B)$ must lie in $E'$, so $|E_G(A,B)|\leq \phi\cdot  \sum_{i'=1}^{k-1}|X_{i'}|$ must hold. We return the cut $(A,B)$.

	Assume now that $\alpha<\frac{1}{4}$. Then we compute the smallest index $a$, such that $\sum_{j=1}^a|Y_j|\geq n/4$. Let $B=\bigcup_{j=1}^aY_j$, and let $A=V(G)\setminus B$. Since for all $1\leq i\leq k$, $|Y_i|\leq n/4$, it must be the case that $|A|,|B|\geq n/4$. Moreover, all edges in $E_G(A,B)$ must lie in $E'$, so $|E_G(A,B)|\leq \phi\cdot  \sum_{i'=1}^{k-1}|X_{i'}|$ must hold. We return the cut $(A,B)$.
\end{proof}

\paragraph{Even-Shiloach Tree for Weighted Directed Graphs.}

We use the generalization of the Even-Shiloach data structure of \cite{EvenS,Dinitz} to weighted directed graphs, due to \cite{ES-tree-directed}, who extended similar results of \cite{HenzingerKing} for unweighted directed graphs.

\begin{theorem}[\cite{ES-tree-directed}, see Section 2.1]\label{thm: directed weighted ESTree}
	There is a deterministic algorithm, whose input consists of a directed graph $H$ with integral  lengths $\ell(e)\geq 1$ for edges $e\in E(H)$, that undergoes an online sequence of edge deletions, a source vertex $s\in V(H)$, and a distance parameter $d\geq 1$. The algorithm maintains a shortest-path tree $T\subseteq H$ rooted at $s$, up to depth $d$. In other words, for every vertex $v\in V(H)$ with $\dist_H(s,v)\leq d$, vertex $v$ lies in the tree $T$, and the length of the shortest $s$-$v$ path 
	in $T$ is equal to that in $H$. The data structure also maintains, for all $v\in V(T)$, the current value $\dist_H(s,v)$. The total update time of the algorithm is $O(m\cdot d)$, where $m$ is the number of edges in the initial graph $H$.
\end{theorem}

We will sometimes refer to the data structure from \Cref{thm: directed weighted ESTree} as $\EST(H,s,d)$.

\section{High-Level Overview of the Algorithm}
\label{sec: high level}

In this section we provide a high-level overview of our algorithm, define the main technical tool that we use, namely the restricted \SSSP problem, state our main result for this problem, and show that the proof of \Cref{thm:main} follows from it.

Recall that, in the \MBM problem, the input is a simple undirected $n$-vertex bipartite graph $G=(L,R,E)$, and the goal is to compute a matching of maximum cardinality in $G$. 
Let $\opt$ denote the value of the optimal solution to the problem.
It is enough to design an algorithm, that, given an integer $0<C^*\leq \opt$, computes a matching of cardinality  at least $C^*$ in $G$. This is since we can perform a binary search over the value $\opt$. Given the current guess $C^*$ on $\opt$, we run the algorithm with this target value $C^*$. If the algorithm successfully returns a matching of cardinality at least $C^*$, then we adjust our guess $C^*$ to a higher value, and otherwise we adjust it to a lower value. For convenience, we will denote this target value $C^*$ by $\opt$. Therefore, we assume from now on that we are given an integer $\opt$, such that the input graph $G$ contains a matching of value at least $\opt$, and our goal is to find a matching of value at least $\opt$.

It is well known that the \MBM problem can be reduced to computing maximum flow in a directed graph with unit edge capacities. In order to do so, we start with the graph $G$, and direct all its edges from vertices of $L$ towards vertices of $R$. We then add a source vertex $s$ that connects with an edge to every vertex of $L$, and a destination vertex $t$, to which every vertex of $R$ is connected. All edge capacities are set to $1$. Let $G'$ denote the resulting directed flow network. It is easy to verify that the value of the maximum $s$-$t$ flow in $G'$ is equal to the cardinality of maximum matching in $G$. Moreover, given an integral flow $f$ in $G'$, we can compute a matching $M$ of cardinality $\operatorname{val}(f)$ in $G$, in time $O(n)$: we simply include in $M$ all edges of $G$ that carry $1$ flow unit. 

As observed by \cite{GK98, Fleischer00,AAP93}, the Multiplicative Weight Update (\MWU) framework can be adapted in order to efficiently compute approximate solutions for various flow problems. Later, Madry \cite{madry2010faster} observed that an efficient implementation of this framework can be reduced to solving a special case of decremental \SSSP or \APSP. However, the use of this technique typically leads to algorithms that compute \emph{approximate} and not \emph{exact} maximum flow. Since our goal is to compute exact maximum matching, we cannot employ this technique directly. Instead, we will carefully combine the augmenting paths based approach for maximum flow with the \MWU framework.

Our algorithm maintains a matching $M$ in the input graph $G$, whose cardinality grows over time, starting with $M=\emptyset$. Given the current matching $M$,  we denote  $\Delta=\opt-|M|$. Once $\Delta$ becomes sufficiently small, we can use augmenting paths-based approach in a straightforward way, to compute a matching of cardinality $\opt$ in $G$ in time $O(\Delta\cdot m)$; we provide the details later. 
%In oder to do so, we perform $\Delta$ iterations. In every iteration, we compute the residual flow network corresponding to the current matching $M$ together with an augmenting path $P$ in that network, and then augment the matching via $P$. It is easy to verify that each iteration can be implemented in time $O(m)$.
As long as $\Delta$ remains sufficiently large, we  perform iterations. If $M$ is the matching that the algorithm maintains at the beginning of an iteration, and $M'$ is the matching obtained at the end of the iteration, then we will ensure that $|M'|\geq |M|+\Omega(\Delta/\poly\log n)$, where $\Delta=\opt-|M|$. This will ensure that the number of iterations is bounded by $O(\poly\log n)$. 

We now proceed to describe a single iteration. We assume that $M$ is the matching that is given at the beginning of the iteration, and let $\Delta=\opt-|M|$. Note that matching $M$ naturally defines an $s$-$t$ flow $f$ in $G'$, of value $|M|$: for every edge $e=(u,v)\in M$, we set the flow on edges $(s,u)$, $(u,v)$ and $(v,t)$ to $1$; the flow on all other edges is set to $0$. We then consider the residual flow network $G'_f$, that we denote for convenience by $H$. We will sometimes say that $H$ is the \emph{residual graph of $G$ associated with the matching $M$}. It is well known that the value of the maximum $s$-$t$ flow in a residual flow network $G'_f$ is equal to the value of the maximum flow in the original flow network $G'$ minus the value of the flow $f$. In other words, the value of the maximum $s$-$t$ flow in graph $H$ is at least $\Delta$. Equivalently, there is a collection $\pset$ of $\Delta$ edge-disjoint $s$-$t$ paths in graph $H$. We note that graph $H$ has a special structure that our algorithm will exploit extensively. Specifically, for every edge $e=(u,v)$ of $G$ with $u\in L$ and $v\in R$, if $e\in M$, then edges $(u,s),(v,u)$, and $(t,v)$ are present in $H$, and otherwise edge $(u,v)$ is present in $H$. 
Therefore, if $x$ is a vertex of $L$, then the in-degree of $x$ in $H$ is $1$: if $x$ is an endpoint of some edge $(x,y)$ in $M$, then the only edge entering $x$ in $H$ is $(y,x)$, and otherwise the only edge entering $x$ in $H$ is $(s,x)$. Similarly, if $y$ is a vertex of $R$, then its out-degree in $H$ is $1$. This ensures that, for any collection $\qset$ of edge-disjoint $s$-$t$ paths in $H$, the paths in $\qset$ are also internally vertex-disjoint. Our goal in the current iteration is to compute a collection $\pset$ of edge-disjoint $s$-$t$ paths in $H$, of cardinality at least $\Omega(\Delta/\poly\log n)$. From our discussion, in graph $H\setminus\set{s,t}$,
 the out-degree of every vertex in $R$ is at most $1$, and the in-degree of every vertex in $L$ is at most $1$; we use this fact later. 

In order to compute the desired collection of paths in $H$, we perform two steps. In the first step, we compute an initial collection $\pset'$ of $\Omega(\Delta/\poly\log n)$ paths in graph $H$ that cause an edge-congestion at most $O(\log n)$. This is done by employing the \MWU framework, and reducing the problem to a special case of decremental \SSSP in directed graphs, that we call restricted \SSSP.

In the second step, we construct a directed graph $H'\subseteq H$, that consists of all vertices and edges that participate in the paths of $\pset'$. Since the paths in $\pset'$ cause edge-congestion $O(\log n)$, due to the special structure of graph $H$, they also cause vertex-congestion $O(\log n)$. Therefore, $|E(H')|\leq O(n\log n)$. Additionally, by sending $\Theta(1/\log n)$ flow units on every path in $\pset'$, it is easy to verify that maximum $s$-$t$ flow in $H'$ has value at least $\Omega(\Delta/\poly\log n)$ even if no congestion is allowed (the capacity of every edge in $H'$ remains unit -- the same as its capacity in $H$). 
We then compute a maximum integral $s$-$t$ flow in $H'$ via standard Ford-Fulkerson algorithm, in time $\otilde(|E(H')|\cdot n)\leq \otilde(n^2)$; we provide the details below.
%We next {\em round} this feasible fractional $s$-$t$ flow into an integral $s$-$t$ flow of the same value using the fact that all edges in $H'$ are unit capacity. This rounding step can be implemented in 
%$O(|E(H')|\log n) = \tilde O(n)$ time by repeatedly peeling off integral flow paths using random walk from $s$ to $t$ (see Theorem 5 in~\cite{LRS13},\cite{GKK10}).
We thus recover a collection $\pset$ of $\Omega(\Delta/\poly\log n)$ edge-disjoint and vertex-disjoint paths in $H$. 
%To summarize, if we denote by $T(m,n,\Delta)$ the time required to perform the first step in an iteration with parameter $\Delta$, then the running time of the whole algorithm is $\tilde O(T(m,n,\Delta)+m\cdot \Delta+n^{1.5})$. 
We summarize our algorithm for Step 1 of each iteration in the following theorem. The majority of the remainder of this paper is dedicated to proving it.

\begin{theorem}\label{thm: one iteration main}
There is a deterministic algorithm, whose input consists of a simple bipartite $n$-vertex graph $G=(L,R,E)$, an integer $\opt>0$, such that $G$ contains a matching of cardinality at least $\opt$, and a matching $M$ in $G$ with $|M|<\opt$. Let $\Delta=\opt-|M|$, and let $H$ be the residual graph of $G$ with respect to $M$. Assume further that $\Delta\geq 256\log (|E|)$. The algorithm computes a collection $\pset$ of $\Omega(\Delta/
\poly\log n)$ $s$-$t$ paths in $H$, that cause edge-congestion at most $O(\log n)$, in time $\otilde\left(\frac{n^{2.5}}{\sqrt{\Delta}}\right ) $.
\end{theorem}

As discussed above, we now obtain the following immediate corollary, whose proof appears in Appendix~\ref{appsec:cor_one_iteration_main}.

\begin{corollary}\label{cor: one iteration main}
	There is a deterministic algorithm, whose input consists of a simple bipartite $n$-vertex graph $G=(L,R,E)$, an integer $\opt>0$, such that $G$ contains a matching of cardinality at least $\opt$, and a matching $M$ in $G$ with $|M|<\opt$. Let $\Delta=\opt-|M|$, and let $H$ be the residual graph of $G$ with respect to $M$. Assume further that $\Delta\geq 256\log (|E|)$. The algorithm computes a collection $\pset$ of $\Omega(\Delta/
	\poly\log n)$ $s$-$t$ paths in $H$, that are internally vertex-disjoint, in time $\otilde\left(\frac{n^{2.5}}{\sqrt{\Delta}}\right ) $.
\end{corollary}

Before we proceed to prove \Cref{thm: one iteration main}, we first show that the proof of \Cref{thm:main} follows from it.

\subsection{Completing the Proof of \Cref{thm:main}}
We fix a threshold $256\log (|E(G)|)<\Delta^*\leq n$ that we will define later. Recall that we can assume that we are given a value $\opt$, such that the input graph $G$ contains a matching of size at least $\opt$. Throughout the algorithm, we maintain a matching $M$, and we denote $\Delta=\opt-|M|$.  We start with $M=\emptyset$, and then perform iterations, as long as $\Delta\geq \Delta^*$. In every iteration, we apply the algorithm from 
\Cref{cor: one iteration main} to graph $G$ and the current matching $M$, to obtain a collection $\pset$ of $\Omega(\Delta/
\poly\log n)$ $s$-$t$ paths in the residual graph $H$ of $G$ corresponding to $M$, that are internally vertex-disjoint. We then augment the matching $M$ via these paths, and obtain a new matching $M'$ with $|M'|\geq |M|+\Omega\left(\frac{\Delta}{\poly\log n}\right )$. Therefore, $\opt-M'\leq \Delta- \Omega\left(\frac{\Delta}{\poly\log n}\right )$, and so the parameter $\Delta$ decreases by at least factor $\left(1-\Omega\left(\frac{1}{\poly\log n}\right )\right)$ in every iteration. We replace $M$ with $M'$, and continue to the next iteration.
Notice that the number of such iterations must be bounded by $O(\poly\log n)$, and the running time of every iteration is bounded by $\otilde\left(\frac{n^{2.5}}{\sqrt{\Delta}}\right ) \leq \otilde\left(\frac{n^{2.5}}{\sqrt{\Delta^*}}\right ) $. Once $\Delta\leq \Delta^*$ holds, we use a straightforward algorithm that iteratively computes an augmenting path in the current residual graph $H$, and then augments the matching $M$ using this path. The number of such iterations is bounded by $\Delta^*$, and each iteration can be executed in time $O(m)$. Overall, the running time of the algorithm is bounded by $O(\Delta^* m)+\otilde\left(\frac{n^{2.5}}{\sqrt{\Delta^*}}\right )$.
Setting $\Delta^*=\frac{n^{5/3}}{m^{2/3}}$ gives running time ${\otilde(m^{1/3}\cdot n^{5/3})}$.

In the remainder of this paper we focus on proving \Cref{thm: one iteration main}.
In order to prove \Cref{thm: one iteration main}, we use the \MWU framework. We start by describing the standard application of this framework, and later present our slight modification of this framework.

\subsection{The MWU Framework} 

Consider the residual flow graph $H$ of $G$ with respect to the current matching $M$. We 
consider the problem of computing a largest cardinality collection  $\pset$ of edge-disjoint paths connecting $s$ to $t$ in $H$. Let $\pset^*$ denote the collection of all simple $s$-$t$ paths in $H$. Then the problem can be expressed in a linear program, where every path $P\in \pset^*$ is associated with a variable $f(P)$; we also provide its dual LP.

\begin{tabular}[t]{|l|l|}\hline &\\
	$\begin{array}{lll}
	\text{\underline{Primal}}&&\\
	\text{Max}&\sum_{P\in \pset^*} f(P)&\\
	\text{s.t.}&&\\
	%&\sum_{P\in \pset^*(a)}f(P)\leq 1&\forall a\in A'\\
	%&\sum_{P\in \pset^*(b)}f(P)\leq 1&\forall b\in B'\\
	&\sum_{\stackrel{P\in\pset^*:}{e\in E(P)}}f(P)\leq
	1&\forall e\in E(H)\\
	&f(P)\geq 0&\forall P\in \pset^*\\
	\end{array}$
	&$
	\begin{array}{lll}
	\text{\underline{Dual}}&&\\
	\text{Min}&\sum_{e\in E(H)}\ell(e)\\
	\text{s.t.}&&\\
	&\sum_{e\in E(P)}\ell(e)\geq 1&\forall P\in \pset^*\\
	&\ell(e)\geq 0&\forall e\in E(H)\\
	&&\\
	\end{array}$\\ &\\ \hline
\end{tabular}

\vspace{1cm}

In the dual LP, there is a variable $\ell(e)$ for every edge $e\in E(H)$, that we will refer to as the \emph{length} of the edge. We now consider the standard application of the \MWU framework, that can be viewed as a primal-dual algorithm. The algorithm that we describe below is a simplified version of the original algorithm of \cite{GK98, Fleischer00}; this variant achieves a slightly weaker approximation, but its analysis is much simpler and the higher approximation factor does not affect the overall performance of our algorithm.

\paragraph{Standard \MWU-Based Algorithm.}
For simplicity, in this informal description of the standard \MWU-based algorithm, we assume that $m=|E(H)|$ is an integral power of $2$.
We start with the graph $H$, and set the length of every edge $e$ in $H$ to $\ell(e)=1/m$. We initialize $\pset=\emptyset$, and set $f(P)=0$ for all $P\in \pset^*$. Next, we perform iterations. In every iteration, we compute a simple $s$-$t$ path $P$ in $H$, whose length is bounded by $1$. We then add $P$ to $\pset$, set $f(P)=1$, double the length of every edge on $P$, and continue to the next iteration. The algorithm terminates once every $s$-$t$ path in $H$ has length greater than $1$. It is immediate to see that the final collection $\pset$ of paths is guaranteed to cause congestion at most $O(\log m)$. Indeed, whenever an edge $e$ is used by a path in $\pset$, we double its length, and once the length becomes greater than $1$, the edge may no longer be used in any paths that are subsequently added to $\pset$. Since initially we set $\ell(e)=1/m$, at most $O(\log m)$ paths in $\pset$ may use $e$. Moreover, if $\opt'$ is the cardinality of the largest collection of edge-disjoint $s$-$t$ paths in $H$, then at the end of the algorithm, $|\pset| = \Omega(\opt')$ holds.
This is since, at the beginning of the algorithm, the value of the dual solution is $1$, and the value of the primal solution is $0$. In every iteration, the value of the primal solution increases by $1$, and the value of the dual solution increases by at most $1$ (since the total length of all edges whose lengths are being doubled is at most $1$). Therefore, if we denote by $C$ the cost of the dual solution obtained at the end of the algorithm, then $C\leq 1+|\pset|\leq 2|\pset|$ holds. As the algorithm terminates only when the length of every $s$-$t$ path in $\pset^*$ is at least $1$, we obtain a valid solution to the dual LP of value $C$ at the end of the algorithm. By weak duality,  the value of the optimal soluton to the primal LP  (and hence $\opt'$) is upper-bounded by the value of the optimal solution to the dual LP, which is in turn upper-bounded by $C$. Therefore, $C\geq \opt'$, and, since $|\pset|\geq C/2$, we get that $|\pset|\geq \opt'/2$ holds. In our case, since $\opt'\geq \Delta$, we are guaranteed that $|\pset| = \Omega(\Delta)$. 

In order to implement this framework, it is sufficient to implement an oracle that, in every iteration, computes a simple $s$-$t$ path of length at most $1$ in $H$, or correctly establishes that no such path exists, while the lengths of the edges in $H$ increase over time. 
As observed by Madry \cite{madry2010faster}, the problem that the oracle needs to solve can be viewed as a special case of decremental \SSSP. Indeed,
the increase in edge lengths can be easily implemented as edge-deletion updates. This is done by constructing another graph $\hat H$, that is obtained from graph $H$, by replacing every edge $e=(u,v)\in E(H)$ with $\log m$ parallel copies of lengths $1/m,2/m,\ldots,1$. Every time the length of $e$ in graph $H$ doubles, we delete the shortest copy of $e$ that lies in $\hat H$. Implementing the oracle now reduces to solving the resulting instance of decremental \SSSP. Unfortunately, this remains a challenging problem, since {\bf exact} decremental \SSSP is known to be a difficult problem, and it seems unlikely that one can get a fast algorithm even for the restricted special case that we obtain.

\paragraph{Modified \MWU Framework.}

We start with some intuition for the modified \MWU framework. In order to overcome the difficulty mentioned above, we will slightly relax the requirements from the oracle: we still require that the paths that the oracle returns have lengths at most $1$, but now the oracle is allowed to terminate once $\dist_H(s,t)\geq \frac{1}{32\log m}$ holds (intuitively, this corresponds to solving the \SSSP problem approximately). Even with this relaxed setting, we are still guaranteed that, at the end of the algorithm, $|\pset|\geq \Omega\left(\frac{\Delta}{\log m}\right )$ holds, and as before, the paths in $\pset$ cause congestion at most $O(\log m)$: the bound on congestion follows from the same reasoning as before. In order to establish that $|\pset|\geq \Omega\left(\frac{\Delta}{\log m}\right )$, we scale the final length of every edge in the dual LP up by factor $32\log m$, thereby obtaining a feasible dual solution. This proves that $|\pset|$ is within factor $O(\log m)$ from the optimal solution to the primal LP.

We can relax the requirements from the oracle even more, and this relaxation will be crucial in reducing the running time of our algorithm. Specifically, we allow the algorithm to terminate, once the following condition is violated:

\begin{properties}{C}
	\item Let $\qset$ be the largest-cardinality collection of $s$-$t$ paths in the current graph $H$, such that the paths in $\qset$ are edge-disjoint and every path has length at most  $\frac 1 8$. Then $|\qset|\geq \frac{\Delta}{256\log^2m}$. \label{stopping condition}
	\end{properties} 

In addition to this relaxation, we will also slightly change the initial lengths that are assigned to every edge. We use a parameter $N=
\frac{64 m\log m}{\Delta}$.
In our algorithm, we will set the initial length $\ell(e)$ of every edge $e\in E(H)$ to $1/N$.
The resulting algorithm, denoted by \algmwu (excluding the implementation of the oracle), appears in Figure \ref{alg: mwu}.

\program{Alg-MWU}{alg: mwu}{
\begin{itemize}
	\item Set $\pset=\emptyset$, and, for each edge $e\in E(H)$, set $\ell(e)=1/N$. Then iterate.
	\item In every iteration:
	
	\begin{itemize}
		\item either compute a simple $s$-$t$ path $P$ in $H$ of length at most $1$; add $P$ to $\pset$; double the length of every edge $e\in E(P)$;
		\item or correctly establish that Condition \ref{stopping condition} no longer holds and halt.
	\end{itemize}
\end{itemize}
}

In the following claim we prove that Algorithm \algmwu is indeed guaranteed to compute a large enough collection of $s$-$t$ paths.

\begin{claim}\label{claim: appox MWU}
Let $\pset$ be the set of paths obtained at the end of Algorithm \algmwu, and assume that $\Delta\geq 128\log m$. Then $|\pset|\geq \frac{\Delta}{128\log m}$, and every edge of $H$ belongs to at most $O(\log m)$ paths of $\pset$.
\end{claim}

\begin{proof}
	Recall that we initially set the length of each edge $e$  of $H$ to be $\frac{1}{N}$. Since 	$N= \frac{64 m\log m}{\Delta}$, and $\Delta\geq 128\log m$, we get that the initial length of every edge $e$ of $H$ is  $\frac 1 N\geq  \frac 2 {m}$. For each such edge $e$, whenever a path that contains $e$ is added to $\pset$, we double the length of $e$. Since every path that is added to $\pset$ has length at most $1$, once $\ell(e)>1$ holds, the edge may no longer be used by paths that are newly added to $\pset$. Therefore, the length of $e$ may be doubled at most $\log m$ times, and at most $\log m$ paths in $\pset$ may use $e$.

	It now remains to prove that $|\pset|\geq \frac{\Delta}{128\log m}$.
	Consider the following thought experiment. We start with the set $\pset$ of paths, and the edge lengths that were obtained at the end of Algorithm \algmwu. Let $H'$ denote the resulting graph (which is identical to $H$ but may have different edge lengths).
	
	Next, we will construct another set $\pset'$ of paths, as follows. Initially, we let $\pset'=\emptyset$, and we let $H''$ be the copy of the graph $H'$. Then we iterate, as long as there is an $s$-$t$ path in $H''$ whose length is at most $\frac{1}{32\log m}$. In every iteration, we select an arbitrary simple $s$-$t$ path $P$ in $H''$ of length at most  $\frac{1}{32\log m}$, add it to $\pset'$, and then double the length of every edge $e\in E(P)$ in graph $H''$. This algorithm terminates when every $s$-$t$ path in $H''$ has length at least $\frac{1}{32\log m}$. 
		We show that, at the end of the above algorithm, $|\pset|+|\pset'|\geq \frac{\Delta}{64\log m}$ must hold.

	\begin{observation}\label{obs: end of alg many paths}
		At the end of the algorithm described above, $|\pset|+|\pset'|\geq \frac{\Delta}{64\log m}$.
	\end{observation}
\begin{proof}
In every iteration of the algorithm, we added one path $P$ to $\pset\cup \pset'$, thereby increasing $|\pset|+|\pset'|$ by $1$. At the same time, we doubled the lengths of all edges in $P$, thereby increasing the value of the dual solution by at most $1$ (since the total length of all edges on $P$ before the doubling step was at most $1$). 
	
	When the algorithm starts, every edge is assigned length $\frac{1}{N}$, so $\sum_{e\in E(H)}\ell(e)=\frac{m}{N}= \frac{\Delta}{64\log m}$, since $N=
	\frac{64 m\log m}{\Delta}$. Therefore, at the beginning of the algorithm, the value of the dual solution is at most $\frac{\Delta}{64\log m}$, and $|\pset|+|\pset'|=0$. After that, in every iteration $|\pset|+|\pset'|$ grows by $1$, while the value of the dual solution grows by at most $1$. So if we denote by $C$ the value of the dual solution at the end of the algorithm, then $C\leq |\pset|+|\pset'|+\frac{\Delta}{64\log m}$.
	
	Since at the end of the algorithm every $s$-$t$ path in $H''$ has length at least $\frac{1}{32\log m}$, by multiplying the lengths of the edges by $32\log m$, we obtain a feasible dual solution, whose value is $C'=C\cdot 32\log m\leq (32\log m)(|\pset|+|\pset'|)+\frac{\Delta}{2}$. 
	Since the  value of the optimal solution to the primal LP, denoted by $\opt'$, is at least $\Delta$, so is the value of the optimal dual solution.
	Therefore, $C'\geq \Delta$ must hold, and so  $|\pset|+|\pset'|\geq \frac{\Delta-\Delta/2}{32\log m}\geq\frac{\Delta}{64\log m}$.
\end{proof}	
	
	Assume now for contradiction that $|\pset|<\frac{\Delta}{128\log m}$, so that $|\pset'|\geq \frac{\Delta}{128\log m}$ holds. We will show that Condition \ref{stopping condition} must have held when Algorithm \algmwu was terminated, reaching a contradiction.
	
%	If we denote by $C$ the value of the dual solution at the end of this algorithm, then
%	We call the algorithm during which the collection $\pset$ of paths was constructed \emph{stage 1}, and we denote the graph $H$ obtained at the end of stage 1 by $H'$ (we assume that this includes edge lengths in $H'$). We then perform the second stage of the algorithm, where we continue adding paths of length at most $1$ to $\pset$ as before (while doubling the length of every edge on such a path), until the length of the shortest path in the resulting graph becomes at least $1/(32\log n)$. We denote by $\pset'$ the collection of paths that were added to $\pset$ during the first phase of the algorithm, and by $\pset''$ the collection of paths added during the second phase. Recall that we are guaranteed that $|\pset'|+|\pset''|\geq \frac{\Delta}{32\log n}$. Assume for contradiction that $|\pset|<\frac{\Delta}{64\log n}$. Then $|\pset''|>\frac{\Delta}{64\log n}$.
	
	Recall that each path in $\pset'$ has length at most $\frac 1{32\log m}$ in graph $H'$. Consider a new graph $H^*$, that is obtained by taking the union of all edges and vertices that appear on the paths in $\pset'$. The length of every edge in $H^*$ is set to be its length in $H'$. Denote $|\pset'|=k$. 
	Since every path in $\pset'$ has length at most $\frac 1 {32\log m}$ in $H'$, $\sum_{e\in E(H^*)}\ell(e)\leq \frac{k}{32\log m}$. 
	
	Using the same arguments that we applied to the set $\pset$ of paths, it is easy to see that the paths in $\pset'$ cause edge-congestion at most $\log m$ in $H''$ and hence in $H^*$. %Due to the special structure of graph $H$, they also cause vertex-congestion at most $\log m$. 
	From the integrality of flow, there is a collection $\qset$ of at least $\frac{k}{\log m}$ edge-disjoint $s$-$t$ paths in $H^*$. We say that a path $Q\in \qset$ is \emph{long}, if its length is greater than $\frac 1 8$. Since $\sum_{e\in E(H^*)}\ell(e)\leq \frac{k}{32\log m}$, the number of long paths in $\qset$ is bounded by $\frac{k}{4\log m}$, and the remaining paths must be short. But then we obtain a collection of at least $\frac{k}{2\log m}\geq \frac{\Delta}{256\log^2m}$ edge-disjoint $s$-$t$ paths of length at most $\frac 1 8$ each in $H'$, so when Algorithm \algmwu terminated, Condition \ref{stopping condition} still held, a contradiction.
\end{proof}

In order to complete the proof of \Cref{thm: one iteration main}, it is now sufficient to provide an efficient implementation of the oracle, that, in every iteration, either produces a simple $s$-$t$ path of length at most $1$ in the current graph $H$, or correctly establishes that Condition \ref{stopping condition} is violated. In order to implement the oracle, we will reduce it to a special case of the decremental \SSSP problem, that we call \emph{restricted \SSSP}, and then provide an algorithm for that problem. But first it will be convenient for us to slightly adjust the lengths of the edges in graph $H$.

\paragraph{Rescaling the Edge Lengths.}
Recall that we have set $N=\frac{64\cdot m\log m}{\Delta}$. Initially, the length of every edge in $H$ is set to be $1/N$, and after that, in every iteration, the lengths of some edges are doubled. When the length of an edge becomes greater than $1$, it may not be doubled anymore. Let $H'$ be the graph that is identical to $H$, except that we multiply all edge lengths by factor $N$. Then, at the beginning of the algorithm, the length of every edge in $H'$ is $1$, and, as the algorithm progresses, some edge lengths may be doubled. However, throughout the algorithm, all edge lengths remain at most $2N$. Throughout, we denote $\Lambda=N/8$. Once the length of an edge becomes greater than $N=8\Lambda$, it may no longer be doubled. Clearly, at all times $\tau$, if $P$ is an $s$-$t$ path in graph $H'$, whose length is $\rho$, then the length of $P$ in $H$ is exactly $\rho/N$. Similarly, if the length of $P$ in $H$ is $\rho'$, then the length of $P$ in $H'$ is $\rho'\cdot N=8\rho'\cdot \Lambda$. Consider now the following new stopping condition.

\begin{properties}[1]{C}
	\item Let $\qset'$ be the largest-cardinality collection of $s$-$t$ paths in the current graph $H'$, such that the paths in $\qset'$ are edge-disjoint, every path has length at most  $\Lambda$  in the current graph. Then $|\qset'|\geq \frac{\Delta}{256\log^2m}$ holds. \label{stopping condition2}
\end{properties} 

Clearly, Condition \ref{stopping condition2} holds in graph $H'$ if and only if Condition \ref{stopping condition} holds in graph $H$.

As mentioned already, we will reduce the problem of designing the oracle to a special case of the \SSSP problem. In order to avoid increases in edge lengths, we replace the initial graph $H'$ with a graph $\hat H$. Let $N'$ be the smallest integral power of $2$ that is greater than $N$. Initially, graph $\hat H$ is obtained from the initial graph $H'$ by replacing every edge $e=(u,v)$ with $\log N'$ parallel edges, of lengths $1,2,4,\ldots,N'$, that we view as \emph{copies} of $e$. Whenever the length of $e$ is doubled, we delete the shortest-length copy of $e$ from $\hat H$. Therefore, $\hat H$ is a dynamic graph that only undergoes edge-deletion updates.

In order to implement the oracle,
it is now sufficient to solve the following instance of the \SSSP problem: we are given a graph $\hat H$ with $n$ vertices and $O(m\log m)$ edges, and a parameter $\Delta$. Let $N=\frac{64 m\log m}{\Delta}$ and $\Lambda=N/8=\frac{8m\log m}{\Delta}$.  For every edge $e\in E(\hat H)$, we are given a length $\ell(e)$, that is an integral power of $2$ between $1$ and $16\Lambda$. 
The algorithm consists of at most $\Delta$ iterations. At the beginning of every iteration, we need to compute a simple $s$-$t$ path in the current graph $\hat H$ of length at most $8\Lambda$, or to correctly establish that the largest-cardinality collection $\qset$ of $s$-$t$ paths in the current graph $\hat H$, such that the paths in $\qset$ are edge-disjoint and every path has length at most  $\Lambda$, is less than $\frac{\Delta}{256\log^2m}$.
In the latter case, the algorithm terminates. In the former case, if $P$ is the $s$-$t$ path that the algorithm returned, we assume that for every edge $e\in E(P)$, it is the shortest among its parallel edges in $\hat H$. The edges of $E(P)$ are then deleted from $\hat H$ and we proceed to the next iteration.

To summarize, we reduced \Cref{thm: one iteration main} (and generally the \MBM problem) to a special case of (decremental) \SSSP problem in directed graphs. Observe that the algorithm for this special case must be able to withstand an adaptive adversary, since the edges that are deleted from the graph depend on the paths that the algorithm returned in response to prior queries. While current state of the art algorithms for \SSSP in directed graphs with an adaptive adversary are not so efficient, this special case of \SSSP has  some nice additional properties that we can exploit, as listed below:

\begin{itemize}
	\item The special structure of the input graph $\hat H\setminus\set{s,t}$: namely, for every vertex, either its in-degree or its out-degree is at most $1$, if we ignore parallel edges;
	
	\item The algorithm is only required to return a short $s$-$t$ path as long as there are many disjoint short $s$-$t$ paths, and once this is not the case, it can be terminated;
	
	\item We only need to respond to \shortestpath query between a specific pair of vertices $s$ and $t$ that is known in advance. Moreover, the only edges that are deleted from the graph are edges that appear on the paths that the algorithm returns. 
	
	\item The number of times we are required to respond to \shortestpath queries is not too large, and we only need to bound the {\em total time} required to both maintain the data structures and respond to queries. %This allows us to use a lazy approach, where, e.g. we are allowed to abort a response to a query if, as we attempt to respond to it, we discover that the data structure needs to be updated (as long as the total running time of the algorithm is suitably bounded.)
	
	\item Generally in the (decremental) \SSSP problem we require that the time required to return a path in response to a \shortpath query (called \emph{query time}) is proportional to the number of the edges on the path. In our case we could afford a higher query time.
	
	\item While much of the work on \SSSP focuses on factor $(1-\eps)$ approximation, we can sustain a much higher approximation factor (the reduction we presented requires approximation factor $8$, but it can be easily adjusted to allow higher approximation factors, that may be as large as polynomial in $n$).
\end{itemize}

The algorithm that we present does not exploit all of the advantages listed above; this may suggest avenues for obtaining even faster algorithms for \MBM via this general approach.

We abstract the special case of \SSSP problem that we need to solve, that we call \emph{restricted} \SSSP, in the following subsection.

\subsection{The Restricted \SSSP Problem}
\label{subsec: restricted SSSP}

\paragraph{Well-Structured Graphs.}
Throughout, we will work with graphs that have a special structure. Let $G$ be the given instance of the Bipartite Matching problem. All graphs that we consider are subgraphs of the residual graph of $G$ with respect to some matching $M$ (except that we may add some parallel edges). Such graphs have special structural properties that we summarize in the following definition.

\begin{definition}[Well-Structured Graph]\label{def: well-structured}
Let $H$ be a directed graph, and let $m$ be a parameter. We say that $H$ is a \emph{well-structured} graph of size $m$, if $|E(H)|\leq m\log m$, and  $H$ has two special vertices $s$ and $t$. Additionally, if we denote $H'=H\setminus\set{s,t}$, then the following hold.

\begin{itemize}	
	\item $V(H')$ is partitioned into two subsets $L$ and $R$, and every edge of $H'$ has one endpoint in $L$ and another in $R$. Edges directed from vertices of $R$ to vertices of $L$ are called \emph{special edges} (they correspond to the edges of the current matching $M$), and the remaining edges of $H'$ are called \emph{regular edges}.
	
	\item Every vertex of $L\cup R$ is incident to at most $\log m$ special edges in $H'$ and all of them are parallel to each other, so the out-degree of every vertex in $R$ (ignoring parallel edges) is at most $1$, and the in-degree of every vertex in $L$ (also ignoring parallel edges) is at most $1$ in $H'$.
	
	\item In graph $H$, vertex $s$ has no incoming edges, and all edges leaving $s$ connect it to vertices of $L$.
	
	\item Similarly, in graph $H$, vertex $t$ has no outgoing edges, and all edges entering $t$ connect vertices of $R$ to it.
	
	\item Every edge of $H'$ has at most $\log m-1$ edges that are parallel to it.
\end{itemize}
\end{definition}
 (We note that strictly speaking, a residual graph corresponding to a matching may contain edges entering $s$ and/or edges leaving $t$, but since such edges will never be used in any augmenting path, we can delete them from the residual network).

Whenever an algorithm is given a well-structured graph as an input, we assume that it is also given the vertices $s$ and $t$ and the corresponding partition $(L,R)$ of $V(H')$.
We are now ready to define the restricted $\SSSP$ problem.

\begin{definition}[Restricted \SSSP problem]
The input to the restricted \SSSP problem is a well-structured graph $G$ of size $m$, where $m$ is a given parameter, together with another parameter $0<\Delta\leq |V(G)|$. We denote $|V(G)|=n$ and  $\Lambda=\frac{8m\log m}{\Delta}$. For every edge $e\in E(G)$, we are also given a length $\ell(e)$, that is an integral power of $2$ between $1$ and $16\Lambda$.

The algorithm consists of at most $\Delta$ iterations. At the beginning of every iteration, the algorithm must compute a simple $s$-$t$ path $P$ in the current graph $G$ of length at most $8\Lambda$ (we sometimes informally say that the algorithm responds to a \shortpath query between $s$ and $t$). We assume that for each edge $e\in E(P)$, it is the shortest from among its parallel copies. Every edge of $E(P)$ is then deleted from $G$, and the iteration ends.

The algorithm is allowed, at any time, to halt and return FAIL, but it may only do so if the largest-cardinality collection $\qset$ of $s$-$t$ paths in the current graph $G$, such that the paths in $\qset$ are edge-disjoint, and every path has length at most  $\Lambda$, has  $|\qset|<\frac{\Delta}{256\log^2m}$. %In case the algorithm returns a path $P$, it should be done in time $\tilde O(|E(P)|)$, and this time is not counted towards the total update time of the algorithm. 
\end{definition}

% with length $\ell(e)>0$ on edges $e\in E(H)$ and a parameter $\Delta$. We denote by $n$ and $m$ the number of vertices and eddges, respectively in the initial graph $H$. The length $\ell(e)$ of every edge $e\in E(H)$ is an itegral power of $2$ between $1$ and $\frac{32n\log n}{\Delta}$. Graph $H$ undergoes an online sequence of edge deletions that arrive in batches. After each batch of edge deletions, we need to either return an $s$-$t$ path $P$ in the current graph, whose length is at most $1$; or to correctly certify that the following condition is violated in the current graph (the condition below is just a restatement of Condition \ref{stopping condition3}):

One of the main technical contributions of our paper is the proof of following theorem.

\begin{theorem}\label{thm: sssp main}
	There is a deterministic algorithm for the Restricted \SSSP problem, that, given a well-structured graph $G$ with $n$ vertices and parameters $\Delta$ and $m$, has running time $\otilde\left(\frac{n^{2.5}}{\sqrt{\Delta}}\right ) $.
	\end{theorem}

We note that the running time of the algorithm from \Cref{thm: sssp main} includes both the time that is needed in order to maintain its data structures (total {\em update time}), and the time it takes to compute and return $s$-$t$ paths (total {\em query time}).

The proof of  \Cref{thm: one iteration main} immediately follows from \Cref{thm: sssp main}: we use the algorithm from \Cref{thm: sssp main} in order to implement the oracle for
\algmwu. In order to do so, we first construct a graph $\hat H$ from graph $H$, as described above. In every iteration, we use the algorithm for the restricted \SSSP problem in order to compute a simple $s$-$t$ path $P$ in graph $\hat H$ of length at most $8\Lambda$, which immediately defines a simple path in graph $H$ of length at most $1$. The doubling of the length of every edge on $P$ is implemented by deleting the shortest copy of each such edge from graph $\hat H$.
When the algorithm from \Cref{thm: sssp main} terminates, we are guaranteed that Condition \ref{stopping condition2} is violated in $\hat H$, and so Condition \ref{stopping condition} is violated in $H$.
 It is easy to verify that the total running time of the algorithm is bounded by $O(m\log m)$ plus the running time of the algorithm for the restricted \SSSP problem, which is, in turn, bounded by $\otilde\left(\frac{n^{2.5}}{\sqrt{\Delta}}\right ) $. Since $\Delta\leq n$, we can bound the total running time of the algorithm by $\otilde\left(\frac{n^{2.5}}{\sqrt{\Delta}}\right ) $.

In the remainder of this paper we focus on the proof of \Cref{thm: sssp main}.
At a high level, the proof of the theorem follows the framework of \cite{SCC} for decremental \SSSP in directed graphs. But since we only need to deal with a restricted special case of the problem, we obtain a faster running time, and also significantly simplify some parts of the algorithm and analysis.
In the following two sections, we present the two main tools that are used in the proof of \Cref{thm: sssp main}: the \ATO framework, together with \SSSP in ``DAG-like'' graphs, and expander-based techniques. We then provide the proof of \Cref{thm: sssp main} in \Cref{sec: SSSP alg}. 

\section{Overview of the \ATO framework}
\label{sec: ATO}

We now provide an overview of the Approximate Topological Order (\ATO) framework 
that was developed in \cite{AlmostDAG2} and later used in  \cite{SCC}.

Let $G=(V,E)$ be a directed $n$-vertex graph  with two special vertices $s$ and $t$, that undergoes edge deletions during a time interval $\tset$. 
Throughout, we refer to a collection $I\subseteq \set{1,\ldots,n}$ of consecutive integers as an \emph{interval}. The \emph{size} of the interval $I$, denoted by $|I|$, is the number of distinct integers in it.

An \ATO consists of the following components:

\begin{itemize}
\item a partition $\xset$ of the set $V$ of vertices of $G$, such that there is a set $S=\set{s}$ and $T=\set{t}$ in $\xset$; and

\item a map $\rho$, that maps every set $X\in \xset$ to an interval $I_X\subseteq \set{1,\ldots,n}$ of size $|X|$, such that all intervals in $\set{I_X\mid X\in \xset}$ are mutually disjoint. 
\end{itemize}

For convenience, we will denote an \ATO by $(\xset,\rho)$. 
Both the partition $\xset$ and the map $\rho$ evolve over time. We denote by $\xset\attime$ the partition $\xset$ at time $\tau$, and by $\rho\attime$ the map $\rho$ at time $\tau$. For every set $X\in \xset\attime$, we denote by $I^{\attime}_X$ the interval $I_X$ associated with $X$ at time $\tau$.

We start initially with the collection $\xset$ of vertex subsets containing three sets, $S=\set{s}$, $T=\set{t}$, and $J=V\setminus\set{s,t}$. We set their corresponding intervals to $I_S=\set{1}$, $I_T=\set{n}$, and $I_J=\set{2,\ldots,n-1}$.
The only type of changes that are allowed to the \ATO over time is the {\bf splitting} of the sets $X\in \xset$. In order to split a set $X\in \xset$, we select some partition $\Pi$ of $X$, and replace the set $X$ in $\xset$ with the sets in $\Pi$. Additionally, for every set $X'\in \Pi$, we select an interval $I_{X'}\subseteq I_X$ of size $|X'|$, such that the resulting intervals in $\set{I_{X'}\mid X'\in \Pi}$ are mutually disjoint, and we set $\rho(X')=I_{X'}$ for all $X'\in \Pi$. 

Notice that an \ATO $(X,\rho)$ naturally defines a left-to-right ordering of the intervals in $\set{I_{X}\mid X\in \xset}$: we  say that interval $I_X$ lies to the left of interval $I_{X'}$, and denote $I_X \prec I_{X'}$ iff every integer in interval $I_X$ is smaller than every integer in $I_{X'}$. We will also sometimes denote $X\prec X'$ in such a case. Notice also that, if we are given an ordering $\sigma=(X_1,X_2,\ldots,X_k)$ of the sets in $\xset$, then this ordering naturally defines the map $\rho$, where for all $1\leq i\leq k$, the first integer in $I_{X_i}$ is $\sum_{j=1}^{i-1}|X_j|+1$, and the last integer is $\sum_{j=1}^i|X_j|$. For convenience, we will denote the map $\rho$ associated with the ordering $\sigma$ by $\rho(\sigma)$.

Given a vertex $v\in V$, we will sometimes denote by $X^v$ the set of $\xset$ that contains $v$. Consider now some edge $e=(u,v)$ in graph $G$, and assume that $X^u\neq X^v$. We say that $e$ is a \emph{left-to-right} edge if $X^u\prec X^v$, and we say that it is a \emph{right-to-left} edge otherwise.

Lastly, we denote by $H=G_{|\xset}$ the \emph{contracted graph} associated with $G$ and $\xset$ -- the graph that is obtained from $G$ by contracting every set $X\in \xset$ into a vertex $v_X$, that we may call a \emph{supernode}, to distinguish it from vertices of $G$. We keep parallel edges but discard loops.

\paragraph{Intuition.}
For intuition, we could let $\xset$ be the set of strongly connected components of the graph $G$. The corresponding contracted graph $H=G_{|\xset}$ is then a DAG, so we can define a topological order of its vertices. This order, in turn, defines an ordering $\sigma$ of the sets in $\xset$, and we can then let the corresponding \ATO be $(\xset,\rho(\sigma))$.
 Notice that in this setting, all edges of $G$ that connect different sets in $\xset$ are left-to-right edges. However, we will not follow this scheme. Instead, like in \cite{SCC}, we will ensure that every set $X\in \xset$ corresponds to some expander-like graph, and we will allow both left-to-right and right-to-left edges.  It would still be useful to ensure that the corresponding contracted graph $H$ is close to being a DAG, since there are efficient algorithms for decremental \SSSP on DAG's and DAG-like graphs. Therefore, we will try to limit both the number of right-to-left edges, and their ``span'', that we define below.

\paragraph{Skip and Span of Edges.}
For every edge $e=(u,v)\in E(G)$ we define two values: the \emph{skip} of edge $e$, denoted by $\skipp(e)$, and the \emph{span} of $e$, denoted by $\spann(e)$. Intuitively, the skip value of an edge $(u,v)$ will represent the ``distance'' between the intervals $I_{X^u}$ and $I_{X^v}$. Non-zero span values are only assigned to right-to-left edges. The span value also represents the distance between the corresponding intervals, but this distance is defined differently. We now formally define the span and skip values for each edge, and provide an intuition for them below.

Consider some edge $e=(x,y)\in E(G)$, and let $X,Y\in \xset$ be the sets containing $x$ and $y$, respectively. If $X=Y$, then $\skipp(e)=\spann(e)=0$. Assume now that $X\prec Y$, that is $e$ is a left-to-right edge. In this case, we set $\spann(e)=0$, and we let $\skipp(e)$ be the distance between the last endpoint of $I_X$ and the first endpoint of $I_Y$. In other words, if the largest integer in $I_X$ is $z_X$, and the smallest integer in $I_Y$ is $a_Y$, then $\skipp(e)=a_Y-z_X$. Note that, as the algorithm progresses, sets $X$ and $Y$ may be partitioned into smaller subsets. It is then possible that $\skipp(e)$ may grow over time, but it may never decrease.
%We also set $\spann(e)=0$.

Lastly, assume that $Y\prec X$, so $e$ is a right-to-left edge. In this case, we set $\skipp(e)$ to be the distance between the right endpoint of $I_Y$ and the left endpoint of $I_X$, and we let $\spann(e)$ be the distance between the left endpoint of $I_Y$ and the right endpoint of $I_X$. In other words, if $a_X$ and $z_X$ are the smallest and the largest integers in $X$ respectively, and similarly $a_Y$ and $z_Y$ are the smallest and the largest integers in $I_Y$ respectively, then $\skipp(e)=a_X-z_Y$ and $\spann(e)=z_X-a_Y$. Notice that, as the algorithm progresses, and sets $X$ and $Y$ are partitioned into smaller subsets, the value $\spann(e)$ may only decrease, while the value $\skipp(e)$ may only grow. Moreover, $\skipp(e)\leq \spann(e)$ always holds.

For intuition, consider again the setting where $\xset$ represents the  strongly connected components of $G$, and $\rho=\rho(\sigma)$, where $\sigma$ is an ordering of the sets in $\xset$ corresponding to  a topological ordering of the vertices of $H=G_{|\xset}$. Then for any path $P$ in $G$, $\sum_{e\in E(P)}\skipp(e)\leq n$ must hold. This is because the path must traverse the sets of $\xset$ in their left-to-right order. Consider now inserting a single right-to-left edge $e=(x,y)$ into $G$, so  $x\in X$, $y\in Y$, and $Y\prec X$. In this case, path $P$ will generally traverse the sets in $\xset$  from left to right, but it may use the edge $e$ in order to ``skip'' back. It is not hard to see that $\sum_{e'\in E(P)}\skipp(e')\leq n+\skipp(e)+\spann(e)$
now holds, since, by traversing the edge $e$, path $P$ is `set back' by at most $\spann(e)$ units. In other words, path $P$ traverses the sets in $\xset$ in their left-to-right order, except when it uses the edge $e$ to ``skip'' to the left. The length of the interval it ``skips'' over is bounded by $\spann(e)$, so this skip adds at most $\skipp(e)+\spann(e)\leq 2\spann(e)$ to $\sum_{e'\in E(P)}\skipp(e')$.

Using the same intuitive reasoning, if we allow graph $G$ to contain an arbitrary number of right-to-left edges, then for any simple path $P$, $\sum_{e\in E(P)}\skipp(e)\leq n+\sum_{e\in E(G)}2\spann(e)$ holds. However, if we are given any collection $\pset$ of edge-disjoint paths in $G$, then there must be some path $P\in \pset$ with $\sum_{e\in E(P)}\skipp(e)\leq n+\frac{\sum_{e\in E(G)}2\spann(e)}{|\pset|}$, since every edge $e\in E(G)$ may contribute the value $2\spann(e)$ to at most one path in $\pset$. We will use this fact extensively.

\paragraph{Overview of the Algorithm.}
Like the algorithm of \cite{SCC} for decremental \SSSP,
our algorithm relies on two main technical tools. The first tool allows us to maintain the partition $\xset$ of $V(G)$ over time. Intuitively, we will ensure that, for every set $X\in \xset$, the corresponding induced subgraph $G[X]$ of $G$ is in some sense expander-like. In particular, we will be able to support short-path queries between pairs of vertices of $X$ efficiently. When $G[X]$ is no longer an expander-like graph, we will compute a sparse cut $(A,B)$ in this graph, and we will replace $X$ with $A$ and $B$ in $\xset$. Some of the edges connecting the sets $A$ and $B$ may become right-to-left edges with respect to the current \ATO $(\xset,\rho)$ that we maintain. We will control the number of such edges and their $\spann$ values by ensuring that the cut is $\phi$-sparse, for an appropriately chosen parameter $\phi$. This part is related to extensive past work on embedding expanders, maintaining expander embeddings in decremental graphs, and decremental \APSP in expanders. However, due to the specifics of our setting, we obtain algorithms that are both significantly simpler than those used in past work, and lead to faster running times.

Given a dynamically evolving \ATO $(\xset,\rho)$ for graph $G$, we can now construct the corresponding contracted graph $H=G_{|\xset}$. 
For convenience, the supernodes of $H$ representing the sets $S=\set{s}$ and $T=\set{t}$ of $\xset$ are denoted by $s$ and $t$, respectively.
Notice that, as the time progresses, graph $H$ undergoes, in addition to edge-deletion, another type of updates, that we call \emph{vertex-splitting}: whenever some set $X\in \xset$ is split into a collection $\Pi$ of smaller subsets, we need to insert supernodes corresponding to the sets $Y\in \Pi$ into $H$. We will always identify one of the sets in $\Pi$ with $X$ itself (so the supernode $v_X$ corresponding to $X$ may continue to serve as the supernode representing this one new set). We will say that the new supernodes corresponding to the remaining sets in $\Pi$ were \emph{split off} from $v_X$. Consider now any edge $e=(x,y)\in E(G)$, and assume that $x\in X$ and $y\in Y$ for some $X,Y\in \xset$ with $X\neq Y$. The length of the edge $e$ in graph $H$ remains the same, namely, $\ell(e)$. But we will also set the \emph{weight} of the edge $e$ to be the smallest integral power of $2$ that is at least $\skipp(e)$. As observed already, as the algorithm progresses, the value $\skipp(e)$ may grow. In order to avoid inserting edges into $H$ overtime, we will create a number of parallel edges with different weight values corresponding to every edge of $G$ whose endpoints lie in different sets of $\xset$. Specifically, let $e=(x,y)$ be any  edge, with $x\in X$, $y\in Y$, and $X\neq  Y$. For every integer $0\leq i\leq \ceil{\log n}$, such that $2^i\geq \skipp(e)$, we add an edge $(v_X,v_Y)$ to $H$, with length $\ell(e)$ and weight $2^i$. It is easy to verify that this construction ensures that, for every vertex $v$ of $H$ and integer $0\leq i\leq \ceil{\log n}$, the number of vertices in the set $\set{u\in V(H)\mid (v,u)\in E(H)\mbox{ and } w(v,u)= 2^i}$ is bounded by $2^{i+2}$. Additionally, by carefully controlling the total span of all right-to-left edges in graph $G$, and using the fact that throughout the algorithm $G$ contains a large number of edge-disjoint $s$-$t$ paths, we will ensure that there is always a short $s$-$t$ path $P$ in $H$, such that $\sum_{e\in E(P)}w(e)$ is relatively small. These two properties are crucial for obtaining a fast algorithm for computing short $s$-$t$ paths in $H$, as it undergoes updates. The \ATO framework, combined with an algorithm for maintaining expander-like graphs, can now be viewed as reducing our problem to the problem of repeatedly computing short $s$-$t$ paths in graph $H$, as it undergoes the updates that we have outlined above. We abstract this problem, that we call \DLSSSP, in \Cref{sec: SSSP in almost DAGS}, and provide an algorithm for it, that is essentially identical to the algorithm of \cite{AlmostDAG2}, and which was also used in exactly the same fashion by \cite{SCC}.

Overall, our algorithm for restricted \SSSP departs from the algorithm of \cite{SCC} for decremental \SSSP in two ways. First, we exploit the fact that $H$ contains a large number of short $s$-$t$ paths, in order to obtain better bounds on the running time of the algorithm of \cite{AlmostDAG2}  for   \DLSSSP. Second, our more relaxed setting allow us to obtain significantly simpler and faster algorithms for maintaining expander-like graphs, than, for example, those presented in \cite{SCC}. In particular, we do not need to use expander pruning, the recursive algorithm for \APSP in expanders, and the recursive algorithm for the cut player in the cut-matching game. We provide simple implementation for the latter two algorithms.

We next describe the two main technical tools that we use: an algorithm for \DLSSSP and algorithms for maintaining expander-like graphs in the following two sections. We then complete the proof of \Cref{thm: sssp main}, by providing an efficient algorithm for the Restricted \SSSP problem, that exploits both these tools. 
\section{SSSP in DAG-like Graphs}
\label{sec: SSSP in almost DAGS}

In this section we define a problem that can be thought of as an abstraction for the types of \SSSP instances that we obtain via \ATO. We call this problem \DLSSSP, and we provide an algorithm for this problem, that is essentially identical to the algorithm of \cite{AlmostDAG2}.

The input to the \DLSSSP problem is a directed graph $G=(V,E)$, with two designated vertices $s$ and $t$, a distance parameter $d>0$, a precision parameter $0<\eps\leq 1/8$, such that $1/\eps$ is an integer, and another parameter $\Gamma>0$. 
The graph $G$ may have parallel edges, but no self-loops. Every edge $e\in E$ of $G$ is associated with an integral \emph{length} $\ell(e)>0$, and an additional \emph{weight} $w(e)$, that must be an integral power of $2$ between $1$ and $2^{\ceil{\log n}}$. The graph undergoes the following two types of updates:

\begin{itemize}
	\item {\bf Edge deletion:} given an edge $e\in E$, delete $e$ from $G$; and
	
	\item {\bf Vertex splitting:} in this operation, we are given as input a vertex $v\not\in\set{s,t}$ in the current graph, and new vertices $u_1,\ldots,u_k$ that need to be inserted into $G$; we may sometimes say that vertices $u_1,\ldots,u_k$ are \emph{split-off from $v$}. In addition to vertex $v$ and vertices $u_1,\ldots,u_k$, we are given a new collection $E'$ of edges that must be inserted into $G$, and for every edge $e\in E'$, we are given an integral length $\ell(e)>0$ and a weight $w(e)$ that must be an integral power of $2$ between $1$ and $2^{\ceil{\log n}}$. For each such edge $e\in E'$, either both endpoints of $e$ lie in $\set{u_1,\ldots,u_k, v}$; or one endpoint lies in $\set{u_1,\ldots,u_k}$, and the other in $V(G)\setminus\set{v,u_1,\ldots,u_k}$. In the latter case, if $e=(u_i,x)$, then edge $e'=(v,x)$ must currently lie in $G$, and both $w(e')=w(e)$ and $\ell(e')=\ell(e)$ must hold. Similarly, if $e=(x,u_i)$, then edge $e'=(x,v)$ must currently lie in $G$, and both $w(e')=w(e)$ and $\ell(e')=\ell(e)$ must hold.
\end{itemize}

We let $\tset$ be the \emph{time horizon} of the algorithm -- that is, the time interval during which the graph $G$ undergoes all updates. 
We denote by $n$ and by $m$ the total number of vertices and edges, respectively, that were ever present in $G$ during $\tset$. We assume w.l.o.g. that for every vertex $v$ that was ever present in $G$, some edge was incident to $v$ at some time during the update sequence, so $n\leq O(m)$.

For every vertex $v\in V(G)$ and integer $0\leq i\leq \ceil{\log n}$, we define a set:
 $$U_i(v)=\set{u\in V(G)\mid \exists e=(v,u)\in E(G)\mbox{ with } w(e)=2^i}.$$ 
 Note that as the algorithm progresses, each such set $U_i(v)$ may change, with vertices both leaving and joining the set. We require that
the following property must hold at all times $\tau\in \tset$:

\begin{properties}{P}
	\item For every vertex $v\in V(G)\setminus\set{t}$, for every integer $0\leq i\leq \ceil{\log n}$, $|U_i(v)|\leq 2^{i+2}$.  \label{prop: few close neighbors}	
\end{properties}

As graph $G$ undergoes edge-deletion and vertex-splitting updates, at any time the algorithm may be asked a query $\pquery$. It is then required to return a {\bf simple} path $P$ connecting $s$ to $t$, whose length $\sum_{e\in E(P)}\ell(e)$ is bounded by $(1+10\eps)d$, in time $O(|E(P)|)$.

Lastly, the algorithm may return FAIL at any time and halt. It may only do so if the current graph $G$ does not contain any $s$-$t$ path $P$  with $\sum_{e\in E(P)}\ell(e)\leq d$, and $\sum_{e\in E(P)}w(e)\leq \Gamma$.

The following theorem is one of the main technical tools that we use; it follows almost immediately from the results of \cite{AlmostDAG2}. Specifically, Theorem 5.1 of \cite{AlmostDAG2} provides a similar result with a slightly different setting of parameters, but their arguments can be easily extended to prove the theorem below. We provide the proof of the theorem for completeness in Section \ref{subsec: proof of DAG-like SSSP} of Appendix.

\begin{theorem}\label{thm: almost DAG routing}
	There is a deterministic algorithm for the \DLSSSP problem with total update time $O\left(\frac{n^2+m+\Gamma\cdot n}{\eps^2}\right)$.
\end{theorem}

\section{Tools for Expander-Like Graphs}
\label{sec: expander tools}

\paragraph{Well-Structured Cuts.}
Suppose we are given a {\bf simple} unweighted directed bipartite graph $H=(L,R,E)$, such that every vertex in $R$ has out-degree at most $1$ and every vertex in $L$ has in-degree at most $1$. We call the edges that are directed from $R$ to $L$ \emph{special edges}, and the remaining edges are \emph{regular} edges.

 We say that a cut $(A,B)$ in $H$ is \emph{well-structured} iff all edges of $E_H(A,B)$ are special edges.

\begin{observation}\label{obs: sparse cut to structured}
	Suppose we are given a
	simple directed   bipartite graph $H=(L,R,E)$ with no isolated vertices, such that every vertex in $R$ has out-degree at most $1$ and every vertex in $L$ has in-degree at most $1$, together with and a cut $(A,B)$ in $H$ of sparsity at most $\phi\leq \frac{1}{4}$. Then there is a well-structured cut $(A',B')$ in $H$ of sparsity at most $2\phi$, such that $A'\subseteq A$,  $|A'|\geq (1-\phi)|A|$, and $\vol_H(A')\geq \Omega(\vol_H(A))$ hold. 
	 
	Moreover, there is an algorithm, that, given cut $(A,B)$, computes such a cut  $(A',B')$, in time $O(\min\set{\vol_H(A),\vol_H(B)})$.
\end{observation}

\begin{proof}
	We start with cut $(A,B)$ and inspect every vertex in $A$. If $v$ is any such vertex, and there is an edge $(v,u)$ in $E_H(A,B)$ that is a regular edge, we move $v$ from $A$ to $B$. Note that the only new edge that this move introduces into the cut $(A,B)$ is a special edge that enters $v$, if such an edge exists. Vertex $v$ is removed from $A$, and we say that edge $(v,u)$ is \emph{responsible} for this move. Notice that every regular edge of $E_H(A,B)$ is responsible for moving at most $1$ vertex from $A$ to $B$.
	
	Let $(A',B')$ denote the final cut obtained at the end of this procedure. For convenience, denote $z=\min\set{|A|,|B|}$. It is easy to verify that $|B'|\geq z$ continues to hold, and $|A'|\geq |A|-|E_H(A,B)|\geq |A|-\phi\cdot z\geq \frac{|A|}{2}$. Therefore, $\min\set{|A'|,|B'|}\geq \frac{\min\set{|A|,|B|}}{2}$.
	From the above discussion, $|A'|\geq |A|-\phi\cdot z\geq (1-\phi)|A|$ also holds.
	
	 Notice also that $|E_H(A',B')|\leq |E_H(A,B)|$, since every regular edge $e\in E(A,B)$ may cause at most one additional special edge to be added to the cut. Altogether, we get that:
	
	\[\frac{|E_H(A',B'|)}{\min\set{|A'|,|B'|}}\leq \frac{2|E_H(A,B)|}{ \min\set{|A|,|B|}}\leq 2\phi.\]

	Notice that for every edge $e\in E_H(A)$, at least one endpoint of $e$ must remain in $A'$. Indeed, consider any such edge $e=(x,y)$. The only way that $x$ is moved from $A$ to $B$ is if there is a regular edge $(x,y')\in E_H(A,B)$, which, in turn may only happen if $x\in L$. Using the same reasoning, $y$ may only be moved to $B$ if $y\in L$. But since $x$ and $y$ are connected by an edge, it is impossible that both vertices lie in $L$. 
	Let $A_1\subseteq A\setminus A'$ be the set of vertices $v\in A$, such that $v$ was moved to $B'$ and $v$ had no neighbors in $A$, and let $A_2=A\setminus A_1$. Then $|A_1|\leq |E_H(A,B)|\leq \phi\cdot |A|\leq \frac{|A|}{4}$, so in particular, $|A_1|\leq O(|A_2|)$.
	Since $|E_H(A,B)|<|A|$, we get that:
	
	\[ \vol_H(A')\geq  \Omega(|E_H(A)|+|A_2|) \geq  \Omega(|E_H(A)|+|A|)\geq  \Omega(|E_H(A)|+|E_H(A,B)|)\geq 
	\Omega(\vol_H(A)).\]
	
	(we have used the fact that $H$ has no isolated vertices).

	The running time of the algorithm is bounded by $O(\vol_H(A))$. If $\vol_H(B)<\vol_H(A)$, then we can instead scan all vertices of $B$, and for each such vertex $u$ try to identify a regular edge $(v,u)\in E(A,B)$. In this case, the running time is bounded by $O(\vol_H(B))$. Notice that we can establish whether $\vol_H(A)\leq \vol_H(B)$ in time $O\left(\min\set{\vol_H(A),\vol_H(B)}\right )$, and so the total running time of the algorithm is bounded by $O\left(\min\set{\vol_H(A),\vol_H(B)}\right)$.
\end{proof}

\subsection{\maintaincluster Problem}

We now define a new problem, called \maintaincluster problem. Intuitively, the goal in this problem is to maintain an expander-like subgraph of a given input graph, and to support short-path queries between pairs of its vertices.

\paragraph{Definition of \maintaincluster Problem.}
The input to the problem is a 
simple  directed bipartite graph $H=(L,R,E)$, such that every vertex in $R$ has out-degree at most $1$ and every vertex in $L$ has in-degree at most $1$, together with integral parameters $d^*>0$ and $\Delta>0$. 
%We are also given a parameter $m^*$ (the number of edges in the original graph $G$), with $|E|\leq m^*$. 
All edge lengths in $H$ are unit. 
%The edges of $E(L,R)$ are called \emph{regular}, and the edges of $E(R,L)$ are called \emph{special}. A cut $(X,Y)$ in $H$ is \emph{well-structured} iff all edges of $E(X,Y)$ are special. 
The algorithm consists of at most $\Delta$ iterations. At the beginning of every iteration, the algorithm is given some pair $x,y\in V(H)$ of vertices, and it must return a {\bf simple} path $P$ connecting $x$ to $y$ in $H$, whose length is at most $d^*$. After that, some edges of $P$ may be deleted from $H$. %If some edge of $H$ participated in $\ceil{\log m^*}$  of the paths that the algorithm returned, then such an edge is always deleted from $H$. 
At any time, the algorithm may provide a well-structured cut $(X,Y)$ in $H$, whose sparsity is at most $O\left(\frac{\log^6n}{d^*} \right )$. If $|X|\leq |Y|$, then the cut must be provided by listing the vertices of $X$, and otherwise it must be provided by listing the vertices of $Y$. After that the vertices of the smaller side of the cut are deleted from $H$. In other words, if $|X|\leq |Y|$, then we continue with the graph $H[Y]$, and otherwise we continue with the graph $H[X]$. Once the number of vertices in $H$ falls below half the original number, the algorithm halts, even if fewer than $\Delta$ iterations occurred.

We emphasize that the only updates that graph $H$ undergoes are the deletion of some edges on the paths that the algorithm returns in response to queries; we think of these updates as performed by the adversary. Additionally, whenever the algorithm produces a sparse well-structured cut $(X,Y)$, it deletes the vertices of the smaller of the two sets from $H$. We think of this type of update as being performed by the algorithm itself. These are the only updates that graph $H$ may undergo.

Below is our main theorem for the \maintaincluster problem; its proof is deferred to \Cref{sec: expander algs}.

\begin{theorem}\label{thm: main for expander tools}
	There is an algorithm for the \maintaincluster problem, that, given an $n$-vertex graph $H$ as in the problem definition and parameters $c^*\log^6n\leq d^*\leq  n$, and $1\leq \Delta\leq n$, where $c^*$ is a large enough constant, has running time at most: \[\otilde(m\cdot d^*+n^2)\cdot  \max\set{1,\frac{\Delta\cdot (d^*)^2}{n}}, \]
	
	where $m$ is the number of edges in $H$ at the beginning of the algorithm.
\end{theorem}

We note that the running time of the algorithm in the above theorem includes both the total update time of the data structure and the time needed to respond to queries.
We also note that some of the results from \cite{SCC} can be viewed as providing an algorithm for the \maintaincluster problem. In particular, by combining their Theorem 4.3 with Theorem 4.4 and Theorem 4.5, one can obtain such an algorithm (except for minor technical differences, e.g. the algorithm produces sparse vertex cuts and not edge cuts). The total update time of their algorithm is roughly $\ohat(|E(H)|\cdot (d^*)^2)$. Our algorithm achieves both a better running time for the specific setting that we consider, where $\Delta\cdot (d^*)^2\leq \otilde(n)$, and is also significantly simpler.

\section{Algorithm for Restricted \SSSP: Proof of \Cref{thm: sssp main}}
\label{sec: SSSP alg}

In this section we provide an algorithm for the restricted \SSSP problem, proving \Cref{thm: sssp main}. Recall that we are given as input a well-structured graph $G$ of size $m$ (see \Cref{def: well-structured}), where $m$ is a given parameter, together with another parameter $\Delta>0$. We denote $|V(G)|=n$ and  $\Lambda=\frac{8m\log m}{\Delta}$. For every edge $e\in E(G)$, we are also given a length $\ell(e)$, that is an integral power of $2$ between $1$ and $16\Lambda$.

The algorithm consists of at most $\Delta$ iterations. At the beginning of every iteration, the algorithm is required to return a simple $s$-$t$ path $P$ in the current graph $G$ of length at most $8\Lambda$. We assume that for each edge $e\in E(P)$, edge $e$ is the shortest among all its parallel copies. Every edge on $P$ is then deleted from $G$, and the iteration ends.

The algorithm is allowed, at any time, to halt and return FAIL, but it may only do so if the  following condition is violated:

\begin{properties}[2]{C}
	\item Let $\qset$ be the largest-cardinality collection of $s$-$t$ paths in the current graph $G$, such that the paths in $\qset$ are edge-disjoint, every path has length at most  $\Lambda$  in the current graph. Then ${|\qset|\geq \frac{\Delta}{256\log^2m}}$ holds. \label{stopping conditionfinal}
\end{properties}

We use a parameter ${d=\sqrt{n/\Delta}}$.
We say that an edge $e\in E(G)$ is \emph{long} if ${\ell(e)\geq \frac{\Lambda}{64d\log m}}$, and otherwise we say that it is {\em short}.
We denote the set of all long edges by $\Elong$.

Our algorithm consists of two parts. The first part is responsible for maintaining an \ATO $(\xset,\rho)$ of graph $G$, and the second part essentially executes the algorithm from \Cref{thm: almost DAG routing} on the corresponding contracted graph $G_{|\xset}$. We now describe each of the two parts in turn.

\subsection{Part 1: Maintaining  the \ATO.} 
In this part we will work with a graph $\hat G$ that is obtained from $G\setminus \Elong$, by setting the length of every edge to $1$, and removing parallel edges. Notice that $\hat G$ is a dynamic graph that undergoes online edge deletions. We will maintain an \ATO  $(\xset,\rho)$ for this graph (that will naturally also define an \ATO for $G$).

We start with a partition $\xset$ of $V(\hat G)$, containing three sets: $S=\set{s},T=\set{t}$, and $J=V(\hat G)\setminus \set{s,t}$.
Their corresponding intervals are initialized to $I_S=\set{1}$, $I_T=\set{n}$, and $I_J=\set{2,\ldots,n-1}$. For convenience, we call the sets of vertices in $\xset$ \emph{clusters}.

 Whenever a new cluster $X$ joins $\xset$ (excluding the sets $S$ and $T$), we set ${d_X=\frac{|X|}{n}\cdot d}$. If $d_X\geq c^*\log^6|X|$, where $c^*$ is the constant from \Cref{thm: main for expander tools}, then we say that $X$ is a \emph{non-leaf cluster}. We then immediately initialize the algorithm for the \maintaincluster problem from \Cref{thm: main for expander tools} on the corresponding graph $\hat G[X]$, with distance parameter $d^*=d_X$, and parameter $\Delta$ remaining unchanged.
 Notice that, since $d=\sqrt{n/\Delta}$, we get that $d_X=\frac{|X|}{n}\cdot d=\frac{|X|}{\sqrt{n\cdot \Delta}}\leq |X|$, and, since $d_X=\frac{|X|}{\sqrt{n\cdot \Delta}}\geq c^*\log^6|X|$, we get that $|X|\geq \sqrt{n\cdot \Delta}$, and so $\Delta\leq \frac{|X|^2}{n}\leq |X|$, as required from the definition of the \maintaincluster problem and the conditions of \Cref{thm: main for expander tools}.

 If $d_X<c^*\log^6|X|$, then we say that $X$ is a \emph{leaf} cluster. In this case, we replace $X$ in set $\xset$ with $|X|$ new clusters, each of which consists of a single distinct vertex of $X$. For each such new cluster $\set{v}$, we assign an interval $I_{\set{v}}\subseteq I_X$, that consists of a single integer arbitrarily, but we ensure that the intervals assigned to vertices of $L$ appear before intervals assigned to vertices of $R$. This ensures that all regular edges of $\hat G[X]$ are left-to-right edges, and the special edges are right-to-left edges. We call all special edges of $\hat G[X]$ \emph{bad edges}, and we say that $X$ \emph{owns} these edges. Clearly, $X$ owns at most $|X|$ bad edges, and each such bad edge $e$ has $\spann(e)\leq |X|$. Observe that, since $d_X=\frac{|X|}{n}\cdot d$ and $d_X< c^*\log^6|X|\leq c^*\log^6n$, it must be the case that $|X|\leq \frac{c^*n\log^6n}{d}$.

  For convenience, for every cluster $X$, we denote by $X_0$ the set $X$ of vertices at the time when it first joined $\xset$, and by $m_0(X_0)$ the number of edges in graph $\hat G[X_0]$ at the time when $X_0$ joined $\xset$. We will also sometimes denote $d_{X}$ by $d_{X_0}$.
  
 Consider now some non-leaf cluster $X$. 
Recall that the algorithm from \Cref{thm: main for expander tools}, when applied to $X$,  has running time at most: 

\[\otilde(m_0(X_0)\cdot d_{X_0}+|X_0|^2)\cdot  \max\set{1,\frac{\Delta\cdot (d_{X_0})^2}{|X_0|}}. \]

Since $d_{X_0}=\frac{|X_0|}{n}\cdot d$ and $d=\sqrt{\frac{ n}{\Delta}}$, we get that  $\frac{\Delta\cdot d_{X_0}^2}{|X_0|}=\frac{\Delta \cdot |X_0| d^2}{n^2}=\frac{|X_0|}{n}\leq 1$. Therefore, the running time of the algorithm is bounded by:

\[\otilde(m_0(X_0)\cdot d_{X_0}+|X_0|^2).\]
 
 Once the number of vertices in $X$ falls below $|X_0|/2$, the algorithm from  \Cref{thm: main for expander tools} terminates, and we update the distance parameter $d_X=\frac{|X|}{n}\cdot d$. If $d_X\geq c^*\log^6|X|$, then we apply the algorithm from \Cref{thm: main for expander tools} to this new cluster $X$ again, with the updated parameter $d_X$; otherwise cluster $X$ becomes a leaf cluster and is treated as such. In either case, we will think of $X$ as a set that newly joined $\xset$.

Whenever the algorithm for the \maintaincluster problem on graph $\hat G[X]$ produces a cut $(A,B)$ in $\hat G[X]$ of sparsity at most $O\left(\frac{\log^6n}{d_{X_0}} \right )$, we proceed as follows. 
 We partition inteval $I_X$ into two intervals: interval $I_A\subseteq I_X$ of size $|A|$ and interval $I_B\subseteq I_X$ of size $|B|$, so that $I_B$ appears to the left of $I_A$. We say that set $X_0$ \emph{owns} the edges in both $E_{\hat G}(A,B)$ and $E_{\hat G}(B,A)$. Notice that edges of $E_{\hat G}(B,A)$ are directed from left to right, and edges of $E_{\hat G}(A,B)$ are directed from right to left. We will say that the edges of $E_{\hat G}(A,B)$ are \emph{bad} edges. 
Recall that, if $|A|\leq |B|$, the algorithm for the \maintaincluster problem deletes $A$ from $X$ and continues with the new set $X=B$.
We then assign interval $I_B$ to the new set $X=B$, and we add the set $A$ to the collection $\xset$ of vertex subsets that we maintain (after which the algorithm for the \maintaincluster problem is applied to $\hat G[A]$ if it is a non-leaf cluster, or it is processed as a leaf cluster otherwise). We say that set $X$ of vertices just underwent \emph{splitting}, and that cluster $A$ was split off from it. Similarly, if $|B|<|A|$, the algorithm for the \maintaincluster problem deletes $B$ from $X$ and continues with the new set $X=A$.
We then assign interval $I_A$ to the new set $X=A$, and we add the set $B$ to the collection $\xset$ of vertex subsets that we maintain. We say that set $X$ of vertices just underwent splitting, and that cluster $B$ was split off from it.

It is easy to verify by a simple charging argument that the total number of bad edges that cluster $X_0$ owns at the end of the algorithm is bounded by $\otilde\left(\frac{|X_0|}{d_{X_0}}\right )$. It is also easy to verify that each such bad edge $e$ has $\spann(e)\leq |X_0|$ when it first becomes a bad edge, and from that point onward $\spann(e)$ may only decrease.
This completes the description of the first part of the algorithm. We now bound its total update time and the total span of all bad edges.

\paragraph{Bounding the Running Time of Part 1.}
Let $\xset^*$ be the collection of all non-leaf sets $X_0$ that were ever added to $\xset$, excluding $S$ and $T$. For all $0\leq i\leq \floor{\log n}$, let $\xset_i\subseteq \xset^*$ be the collection of all sets $X_0$ with $2^i\leq |X_0|<2^{i+1}$. It is easy to verify that all sets in the collection $\xset_i$ are mutually disjoint, and so $|\xset_i|\leq \frac{n}{2^i}$. 

Recall that for each set $X_0\in \xset_i$, the total update time of the algorithm for the \maintaincluster problem on $\hat G[X]$ from \Cref{thm: main for expander tools} is most:

\[\otilde(m_0(X_0)\cdot d_{X_0}+|X_0|^2).\]

Since $d_{X_0}=\frac{|X_0|}{n}\cdot d$ and $|X_0|\leq 2^{i+1}$, we get that this time is bounded by:

\[
\begin{split}
\otilde(m_0(X_0)\cdot d_{X_0}+|X_0|^2)&\leq\otilde \left(m_0(X_0)\cdot \frac{|X_0|}{n}\cdot d+|X_0|^2\right ) \\
%&\leq  \ohat \left(\frac{|X_0|^3}{n}\cdot d+ \frac{|X_0|^3}{n^2}\cdot d^2\right )\\
&\leq \otilde \left(\frac{2^{i}\cdot m_0(X_0)}{n}\cdot d+ 2^{2i}\right )
\end{split} \]

Since $|\xset_i|\leq \frac{n}{2^i}$, the total running time for all clusters $X_0\in \xset_i$ is bounded by: 

\[\otilde  \left(\frac{2^{i}\cdot m\cdot d}{n}+ n\cdot 2^{i}\right ),\]

where $m$ is the number of edges in $G$ at the beginning of the algorithm.

Summing over $0 \leq i \leq \floor{\log n}$, we get that the total update time over all invocations of the algorithm for the \maintaincluster problem is bounded by ${\otilde \left(md+n^2\right )}$. The remaining time that is required in order to maintain the \ATO is subsumed by this time.

\paragraph{Bounding the Total Span of All Bad Edges.}
We partition the set of all bad edges into two subsets. The first subset, $E_1$, contains all edges that are owned by non-leaf clusters.

Recall that for all $0\leq i\leq \floor{\log n}$, every cluster $X_0\in \xset_i$  may own at most $\otilde\left(\frac{|X_0| }{d_{X_0}}\right)$ bad edges. Initially, when an edge $e$ is added to the set of bad edges, it has $\spann(e)\leq |X_0|$, and after that  its $\spann(e)$ value may only decrease. In our calculations, for each bad edge $e$, we consider the initial value $\spann(e)$, when the edge was just added to the set of bad edges. The summation of all $\spann(e)$ values of all bad edges owned by clusters in $\xset_i$ is then bounded by:

\[
\begin{split}
\sum_{X_0\in \xset_i}\otilde \left(\frac{|X_0|^2}{d_{X_0}}\right )& \leq \sum_{X_0\in \xset_i}\frac{|X_0|\cdot\otilde(n)}{d}\\
&\leq \otilde\left(\frac{n^{2}}{d}\right ),
\end{split}
\]

since $d_{X_0}=\frac{|X_0|}{n}\cdot d$.
Summing over $0 \leq i \leq \floor{\log n}$, we conclude that  $\sum_{e\in E_1}\spann(e)\leq \otilde\left(\frac{n^{2}}{d}\right )$.

We denote by $E_2$ the set of all bad edges that are owned by leaf clusters. Recall that each leaf cluster $X$ owns at most $|X|$ bad edges, and the $\spann(e)$ value of each such edge is bounded by $|X|\leq \otilde\left(\frac{n}{d}\right )$. Since all leaf clusters that were ever added to $\xset$ are disjoint, we get that $\sum_{e\in E_2}\spann(e)\leq \otilde\left(\frac{n^2}{d}\right )$. Overall, we get that the total span value of all bad edges is bounded by $\otilde\left(\frac{n^2}{d}\right )$.

So far we viewed $(\xset,\rho)$ as an \ATO in the graph $\hat G$. It equivalently defines an $\ATO$ in the original graph $G$. However, for each right-to-left edge $e$ of $\hat G$, we may now have up to $\log m$ parallel copies of $e$ in $G$, and additionally, for every long edge $e$, it is possible that the edge is a right-to-left edge.

Observe first that, from our discussion so far, $\sum_{e\in E(G)\setminus \Elong}\spann(e)\leq \otilde\left(\frac{n^2}{d}\right )$ holds at all times (recall that $\Elong$ is the set of the long edges).

As long as Condition 
\ref{stopping conditionfinal} holds, since the paths in $\qset$ are edge-disjoint, there must be some $s$-$t$ path $Q$ whose length is at most $\Lambda$ in the current graph, and $\sum_{e\in E(Q)\setminus \Elong}\spann(e)\leq 
\otilde \left(\frac{n^{2}}{d}\cdot  \frac{256\log^2m}{\Delta}\right )=\otilde \left(\frac{n^{2}}{d\cdot \Delta}\right )$.

Since every edge $e\in \Elong$ has length at least $\frac{\Lambda}{64d\log m}$, and the length of $Q$ is at most $\Lambda$, we get that at most $64d\log m$ edges of $\Elong$ may lie in $Q$. Since for each such edge $e\in \Elong$, $\spann(e)\leq n$, we conclude that:
$\sum_{e\in E(Q)}\spann(e)\leq \otilde \left(\frac{n^2}{d\cdot \Delta}+n\cdot d\right )$.

It is then immediate to verify that, for this path $Q$, $\sum_{e\in E(Q)}\skipp(e)\leq \otilde \left(\frac{n^{2}}{d\cdot \Delta}+n\cdot d\right )$. By selecting an appropriate parameter $\Gamma=\otilde \left(\frac{n^{2}}{d\cdot \Delta}+n\cdot d\right )$, we ensure that, as long as Condition \ref{stopping conditionfinal} holds, there is an $s$-$t$ path $Q$ in $G$ with $\sum_{e\in E(Q)}\skipp(e)\leq \Gamma/4$, such that the length of $Q$ is at most $\Lambda$.

\subsection{Part 2: Solving \DLSSSP on the Contracted Graph}
The second part of our algorithm maintains the contracted graph $H=G_{|\xset}$. Once this graph is intialized, it is easy to verify that all changes to this graph can be implemented via vertex-splitting and edge-deletion updates. Whenever a cluster $A$ is split off from some cluster $X\in \xset$, we perform a vertex splitting operation, in which vertex $v_A$ is split off from $v_X$. We include in the corresponding set $E'$ all edges that have one endpoint in $A$ and another outside of $A$. Notice that each such edge either connects a vertex of $A$ and a vertex of $X$; or it has one endpoint in $A$ and another in $V(G)\setminus\set{A\cup X}$. In the latter case, if $e=(a,b)$ with $a\in A$ and $b\not\in A\cup X$, then edge $(v_X,v_Y)$, where $Y$ is the cluster of $\xset$ containing $b$, was already present in $H$. After the vertex-splitting update is complete, we delete this edge from $H$ via standard edge-deletion update if needed. Similarly, if the edge is $e=(b,a)$ with 
$a\in A$ and $b\not\in A\cup X$, then edge $(v_Y,v_X)$, where $Y$ is the cluster of $\xset$ containing $b$, was already present in $H$. After the vertex-splitting update is complete, we delete this edge from $H$ via standard edge-deletion update. 
Similarly, if $X$ is a leaf cluster that was just added to $\xset$, then we implement the splitting of $X$ into clusters that consist of individual vertices of $X$ via the vertex-splitting update, applied to supernode $v_X$.
Notice that this introduces a number of extra edges in graph $H$. For every edge $e$, whenever a cluster containing one endpoint $a$ of $e$ is split, if $a$ belongs to the smaller side of the partition, then a new copy of $e$ is introduced into $H$, and the older copy is deleted from it. But it is easy to verify that for each edge $e$ of $G$, at most $O(\log n)$ copies may be introduced, so the number of edges ever present in $H$ is still bounded by $O(m\log^2m)$.

For every edge $e$ of $H$, the weight $w(e)$ is set to be the smallest integral power of $2$ that is at least $\skipp(e)$. Recall that the value $\skipp(e)$ of an edge may grow overtime. In order to avoid having to  insert edges into $H$ outside of the vertex-splitting updates, for every edge $e$, for every integer $\ceil{\log(\skipp(e))}\leq i\leq \ceil{\log n}$, we add a copy of $e$ whose weight is $2^i$. As the value $\skipp(e)$ grows, we simply remove copies of $e$ as needed. The total number of edges ever present in $H$ is now bounded by $O(m\log^3m)$.

Consider now some vertex $v=v_X\in V(H)$, and some integer $0\leq i\leq \ceil{\log n}$. Recall that we have denoted by  $U_i(v)=\set{u\in V(H)\mid \exists e=(v,u)\in E(H)\mbox{ with } w(e)=2^i}$. 
Equivalently, set $U_i(v)$ contains all nodes $v_Y$, such that some edge $e$ of $H$ connects a vertex of $X$ to a vertex of $Y$, and has value $\skipp(e)\leq 2^{i}$. In particular, the distance between intervals $I_X$ and $I_Y$ is at most $2^i$, and so at most $2^i$ other intervals may lie between them. Therefore, it is immediate to verify that $|U_i(v)|\leq 2^{i+2}$ always holds, establishing Property \ref{prop: few close neighbors}.	

We apply the algorithm from \Cref{thm: almost DAG routing} to the resulting instance of the \DLSSSP problem, with the parameter $\Gamma$ that we defined above, distance parameter $d=\Lambda$, and $\eps=1/20$. Recall that the algorithm may only return FAIL if the current graph $H$ does not contain any $s$-$t$ path $P$  with $\sum_{e\in E(P)}\ell(e)\leq \Lambda$, and $\sum_{e\in E(P)}w(e)\leq \Gamma$. From our discussion, this may only happen when Condition \ref{stopping conditionfinal} is violated. Recall that the running time of the algorithm is bounded by:

%\[\otilde \left(n^2 + \Gamma\cdot n\right )\leq \otilde \left( n^{2}\cdot d+\frac{n^{3}}{d\cdot \Delta}\right ).\]

\[\otilde \left(m+n^2 + \Gamma\cdot n\right )\leq \otilde \left( n^{2}\cdot d+\frac{n^{3}}{d\cdot \Delta}\right ),\]

since $\Gamma=\otilde \left(\frac{n^{2}}{d\cdot \Delta}+n\cdot d\right )$.

The total running update time of both parts of the algorithm is bounded by:

\[\otilde \left(n^{2} d+\frac{n^{3}}{d\cdot \Delta}\right ) \]

Substituting the value $d=\sqrt{n/\Delta}$, we obtain total update time:

\[{\otilde\left(\frac{n^{2.5}}{\sqrt{\Delta}}\right ) }\]

\paragraph{Completing the Algorithm.}
At the beginning of every iteration, we use the algorithm from \Cref{thm: almost DAG routing} to compute a simple $s$-$t$ path $P$ in graph $H$, whose length is at most $(1+10\eps)\Lambda\leq 1.5\Lambda$. Let $(v_S=v_{X_1},v_{X_2},\ldots,v_{X_k}=v_T)$ be the  sequence of supernodes on this path. For all $1\leq j<k$, let $e_j=(x_j,y_j)$ be the corresponding edge in graph $G$, with $x_j\in X_j$ and $y_j\in X_{j+1}$.

For all $1< j< k$, 
if $|X_j|=1$, then we let $Q_j$ be a path that consists of a single vertex $y_{j-1}=x_{j}$. Otherwise,
we use the data structure from the MaintainCluster problem on graph $G[X_{j}]$ in order to compute a path $Q_j$, connecting $y_{j-1}$ to $x_{j}$. Recall that the path contains at most $d_{X_{j}}$ edges, and, since each such edge is short, the length of the path in $G$ is bounded by $d_{X_{j}}\cdot \frac{\Lambda}{64d\log m}$.

We let $P'$ be the path obtained by concatenating $e_1,Q_2,e_2,\ldots,Q_{k-1},e_{k-1}$. It is easy to see that $P'$ is an $s$-$t$ path in the current graph $G$. Next, we bound its length. Recall that the length of the original path $P$ was at most $1.5\Lambda$. For convenience, for $1<j<k$, if $|X_j|=1$, then we denote $d_{X_j}=0$. The total lengths of all paths in $\set{Q_2,\ldots,Q_{k-1}}$ can be bounded by:

\[ 
\begin{split}
 \frac{\Lambda}{64d\log m}\cdot \sum_{j=2}^{k-1} d_{X_j}&\leq \frac{\Lambda}{64d\log m}\cdot \sum_{X_0\in \xset^*}d_{X_0}\\
&\leq \frac{\Lambda}{64d\log m}\cdot \sum_{i=0}^{\floor{\log n}}\sum_{X_0\in \xset_i}d_{X_0}\\
&=\frac{\Lambda}{64d\log m}\cdot \sum_{i=0}^{\floor{\log n}}\sum_{X_0\in \xset_i}\frac{|X_0|}{n}\cdot d\\
&\leq \frac{\Lambda}{64d\log m} \cdot d\log n\\
&\leq \frac{\Lambda}{64},
\end{split}
 \]

since for all $X_0\in \xset$, $d_{X_0}=\frac{|X_0|}{n}\cdot d$, and for all $0\leq i\leq \floor{\log n}$, $\sum_{X_0\in \xset_i}|X_0|\leq n$.
Altogether, the length of the final path $P'$ is bounded by $2\Lambda$ as required. %In time $O(\Lambda)$, we turn path $P'$ into a simple path that the algorithm then returns.
Note that the initial path $P$, returned by the algorithm from  \Cref{thm: almost DAG routing}  is a simple path. From the definition of the \maintaincluster problem, each of the paths $Q_2,\ldots,Q_{k-1}$ are also simple paths. Moreover, since all clusters in $\xset$ are disjoint, and since path $P$ is simple, it is easy to verify that the final path $P'$ is simple as well.

After our algorithm returns a path $Q$, the edges of $Q$ are deleted from $G$. We then update the \maintaincluster algorithms for clusters $X\in \xset$ with the edges that were deleted from $\hat G[X]$, and then update graph $H=G_{|\xset}$ with edge deletions and any vertex splits that may have resulted from splitting of clusters of $\xset$.

Let $q$ denote the number of iterations in the algorithm, and recall that $q\leq \Delta$. In each iteration $i$, the time required to compute the initial path $P_i$ in graph $H$ is bounded by $O(|E(P_i)|)$. The time that is required to respond to \shortpath queries within the clusters $X_j$ is already accounted for as part of total running time of the \maintaincluster algorithm on $\hat G[X_j]$. Let $P'_i$ denote the final path that the algorithm returns in iteration $i$, so $E(P_i)\subseteq E(P'_i)$. Then the additional time required to respond to queries is bounded by $\sum_{i=1}^qO(|E(P_i)|)$. Since, for all $1\leq i\leq q$, the edges of $P'_i$ are deleted from 
$G$ at the end of iteration $i$, the edges of $P_i$ are also deleted from $G$ at the end of iteration $i$. Therefore, we get that $\sum_{i=1}^qO(|E(P_i)|)\leq \otilde(|E(G)|)\leq \otilde(n^2)$. Since $\Delta\leq n$,
 the total running time of the algorithm is then bounded by $\otilde \left(\frac{n^{2.5}}{\sqrt{\Delta}}\right )+\otilde(n^2)\leq {\otilde \left(\frac{n^{2.5}}{\sqrt{\Delta}}\right ) }$.

\section{Algorithm for the \maintaincluster Problem -- Proof of \Cref{thm: main for expander tools}}
\label{sec: expander algs}

We will use the following theorem as a subroutine. The theorem provides an efficient algorithm that either embeds a large expander into the input graph $H$, or computes a sparse cut in $H$, with both sides of the cut being sufficiently large. The proof of the theorem uses standard techniques (namely the Cut-Matching Game), though we provide a simpler algorithm for one of the ingredients (algorithm for the Cut Player), and an algorithm achieving slightly better parameters for another ingredient (algorithm for the Matching Player). The theorem uses parameter $\alpha_0$ from \Cref{thm:explicit expander}; its proof is deferred to Section \ref{sec: embed or cut}.

\begin{theorem}\label{thm: embed expander or cut}	
	There is a deterministic algorithm, whose input consists of a directed $n$-vertex bipartite graph $G=(L,R,E)$ with $|E|=m$, such that every vertex in $R$ has out-degree at most $1$ and every vertex in $L$ has in-degree at most $1$, and a parameter $4\leq d'\leq 2n$.
		The algorithm returns one of the following:
	\begin{itemize}
		\item Either a cut $(X,Y)$ in $G$ with $|E_G(X,Y)|\leq \frac{2n}{d'}+1$ and $|X|,|Y|\geq \Omega\left ( \frac{n}{\log^3 n}\right )$; or
		
		\item A directed graph $W^*$ with $|V(W^*)|\geq n/8$ and maximum vertex degree at most $18$, such that $W^*$ is an $\alpha_0$-expander, together with a subset $F\subseteq E(W^*)$ of at most $\frac{\alpha_0\cdot n}{500}$ edges of $W^*$, and an embedding $\pset^*$ of $W^*\setminus F$ into $G$ via paths of length at most $O(d'\log^2 n)$, that cause congestion at most $O(d'\log^4n)$, such that every path in $\pset^*$ is simple.
	\end{itemize}
	The running time of the algorithm is $\otilde(md'+n^2)$.
\end{theorem}

Note that in the above theorem we assume that all edge lengths in the input graph $G$ are unit.

We will also use the following simple observation that uses the standard ball-growing technique. The proof is deferred to Section \ref{sec: appx: ball growing 2} of Appendix.

\begin{observation}\label{obs: ball growing2}
	There is a deterministic algorithm, that receives as input a directed bipartite $n$-vertex graph $H=(L,R,E)$, where the edges in $E_H(L,R)$ are called regular and the edges in $E_H(R,L)$ are called special. Every vertex of $H$ may be incident to at most one special edge. Additionally, the algorithm is given a set $X\subseteq V(H)$ of at least $\frac{n}{32}$ vertices, and a vertex $v$, such that $\dist_H(v,X)\geq d$, for some given parameter ${d\geq 2^{11}\log n}$. The algorithm computes a  well-structured cut $(A,B)$ in $H$, with $v\in A$, $X\subseteq B$, and $|E_H(A,B)|\leq \phi\cdot\min\set{|A|,|B|}$, for 	${\phi=\frac{2^{11}\log n}{d}}$. The running time of the algorithm is $O(\vol_H(A))$.
\end{observation}

We are now ready to complete the proof of \Cref{thm: main for expander tools}. Recall that we are given as input a 
simple  bipartite $n$-vertex graph $H=(L,R,E)$ with $|E|=m$, such that every vertex in $R$ has out-degree at most $1$ and every vertex in $L$ has in-degree at most $1$, together with integral parameters $c^*\log^6n\leq d^*\leq n$, where $c^*$ is a large enough constant and $1\leq \Delta\leq n$.
Recall that we refer to the edges of $E_H(R,L)$ as \emph{special}, and the edges of $E_H(L,R)$ as \emph{regular}. Recall also that a cut $(X,Y)$ in $H$ is \emph{well-structured} iff all edges in $E_H(X,Y)$ are special.

The algorithm is required to respond to at most $\Delta$ \shortpath queries: given a pair $x,y\in V(H)$ of vertices, it must return a simple path $P$ connecting $x$ to $y$ in $H$, whose length is at most $d^*$. Once the algorithm responds to a query, it receives a set of edges that need to be deleted from $H$, with all the deleted edges lying on path $P$.  At any time, the algorithm may provide a well-structured cut $(X,Y)$ of $H$, whose sparsity is at most $O\left(\frac{\log^6 n}{d^*}\right )$. 
If $|X|\leq |Y|$, then the cut must be provided by listing the vertices of $X$, and otherwise it must be provided by listing the vertices of $Y$. After that the vertices of the smaller side of the cut are deleted from $H$. %Once the number of vertices in $H$ falls below half the original number, the algorithm halts, even if fewer than $\Delta$ queries were asked. 

Note that, over the course of the algorithm, vertices may be deleted from $H$. We denote by $n$ the initial number of vertices in $H$. Once $|V(H)|$ falls below $n/2$, the algorithm terminates, even if fewer than $\Delta$ queries were asked.

Throughout, we use a parameter ${d=\frac{d^*}{c'\log^3 n}}$, where $c'$ is some large enough constant, whose value we will set later. Since $d^*\geq c^*\cdot \log^6 n$ for a large enough constant $c^*$, we can ensure that $d\geq 2^{11}\log n$.

Our algorithm consists of a number of phases. 
Let $N=\frac{\alpha_0 n}{2^{20} (d^*)^2\log n}$, where $\alpha_0$ is the constant from \Cref{thm:explicit expander}, and let $\hat c$ be a large enough contant whose value we choose later. A phase terminates whenever one of the following two conditions are met: either (i) the number of vertices in $H$ decreased by the factor of at least ${\left( 1-\frac{1}{\hat c\log^5n}\right )}$ over the course of the current phase (in which case we say that the phase is of \emph{type 1}); or (ii) at least $N$ \shortpath queries were responded to during the current phase (in which case the phase is of \emph{type 2}). Clearly, at most $O(\log^5n)$ phases may be of type 1, and the number of type-2 phases is bounded by $\max\set{1,\frac{\Delta}{N}}\leq \max\set{1,O\left ( \frac{\Delta\cdot (d^*)^2\log n}{n}\right )}$. Overall, the number of phases in the algorithm is bounded by $O\left (\max\set{\log^5n,\frac{\Delta\cdot (d^*)^2\cdot \log^6n}{n}}\right )$.
At a high level, we follow a scheme that is by now quite standard: we embed a large $\Omega(1)$-expander $W^*$ into graph $H$, and then maintain two \EST data structures in  $H$, both of which are rooted at the vertices of $W^*$; in the first tree all edges are directed away from the root, while in the second tree the edges are directed towards the root. Given a query $(x,y)$, we use the \EST data structures in order to connect $x$ to some vertex $x'\in V(W^*)$, and to connect some vertex $y'\in V(W^*)$ to $y$, via short paths. Lastly, we compute a short path connecting $x'$ to $y'$ in $W^*$ by performing a BFS, which is then transformed into an $x'$-$y'$ path in $H$ by exploiting the embedding of $W^*$ into $H$. At a high level, we can view each phase as consisting of two steps. In the first step, we embed an expander $W^*$ into $H$, and the second step contains the remainder of the algorithm that allows us to support \shortpath queries. We now describe each of the two steps in turn.

\subsection*{Step 1: Embedding an Expander}
During the first step of the phase,  we will attempt to embed a large expander into the current graph $H$ via \Cref{thm: embed expander or cut}. Specifically, we perform a number of iterations. In every iteration, we apply the algorithm from \Cref{thm: embed expander or cut} to the current graph $H$ with parameter $d'=d$.

Assume first that the algorithm returned a cut $(X,Y)$ in $H$ with $|E_H(X,Y)|\leq \frac{2n}{d}+1$ and $|X|,|Y|\geq \Omega\left ( \frac{n}{\log^3 n}\right )$. In this case, we say that the current iteration is a \emph{bad} iteration.
Clearly, $|E_H(X,Y)|\leq \frac{2n}{d}+1\leq \frac{4n}{d}\leq O\left(\frac{\log^3n}{d} \right )\cdot \min\set{|X|,|Y|}$. In other words, the sparsity of the cut $(X,Y)$ is at most $O\left(\frac{\log^3n}{d} \right )=O\left(\frac{\log^6n}{d^*} \right )$. 
Since $d^*\geq c^*\log^6n$, where $c^*$ is a large enough constant, we can ensure that the sparsity of the cut, that we denote by $\phi$, is bounded by $1/4$.
Using the algorithm from \Cref{obs: sparse cut to structured}, we compute, in time $O(|E(H)|)$, a well-structured cut $(X',Y')$ in $H$, whose sparsity remains at most $O\left(\frac{\log^6n}{d^*} \right )$. 
%Since $d^*\geq c^*\log^4n$, where $c^*$ is a large enough constant, we can ensure that the sparsity of the cut, that we denote by $\phi$, is bounded by $1/4$.
%We are also guaranteed that $\vol_H(X'),\vol_H(Y')\geq \Omega(\min\set{\vol_H(A),\vol_H(B)})$.

Denote the sets $X$ and $Y$ by $Z$ and $Z'$, so that $Z$ is the smaller set of the two, breaking the tie arbitrarily. Notice that the algorithm from \Cref{obs: sparse cut to structured} ensures that $|Z|\geq (1-\phi)\cdot \min\set{|X|,|Y|}\geq \Omega\left ( \frac{n}{\log^3 n}\right )$. We output the cut $(X,Y)$ and delete the vertices of $Z$ from graph $H$. If the number of vertices in $H$ falls below $n/2$, then we terminate the algorithm. Otherwise, we continue to the next iteration. 

Assume now that the algorithm from \Cref{thm: embed expander or cut}	 returned a directed graph $W^*$ with $|V(W^*)|\geq |V(H)|/8$ and maximum vertex degree at most $18$, such that $W^*$ is an $\alpha_0$-expander, together with a subset $F\subseteq E(W^*)$ of at most $\frac{\alpha_0\cdot |V(H)|}{500}$ edges of $W^*$, and an embedding $\pset^*$ of $W^*\setminus F$ into $H$ via simple paths of length at most $O(d\log^2 n)$ each, that cause congestion at most $O(d\log^4n)$.
In this case, we say that the current iteration is  a \emph{good} iteration. We then terminate the first step of the current phase and continue to the second step.

Since in every bad iteration at least $\Omega\left ( \frac{n}{\log^3 n}\right )$ vertices are deleted from $H$, the total number of iterations in the first step is bounded by $O(\log^3n)$. In every iteration we apply the algorithm from \Cref{thm: embed expander or cut}, whose running time is $\otilde(|E(H)|d+n^2)$ to the current graph $H$, and then possibly use the algorithm from \Cref{obs: sparse cut to structured}, whose running time is at most $O(|E(H)|)$. Overall, the running time of Step 1 of a single phase is $\otilde(|E(H)|\cdot d+n^2)$.

\subsection*{Step 2: Even-Shiloach Trees and Supporting Short-Path Queries}
We denote by $n'$ the number of vertices in $H$ at the beginning of Step 2. The purpose of Step 2 is to support \shortpath queries. 
Over the course of Step 2, the number of vertices in $H$ may decrease. Whenever this number falls below $n'\cdot \left(1-\frac{1}{\hat c\log^5n}\right )$, for a large enough constant $\hat c$ whose value we choose later, we terminate the algorithm for Step 2 and the current phase.
We start by defining additional data structures that we maintain.

Recall that, at the beginning of Step 2, we are given an $\alpha_0$-expander graph $W^*$ with  $|V(W^*)|\geq n'/8$ and maximum vertex degree at most $18$, together with a subset $F\subseteq E(W^*)$ of at most $\frac{\alpha_0\cdot n'}{500}$ edges of $W^*$, and an embedding $\pset^*$ of $W^*\setminus F$ into $H$ via simple paths of length at most $O(d\log^2 n)$, that cause congestion at most $O(d\log^4n)$.

We denote the number of vertices in $W^*$ at the beginning of the phase by $n''$, so $n''\geq n'/8\geq n/16$.
 Over the course of Step 2, some vertices may be deleted from $W^*$, but we will ensure that throughout the algorithm for Step 2, $|V(W^*)|\geq n''/4\geq |V(H)|/32$ holds.
 We also use a parameter $\hat d=\frac{2^{30}\log n}{\alpha_0}$, and we will ensure that at all times, for every pair $x,y$ of vertices of $W^*$, $\dist_{W^*\setminus F}(x,y),\dist_{W^*\setminus F}(y,x)\leq \hat d$.

We maintain two additional graphs: graph $H'$ that is obtained from $H$ by adding a source vertex $s'$, and connecting it to every vertex in $V(W^*)$ with an edge of length $1$, and graph $H''$ that is obtained from $H$ by adding a vertex $s''$, connecting every vertex in $V(W^*)$ to it with an edge of length $1$, and then reversing the directions of all edges. Whenever a vertex $v$ leaves the set $V(W^*)$ of vertices, we delete the edge $(s',v)$ from $H'$, and we also delete the edge $(s'',v)$ from $H''$. Additionally, whenever an edge is deleted from graph $H$ due to an update operation, we delete the corresponding edge from both $H'$ and $H''$. When vertices are deleted from graph $H$ (after sparse cuts are returned by the algorithm), we also update both $H'$ and $H''$ with the deletions of these vertices.

We initialize the \EST data structure from \Cref{thm: directed weighted ESTree} for graph $H'$, with source vertex $s'$ and distance parameter $d+1$, and we similarly initialize \EST data structure for $H''$ with source vertex $s''$ and distance parameter $d+1$. We denote the corresponding trees by $T'$ and $T''$, respectively.

We will ensure that, throughout the algorithm, every vertex $v$ of $H$ lies in both $T'$ and $T''$, so that $\dist_H(V(W^*),v),\dist_H(v,V(W^*))\leq d$ always holds. In order to ensure this, we use the \EST data structures to identify vertices $v$ for which this no longer holds. Whenever we identify a vertex $v$ with $\dist_H(v,V(W^*))>d$, we  use the algorithm from \Cref{obs: ball growing2} in order to compute a well-structured cut $(A,B)$ in $H$ with $v\in A$, $V(W^*)\subseteq B$, and $|E_H(A,B)|\leq \frac{2^{11}\log n}{d}\cdot \min\set{|A|,|B|}\leq O\left(\frac{\log^6n}{d^*} \right )\cdot \min\set{|A|,|B|}$. If $|A|\leq |B|$, then we delete the vertices of $A$ from $H$, (which includes deleting these vertices and their incident edges from graphs $H,H'$, and $H''$, and both \EST data structures $T'$ and $T''$), and then continue the algorithm. If $|V(H)|$ falls below $n'\left(1-\frac{1}{\hat c\log^5n}\right )$, then the algorithm is terminated. Alternatively, if $|A|>|B|$, then we delete the vertices of $B$ from $H$, and terminate the current phase. Notice that in this case, at least $|V(W^*)|\geq n''/4\geq n/64$  vertices have been deleted from $H$ in the current phase.

Similarly, whenever we identify a vertex $v$ with $\dist_H(V(W^*),v)>d$, we  use the algorithm from \Cref{obs: ball growing2} (but with reversed directions of edges) in order to compute a well-structured cut $(A,B)$ in $H$ with $V(W^*)\subseteq  A$, $v\in B$, and $|E_H(A,B)|\leq \frac{2^{11}\log n}{d}\cdot \min\set{|A|,|B|}\leq O\left(\frac{\log^6n}{d^*} \right )\cdot \min\set{|A|,|B|}$. As before, if $|B|\leq |A|$, then we delete the vertices of $B$ from $H$, and from both \EST data structures. If $|V(H)|$ falls below $n'\left(1-\frac{1}{\hat c\log^5n}\right )$, then the algorithm is terminated, and otherwise we continue the algorithm as before. Alternatively, if $|B|>|A|$, then we delete the vertices of $A$ from $H$, and terminate the current phase. In the latter case, at least $n/64$ vertices have been deleted from $H$ in the current phase.

In either case, if the set of vertices that was deleted from $H$ is denoted by $Z$, then the running time of the algorithm from \Cref{obs: ball growing2}  is bounded by $O(\vol_H(Z))$. Therefore, the total time that the algorithm ever spends over the course of Step 2 to cut sparse cuts off from $H$ via \Cref{obs: ball growing2}  is bounded by $O(m)$. Additionally, the algorithms for maintaining the two \EST data structures have total update time $O(m\cdot d)$.
Therefore, the total update time of Step 2 of a single phase is so far bounded by $O(m\cdot d)\leq \otilde(m\cdot d^*)$.

For every edge $e\in E(H)$, we maintain a set $S(e)$ of all edges $e'\in E(W^*)\setminus F$ whose embedding path $P(e')\in \pset^*$ contains $e$. We can initialize the sets $S(e)$ of edges for all $e\in E(H)$ by scanning every edge $e'\in E(W^*)\setminus F$ and considering all edges on path $P(e')$. Clearly, the total time required to initialize the sets $S(e)$ of edges for all $e\in E(H)$ is subsumed by the time spent on constructing the embedding $\pset^*$ of $W^*\setminus F$ into $H$ in Step 1.

\subsection*{A Cleanup Procedure}
The purpose of the cleanup procedure is to ensure that every vertex of $H$ lies in both the trees $T'$ and $T''$, and additionally, that for every pair $x,y$ of vertices of $W^*$, $\dist_{W^*\setminus F}(x,y),\dist_{W^*\setminus F}(y,x)\leq \hat d$ holds. The cleanup procedure is executed before each query. In other words, we execute the procedure once at the beginning of Phase 2, before responding to any queries. Then, after every query, once graph $H$ is updated with the relevant edge deletions, we execute the procedure again.

Consider for example a \shortpath query, and suppose the algorithm returned a path $P^*$ in response to the query. Recall that some edges of $P^*$ may be deleted from $H$. As a result, it is possible that some vertices of $H$ no longer belong to one or both of the trees $T',T''$. As described already, we will compute a sparse cut in $H$ in order to fix this issue, but this may in turn affect the expander $W^*$. Additionally, the deletion of the edges from $H$ may lead to the deletion of some edges from the expander $W^*$. As a result, it is possible that the diameter of $W^*\setminus F$ grows beyond $\hat d$. In such a case, we will  compute  a sparse cut in graph $W^*\setminus F$, but this may in turn affect the trees $T',T''$, since, whenever a vertex $x$ leaves the set $V(W^*)$ of vertices, we delete edge $(s',x)$ from $T'$ and we delete edge $(s'',x)$ from $T''$. Similarly, it is possible that, at the beginning of Phase 2, not all vertices of $H$ lie in the trees $T',T''$, or that the diameter of $W^*\setminus F$ is greater than $\hat d$. Fixing each of these issues may in turn create new issues as described above. We now formally describe the cleanup procedure.

At the beginning of the cleanup procedure, we select an arbitrary vertex  $v\in V(W^*)$, and compute two shortest-paths trees $\hat T'$ and $\hat T''$ rooted at $v$ in graph $W^*\setminus F$, where the edges in $\hat T'$ are directed away from $v$, and the edges in $\hat T''$ are directed towards $v$. In other words, for every vertex $a\in V(W^*)$, the unique $v$-$a$ path in $\hat T'$ is a shortest $v$-$a$ path in $W^*\setminus F$, and the unique $a$-$v$ path in $\hat T''$ is a shortest $a$-$v$ path in $W^*\setminus F$. Note that both trees can be computed in time $O(n)$. We then perform a number of iterations. At the beginning of every iteration, we assume that we are given some vertex $v\in V(W^*)$, and two shortest-paths trees $\hat T'$ and $\hat T''$ in $W^*\setminus F$ that are both rooted at $v$, with the edges of $\hat T'$ directed away from $v$, and the edges of $\hat T'$ directed towards $v$. The iterations are executed as long as at least one of the following two conditions hold: (i) some vertex $u\in V(H)$ does not lie in $T'$ or in $T''$; or (ii) the depth of at least one of the trees $\hat T',\hat T''$ is greater than $\hat d$. We now describe the execution of a single iteration.

Assume first that some vertex $u\in V(H)$ no longer lies in either $T'$ or $T''$.
In this case, we say that the current iteration is a \emph{type-1} iteration. Assume first that $u\not \in T''$, 
 in other words, $\dist_H(u,V(W^*))>d$ holds. As we described already, we  use the algorithm from \Cref{obs: ball growing2} in order to compute a well-structured cut $(A,B)$ in $H$ with $u\in A$, $V(W^*)\subseteq B$, and $|E_H(A,B)|\leq \frac{2^{11}\log n}{d}\cdot \min\set{|A|,|B|}\leq O\left(\frac{\log^6n}{d^*} \right )\cdot \min\set{|A|,|B|}$. If $|A|\leq |B|$, then we delete the vertices of $A$ from $H$, (which includes deleting these vertices and their incident edges from graphs $H,H'$, and $H''$, and both \EST data structures $T'$ and $T''$). If $|V(H)|$ falls below $n'\left(1-\frac{1}{\hat c\log^5n}\right )$, then the algorithm is terminated, and otherwise we continue the algorithm. Alternatively, if $|A|>|B|$, then we delete the vertices of $B$ from $H$, and terminate the current phase.
Similarly, if some vertex $u\in V(H)$ no longer lies in $T'$, then, as desribed above, we  use the algorithm from \Cref{obs: ball growing2} (but with reversed directions of edges) in order to compute a well-structured cut $(A,B)$ in $H$ with $V(W^*)\subseteq  A$, $u\in B$, and $|E_H(A,B)|\leq \frac{2^{11}\log n}{d}\cdot \min\set{|A|,|B|}\leq O\left(\frac{\log^6n}{d^*} \right )\cdot \min\set{|A|,|B|}$. As before, if $|B|\leq |A|$, then we delete the vertices of $B$ from $H$, and from both \EST data structures. If $|V(H)|$ falls below $n'\left(1-\frac{1}{\hat c\log^5n}\right )$, then the algorithm is terminated, and otherwise we continue the algorithm. Alternatively, if $|B|>|A|$, then we delete the vertices of $A$ from $H$, and terminate the current phase. In the latter case, at least $n/64$ vertices have been deleted from $H$ in the current phase.
Let $Z$ be the set of all vertices that were deleted from $H$ in the current iteration, and let
$\hat E=\bigcup_{e\in E_H(A,B)}S(e)$, that is, these are all edges $e'\in E(W^*)\setminus F$ whose embedding path contains an edge of $E_H(A,B)$.
Since the embedding of $W^*\setminus F$ into $H$ causes congestion at most $O(d\log^4 n)$, we get that $|\hat E|\leq O(d\log^4 n)\cdot |E_H(A,B)|\leq O(|Z|\cdot \log^6n)$. 
We delete the edges of $\hat E$ from graph $W^*$. If $Z\cap V(W^*)=\emptyset$, then we say that the current iteration is a \emph{type-(1a) iteration}, and we continue to the next iteration. Notice that in this case, the running time of the current iteration is $\otilde(\vol_H(Z)\cdot d^*)$, and we \emph{charge} the running time of the current iteration to the set $Z$ of vertices.

Assume now that $Z\cap V(W^*)\neq \emptyset$. In this case, we say that the current iteration is a \emph{type-(1b) iteration}.  Let $X=Z\cap V(W^*)$. So far the total running time of the iteration is  $\tilde O(\vol_H(Z)\cdot d^*)$, and it is charged to the set $Z$ of vertices. Next, we delete the vertices of $X$ from $W^*$. The remainder of the running time of the current iteration will be charged to set $X$. Notice that, if $Z=A$, and we denote $Z'=B$, then graph $W^*\setminus F$ no longer contains any edges connecting vertices of $Z\cap V(W^*)$ to vertices of $Z'\cap V(W^*)$. In other words, $(Z\cap V(W^*),Z'\cap V(W^*))$ is a cut of sparsity $0$ in $W^*\setminus F$.
Similarly, if $Z=B$, and we denote $Z'=A$, then cut $(Z'\cap V(W^*),Z\cap V(W^*))$ is a cut of sparsity $0$ in $W^*\setminus F$. Once the vertices of $X$ are deleted from $W^*$, we select an arbitrary vertex $v\in V(W^*)$, and compute two shortest-paths trees $\hat T'$ and $\hat T''$ rooted at $v$ in graph $W^*\setminus F$, where the edges in $\hat T'$ are directed away from $v$, and the edges in $\hat T''$ are directed towards $v$. It is easy to verify that the total running time of a type-(1b) iteration is bounded by $\otilde(\vol_H(Z)\cdot d^*)+O(n)$. We charge time $\otilde(\vol_H(Z)\cdot d^*)$ to the set  $Z$ of vertices that is deleted from $H$ in the current iteration, and we charge time $O(n)$ to the set $X$ of vertices that is deleted from $W^*$ in the current iteration. 
Lastly, graphs $H'$ and $H''$ need to be updated to reflect the deletion of the vertices of $X$ from $W^*$: namely, for all $u\in X$, we delete the edge $(s',u)$ from $H'$, and we delete the edge $(s'',u)$ from $H''$. This, in turn, may require updating the \EST data structures $T'$ and $T''$. The time required to update the two trees is counted as part of the total update time of these data structures, and not as part of the running time of the current cleanup procedure.
This completes the description of a type-(1b) iteration.

If the depth of either of the trees $\hat T',\hat T''$ is greater than $\hat d$, then we say that the current iteration is a type-2 iteration. In this case, we can compute, in time $O(1)$, a pair $a,b\in V(W^*)$ of vertices with $\dist_{W^*\setminus F}(a,b)>\hat d$.
Recall that  $\hat d=\frac{2^{30}\log n}{\alpha_0}$, and the maximum vertex degree in $W^*$ is bounded by $D=18$. We can now use the algorithm from \Cref{obs: ball growing} to compute, in time $O(n)$, a cut $(A,B)$ in $W^*\setminus F$ of sparsity at most $\phi=\frac{32D\log n}{\hat d}\leq \frac{\alpha_0}{2^{20}}$. If $|A|\leq |B|$, then we delete from $W^*$ the vertices of $A$, and otherwise we delete from $W^*$ the vertices of $B$.
If we denote by $X$ the set of vertices deleted from $W^*$ in the current iteration, then graphs $H'$ and $H''$ need to be updated accordingly (namely, for all $u\in X$, we delete the edge $(s',u)$ from $H'$ and we delete the edge $(s'',u)$ from $H''$). This, in turn, may require updating the \EST data structures $T'$ and $T''$. The time required to update the two trees is counted as part of the total update time of these data structures, and not as part of the running time of the current cleanup procedure.
We then select an arbitrary vertex $v\in V(W^*)$, and compute two shortest-paths trees $\hat T'$ and $\hat T''$ rooted at $v$ in graph $W^*\setminus F$, where the edges in $\hat T'$ are directed away from $v$, and the edges in $\hat T''$ are directed towards $v$. Clearly, the running time of a type-2 iteration is $O(n)$, and we charge it to the set of vertices (either $A$ or $B$) that was deleted from $W^*$ over the course of the current iteration.  

The algorithm for the cleanup procedure terminates once every vertex of $H$ lies in both $T'$ and $T''$, and the depths of the trees $\hat T'$ and $\hat T''$ are both bounded by $\hat d$.

We now bound the total time required to perform the cleanup procedure over the course of the entire Step 2.
Since we only need to respond to at most $N=O\left(\frac{n}{(d^*)^2\log n}\right )$ queries, the number of times that the cleanup procedure is called is bounded by $N$. The time that is required in order to compute the shortest-path trees $\hat T',\hat T''$ in $W^*\setminus F$ for the first time at the beginning of each cleanup procedure is then bounded by $O(n\cdot N)\leq \otilde(n^2)$. 

If we denote by $Z$ the set of all vertices that are deleted over the course of the phase from graph $H$, and by $\xset$ the collection of all subsets of vertices that were deleted from $V(W^*)$ over the course of the phase, then, from our charging scheme, the remaining running time that is spent on the repeated executions of the cleanup procedure can be bounded by $\otilde(\vol_H(Z)\cdot d^*)+O(|\xset|\cdot n)\leq \otilde(md^*+n^2)$. Therefore, the total time spent over the course of a single phase on executing the cleanup procedure is bounded by $\otilde(md^*+n^2)$.

\subsection*{Updates to the Expander $W^*$}
Recall that, at the beginning of Step 2, we are given an $\alpha_0$-expander $W^*$, with $|V(W^*)|=n''\geq n'/8$ and maximum vertex degree at most $18$, together with a subset $F\subseteq E(W^*)$ of at most $\frac{\alpha_0\cdot n'}{500}$ edges of $W^*$, and an embedding $\pset^*$ of $W^*\setminus F$ into $H$ via paths of length at most $O(d\log^2 n)$, that cause congestion at most $O(d\log^4n)$. 

Over the course of the algorithm, we may delete some edges from $W^*$, and we may also compute sparse cuts in $W^*$ that decompose $W^*$ into smaller subgraphs. For convenience, we will maintain a collection $\xset$ of subsets of vertices of $W^*$ that partition the original set $V(W^*)$. At the beginning of the algorithm, $\xset$ contains a single set $V(W^*)$ of vertices. Over the course of the algorithm, if we identify a sparse cut $(X,Y)$ in the current graph $W^*\setminus  F$, we replace $V(W^*)$ with $X$ and $Y$ in $\xset$, and we delete the vertices of the smaller of the two sets from $W^*$. The key is to ensure that, since the cuts that we compute are sparse, and since the number of edges deleted from $W^*$ over the course of the second step is relatively small, $|V(W^*)|\geq n''/4$ continues to hold throughout Step 2 of the phase.

There are three types of updates that expander $W^*$ may undergo over the course of the algorithm, that we describe next.

\paragraph{First Type of Updates.}
Suppose some edge $e$ was deleted from $H$, as part of the update sequence that  graph $H$ undergoes. In this case, we delete every edge $e'\in S(e)$ from $W^*$. For each such edge $e'$, we also delete $e'$ from the set $S(e'')$ of every edge $e''\in E(H)$ that lies on the embedding path $P(e')\in \pset^*$. Recall that, over the course of the current phase, our algorithm will need to respond to at most $N$ \shortpath queries. If $Q$ is a path that is returned in response to such a query, then $Q$ may contain at most $d^*$ edges, and the only edges that may be deleted from $H$ over the course of the current phase as part of the update sequence are edges that lie on such paths. Therefore, the total number of edges that may be deleted from $H$ as part of the update sequence during a single phase is at most $N\cdot d^*$. Since the embedding of $W^*\setminus F$ causes congestion at most $O(d\log^4n)$ in $H$, the number of corresponding edge deletions from $W^*$ is bounded by:

\[O(N\cdot d^*\cdot d\log^4n)\leq N\cdot (d^*)^2\cdot \log n\leq \frac{\alpha_0 n}{2^{20}}, \]

since $d=\frac{d^*}{c'\log^3 n}$, where $c'$ can be chosen as a large enough constant and $N=\frac{\alpha_0 n}{2^{20} (d^*)^2\log n}$.
Throughout, we denote by $E'$ the set of edges that were deleted from $W^*$ as the result of this type of updates to $W^*$, so $|E'|\leq \frac{\alpha_0 n}{2^{20}}$ always holds. The time that is required to update graph $W^*$ with this type of updates is subsumed by the time spent in Step 1 of the algorithm.

\paragraph{Second Type of Updates.} 
The second type of update happens when the algorithm for Step 2 identifies a well-structured cut $(A,B)$ in $H$ with $|E_H(A,B)|\leq \frac{2^{11}\log n}{d}\cdot \min\set{|A|,|B|}$. Denote sets $A$ and $B$ by $Z$ and $Z'$, so that $Z$ is the smaller of the two sets, breaking ties arbitrarily. Recall that our algorithm then deletes the vertices of $Z$ from $H$. If, as the result of this deletion, the number of vertices in $H$ falls below $\left (1-\frac{1}{\hat c \log^5 n}\right )n'$, then the current phase terminates. If this did not happen, then we update graph $W^*$ as follows. Let $\hat E=\bigcup_{e\in E_H(A,B)}S(e)$, that is, these are all edges $e'\in E(W^*)\setminus F$ whose embedding path contains an edge of $E_H(A,B)$. Since the embedding of $W^*\setminus F$ into $H$ causes congestion at most $O(d\log^4 n)$ in $H$, we get that $|\hat E|\leq O(d\log^4 n)\cdot |E_H(A,B)|\leq O(|Z|\cdot \log^5n)$.

We start by deleting the edges of $\hat E$ from $W^*$. Note that, if $Z=A$ and $Z'=B$, then graph $W^*\setminus F$ no longer contains any edges connecting vertices of $Z\cap V(W^*)$ to vertices of $Z'\cap V(W^*)$. In other words, $(Z\cap V(W^*),Z'\cap V(W^*))$ is a cut of sparsity $0$ in $W^*\setminus F$. We replace the set $V(W^*)$ of vertices in $\xset$ by two sets: $Z\cap V(W^*)$ and $Z'\cap V(W^*)$. We also delete the vertices of $Z$ from $W^*$. Similarly, if $Z=B$ and $Z'=A$, then $(Z'\cap V(W^*),Z\cap V(W^*))$ is a cut of sparsity $0$ in $W^*\setminus F$. We replace the set $V(W^*)$ of vertices in $\xset$ by two sets: $Z\cap V(W^*)$ and $Z'\cap V(W^*)$, and we delete the vertices of $Z$ from $W^*$.

We denote by $E''$ the set of edges that were deleted from $W^*$ as the result of this type of transformation: in other words, for each sparse cut $(A,B)$ of $H$ that our algorithm computes, we add all edges of $\bigcup_{e\in E_H(A,B)}S(e)$ to $E''$. Recall that the number of edges added to $E''$ due to such cut $(A,B)$ is bounded by 
$O(\log^5n)\cdot\min\set{|A|,|B|}$, and that the vertices of $A$ or the vertices of $B$ are deleted from $H$. Since the total number of vertices that may be deleted from $H$ before the phase terminates is bounded by $\frac{n'}{\hat c\log^5 n}$, where $\hat c$ can be chosen to be a large enough constant, we can ensure that throughout the phase $|E''|\leq \frac{\alpha_0 n}{2^{20}}$ holds.

\paragraph{Third Type of Updates.}
Recall that the cleanup procedure may sometimes identify cuts $(A,B)$ in $W^*\setminus F$, with $|E_{W^*\setminus F}(A,B)|\leq \frac{\alpha_0}{2^{20}}\cdot \min\set{|A|,|B|}$. Whenever this happens, it replaces the set $V(W^*)$ of vertices in $\xset$ with $V(A)$ and $V(B)$. Denote $A$ and $B$ by $Z$ and $Z'$, where $Z$ is the smaller of the two sets, breaking ties arbitrarily. We then delete the vertices of $Z$ from graph $W^*$.
Notice that $|E_{W^*\setminus F}(A,B)|\leq \frac{\alpha_0}{2^{20}}\cdot|Z|$.

 Let $E'''$ denote the collection of all edges that participate in such cuts. In other words, for each such sparse cut $(A,B)$ of $W^*\setminus F$, we include the edges of $E_{W^*\setminus F}(A,B)$ in $E'''$. From the above discussion, throughout the algorithm, $|E'''|\leq \frac{\alpha_0}{2^{20}}\cdot n''\leq \frac{\alpha_0}{2^{20}}\cdot n$ holds.
 
 Lastly, we denote $E^*=F\cup E'\cup E''\cup E'''$. Recall that $|F|\leq \frac{\alpha_0\cdot n'}{500}\leq \frac{\alpha_0\cdot n''}{62}$, since $n''\geq n'/8$. The cardinality of each of the three sets $E',E''$ and $E'''$ is bounded by $\frac{\alpha_0\cdot n}{2^{20}}\leq \frac{\alpha_0\cdot n'}{2^{19}}\leq \frac{\alpha_0\cdot n''}{2^{14}}$.
 Overall, we get that $|E^*|\leq \frac{\alpha_0\cdot n''}{32}$ holds throughout the algorithm.
We now show that, throughout the algorithm, $|V(W^*)|\geq n''/4$ holds.

\begin{claim}\label{claim: expander remains large}
	Throughout the algorithm,  $|V(W^*)|\geq n''/4$.
\end{claim}
\begin{proof}
	Assume otherwise. Consider some time $\tau$ at which $|V(W^*)|< n''/4$ holds. For convenience, we denote by $W$ the original graph $W^*$, which is an $\alpha_0$-expander, with $n''$ vertices and we denote by $W'$ the graph $W\setminus E^*$, with respect to the set $E^*$ of edges at time $\tau$. We also consider the set $\xset$ of vertex subsets at time $\tau$. We denote the sets of $\xset$ by $X_1,X_2,\ldots,X_r$, where $X_r=V(W^*)$ at time $\tau$, and the remaining sets are indexed in the order in which they were added to $\xset$. For all $1\leq j<r$, let $X'_j=X_{j+1}\cup\cdots\cup X_r$. At the time when $X_j$ was added to $\xset$, it defined a sparse cut in graph $W^*\setminus F$. In other words, either $|E_{W'}(X_j,X'_j)|\leq \frac{\alpha_0}{2^{20}}|X_j|$, or $|E_{W'}(X'_j,X_j)|\leq \frac{\alpha_0}{2^{20}}|X_j|$ holds. Let $\alpha$ be such that $\max_{X_j\in \xset}|X_j|=\alpha\cdot n''$. Then it is easy to verify that $\alpha\leq 1/2$ must hold. From \Cref{obs: from many sparse to one balanced}, there is a cut $(A,B)$ in $W'$ with $|A|,|B|\geq n''/4$, and $|E_{W'}(A,B)|\leq \frac{\alpha_0\cdot n''}{2^{20}}$.
	
	Observe that $|E_W(A,B)|\leq |E^*|+\frac{\alpha_0\cdot n''}{2^{20}}\leq \frac{\alpha_0\cdot n''}{16}$. Since $|A|,|B|\geq n''/4$, we obtain a cut $(A,B)$ in $W$, whose sparsity is below $\alpha_0$, contradicting the fact that $W$ is an $\alpha_0$-expander.
\end{proof}

This completes the description of the data structures that we maintain over the course of Step 2. Next, we describe an algorithm for responding to a \shortpath queries.

\subsection*{Responding to \shortpath queries}
In a \shortpath query, we are given a pair $x,y$ of vertices of $H$, and our goal is to compute a simple path connecting $x$ to $y$ in $H$ of length at most $d^*$. %We note that we only need to ensure that the total time spent on responding to all $N$ queries is suitably bounded, but the time spent on responding to a specific query may be quite large. 
%The algorithm for responding to a query consists of two steps. In the first step, we compute the desired path connecting $x$ to $y$ in $H$, whose length is at most $d^*$. After that, some edges may be deleted from $H$. We then execute a second, clean-up step, whose purpose is to ensure that for every pair $a,b$ of vertices that remain in the expander $W^*$,
%$\dist_{W^*}(a,b),\dist_{W^*}(b,a)\leq \hat d$ holds, where  $\hat d=\frac{2^{30}\log n}{\alpha_0}$ is the parameter that we defined above. Additionally, we will need to ensure that every vertex of $H$ lies in both the trees $T'$ and $T''$.

In order to respond to the query, we use \EST $T''$ to compute a path $P'$ in graph $H$ of length at most $d$, connecting $x$ to some vertex $x'\in V(W^*)$, in time $O(d)$. Similarly, we use \EST $T'$ in order to compute a path $P''$ of length at most $d$, connecting some vertex $y'\in V(W^*)$ to $y$, in time $O(d)$. Next, we perform a BFS in graph $W^*\setminus F$, from vertex $x'$, in order to compute the shortest $x'$-$y'$ path $Q$ in $W^*\setminus F$. Recall that we are guaranteed that the length of $Q$ in $W^*\setminus F$ is bounded by  $\hat d=\frac{2^{30}\log n}{\alpha_0}$. We then use the embedding of $W^*\setminus F$ into $H$ via paths of lengths at most $O(d\log^2 n)$ in order to obtain a path $\hat P$ connecting $x'$ to $y'$ in $H$, whose length is at most $O(d \log^3n)$.
Let $P^*$ be the path that is obtained by concatenating  $P',\hat P$ and $P''$. Path $P^*$ connects $x$ to $y$ in graph $H$, and its length is bounded by $2d+O(d\log^3n)$. Since we have set $d=\frac{d^*}{c'\log^3 n}$, where $c'$ is a large enough constant, the length of $P^*$ is bounded by $d^*$. We transform $P^*$ into a simple path in time $O(d^*)$ and return the resulting path. This finishes the algorithm for responding to the query. Notice that its running time is bounded by $O(n+d^*)$, since a BFS search in $W^*$ can be performed in time $O(n)$, and the time required to compute paths $P'$ and $P''$ and to transform $Q$ into path $\hat P$ is bounded by $O(d^*)$. 
Since the number of \shortpath queries that we need to process in a single phase is bounded by $N=\tilde O(n/(d^*)^2)$, the total time that is required in order to respond to \shortpath queries over the course of a single phase is bounded by $\tilde O(n^2)$.
%Before continuing to the second clean-up step, we select an arbitrary vertex $v\in V(W^*)$, and compute two shortest-paths trees $\hat T'$ and $\hat T''$ rooted at $v$ in graph $W^*$, where the edges in $T'$ are directed away from $v$, and the edges in $T''$ are directed towards $v$. In other words, for every vertex $a\in V(W^*)$, the unique $s$-$a$ path in $T'$ is the shortest $s$-$a$ path in $W^*$, aand the unique $a$-$s$ path in $T''$ is the shortest $a$-$s$ path in $W^*$.

\subsection*{Final Analysis of the Running Time.}
We now bound the total running time of a single phase. Recall that the running time of Step 1 is bounded by $O(m\cdot d+n^2)\leq \otilde(md^*+n^2)$. The total time spent on executing the cleanup procedure over the course of the phase is bounded by  $\otilde(md^*+n^2)$, and the total time spent on responding to queries is bounded by $\tilde O(n^2)$. Lastly, and the time spent on maintaining \EST data strucures $T'$ and $T''$ is bounded by $\tilde O(m d)\leq \otilde(md^*)$. 
Additional time that is required to maintain the expander $W^*$ is bounded by the above running times, and so the running time of a single phase is bounded by $\tilde O(m\cdot d^*+n^2)$.

Since the number of phases is bounded by $\tilde O\left (\max\set{1,\frac{\Delta\cdot (d^*)^2}{n}}\right )$, we get that the total running time of the algorithm is bounded by:

\[\otilde(m\cdot d^*+n^2)\cdot  \max\set{1,\frac{\Delta\cdot (d^*)^2}{n}}. \]

%------------------------------------------
%------------------------------------------
%------------------------------------------
%------------------------------------------
%------------------------------------------

\subsection{Algorithm \algembedorcut: Proof of \Cref{thm: embed expander or cut}}
\label{sec: embed or cut}
We use the Cut-Matching Game -- a framework that was first proposed by \cite{KRV}. The specific variant of the game is due to \cite{KKOV}. Its slight modification, that we also use, was proposed in \cite{detbalanced}, and later adapted to directed graphs in \cite{SCC}.
We also note that \cite{directed-CMG} extended the original cut-matching game of \cite{KRV} to directed graphs, though we do not use this variation here.
We now define the variant of the game that we use.

The game is played between two players: the Cut Player and the Matching Player. 
The game starts with a directed graph $W$ that contains $n$ vertices  and no edges, where $n$ is an even integer, and proceeds in iterations. In the $i$th iteration, the Cut Player chooses a cut $(A_i,B_i)$ in $W$ with $|A_i|,|B_i|\geq n/10$, and $|E_W(A_i,B_i)|\leq n/100$. It then computes two arbitrary equal-cardinality sets $A_i',B_i'$ of vertices with $|A'_i|\geq n/4$, such that either $A_i\subseteq A'_i$, or $B_i\subseteq B'_i$. The matching player then needs to return two perfect matchings: a matching $M_i\subseteq A'_i\times B'_i$, and  a matching $M'_i\subseteq B'_i\times A'_i$, both of cardinality $|A'_i|$. The edges of $M_i\cup M'_i$ are added to $W$, and the iteration terminates. The following bound on the number of iterations in this variant of the Cut-Matching Game was shown in \cite{SCC}, extending a similar result of \cite{KKOV} for undirected graphs.

\begin{theorem}[Lemma 7.4 in \cite{SCC}; see also Theorem 7.3 and the following discussion]\label{thm: CMG bound}
The number of iterations in the above Cut-Matching Game is bounded by $O(\log n)$.
\end{theorem}

From the above theorem, there is an absolute constant $\cCMG>0$, such that, after $\cCMG\log n$ iterations of the game, for every cut $(A,B)$ in $W$ with $|A|,|B|\geq n/10$, $|E_W(A,B)|>n/100$ must hold. Notice that the maximum vertex degree in the final graph $W$ is bounded by $2\cCMG\log n$.

In order to implement the Cut-Matching Game, we design two algorithms: an algorithm for the Cut Player, and an algorithm for the Matching Player. We note that previous work, such as, e.g. \cite{SCC} also designed similar algorithms, but since we focus here on a more specific setting, we believe that our algorithm for the cut player is significally simpler, while the algorithm for the Matching Player achieves better parameters.

The algorithm for the Cut Player is summarized in the following theorem, whose proof is deferred to \Cref{subsec: algorithm for the cut player}.
The theorem uses the constant $\alpha_0$ from \Cref{thm:explicit expander}.
\begin{theorem}\label{thm: cut player outer}
	There is a deterministic algorithm, that, given an $n$-vertex directed graph $W$ with maximum vertex degree at most $O(\log n)$, returns one of the following:
	
	\begin{itemize}
		\item Either a cut $(A,B)$ in $W$ with $|E_W(A,B)|\leq n/100$ and $|A|,|B|\geq n/10$; or
		
		\item another directed graph $W'$ with $|V(W')|\geq n/4$ and maximum vertex degree at most $18$, such that $W'$ is an $\alpha_0$-expander, together with a subset $F\subseteq E(W')$ of at most $\frac{\alpha_0\cdot n}{1000}$ edges of $W'$, and an embedding $\pset$ of $W'\setminus F$ into $W$ via paths of length at most $\td=O(\log^2 n)$, that cause congestion at most $\eta=O(\log^2 n)$.
	\end{itemize}
	
	The running time of the algorithm is $\otilde(n^2)$.
\end{theorem}

We note that \cite{SCC}, generalizing the results of \cite{detbalanced} to directed graphs, provided an algorithm for the cut player with faster running time $\ohat(n)$ (see Theorem 7.3 in \cite{SCC}). We could have used their algorithm as a blackbox instead of the algorithm from \Cref{thm: cut player outer}. However, the proof of \Cref{thm: cut player outer} is significantly simpler than that of \cite{SCC} and \cite{detbalanced}, and we can afford this  slower running time since other components of our algorithm have slower running time.

The algorithm for the Matching Player is summarized in the following theorem, whose proof is deferred to \Cref{sec: matching player alg}.

\begin{theorem}\label{thm: alg for matching player}
	There is a deterministic algorithm, whose input consists of a directed $n$-vertex bipartite graph $G=(L,R,E)$ with $|E|=m$, where $n$ is an even integer, such that every vertex in $R$ has out-degree at most $1$ and every vertex in $L$ has in-degree at most $1$, together with two disjoint sets $A,B$ of vertices of $G$ of cardinality $n/2$ each, and two parameters $4\leq d'\leq 2n$ and $z\geq 1$.
		The algorithm returns one of the following:
	\begin{itemize}
		\item either a cut $(X,Y)$ in $G$ with $|E_G(X,Y)|\leq \frac{2n}{d'}$ and $|X|,|Y|\geq z$; or
		
		\item a collection $\qset$ of at least $(n/2-z)$ paths in graph $G$, with every path $Q\in \qset$ connecting some vertex of $A$ to some vertex of $B$, such that the length of every path in $\qset$ is at most $2d'+1$, the paths cause congestion at most $O(d'\log n)$, and the endpoints of all paths are distinct.
	\end{itemize}
	The running time of the algorithm is $\otilde(md')$.
\end{theorem}

Note that in the above theorem we assume that all edge lengths in the input graph $G$ are unit.

We are now ready to complete the algorithm for \Cref{thm: embed expander or cut}. We will assume that $n=|V(G)|$ is an even integer; otherwise, we temporarily delete an arbitrary vertex $v_0$ from $G$. 
We let $c$ be a large enough constant, so that,  
in \Cref{thm: cut player outer} the parameter $\eta\leq c\log^2 n$. 

 We start with a graph $W$ that contains $n$ vertices and no edges, and then run the Cut-Matching game on it. Recall that the number of iterations in the game is bounded by $\cCMG\cdot \log n$. 
This ensures that the graph $W$ that is maintained over the course of the game has maximum vertex degree $O(\log n)$.
We now describe the execution of the $i$th iteration.

First, we apply the algorithm from \Cref{thm: cut player outer} to the current graph $W$. We say that the current iteration is a type-1 iteration, if the algorithm returns 
a cut $(A_i,B_i)$ in $W$ with $|E_W(A_i,B_i)|\leq n/100$ and $|A_i|,|B_i|\geq n/10$, and otherwise we say that it is a type-2 iteration. Assume first that the current iteration is a type-1 iteration. Then we compute an arbitrary partition $(A'_i,B'_i)$ of $V(W)$ into two subsets of cardinality $n/2$ each, such that either $A_i\subseteq A'_i$ or $B_i\subseteq B'_i$, and we treat $(A'_i,B'_i)$ as the move of the cut player.

Next, we twice apply the algorithm from \Cref{thm: alg for matching player} to  graph $G$, with parameter $d'$ that remains unchanged, and parameter ${z=\frac{\alpha_0}{2000c\cdot \cCMG\cdot \log^3n}\cdot n}$: once for the ordered pair of sets $(A'_i,B'_i)$, and once for the ordered pair of sets $(B'_i,A'_i)$.

Assume first that in any of the two applications of the algorithm, it returns  a cut $(X,Y)$ in $G$ with $|E_G(X,Y)|\leq \frac{2n}{d'}$ and $|X|,|Y|\geq z\geq \Omega(n/\log^3 n)$. In this case, we terminate the algorithm and return the cut $(X,Y)$ as its outcome. 
If the initial value $|V(G)|$ was odd, then we add vertex $v_0$ to $X$ if $v_0\in R$, and to $Y$ otherwise. In either case, the only newly added edge to $E_G(X,Y)$ is the special edge incident to $v_0$ (if such an edge exists), and so $|E_G(X,Y)|\leq \frac{2n}{d'}+1$ holds.
We also say that the current iteration is of \emph{type (1a)}. 

Otherwise, the algorithm computes a collection $\qset'_i$ of at least $(n/2-z)$ paths in graph $G$, with every path $Q\in \qset_i$ connecting a distinct vertex of $A'_i$ to a distinct vertex of $B'_i$, 
and a collection $\qset''_i$ of at least $(n/2-z)$ paths in  $G$, with every path $Q'\in \qset'_i$ connecting a distinct vertex of $B'_i$ to a distinct vertex of $A'_i$. The paths in $\qset'_i\cup \qset''_i$ have lengths at most $2d'+1$ each, and they cause congestion at most $O(d'\log n)$.
In this case, we say that the current iteration is of \emph{type (1b)}.
 The set $\qset'_i$ of paths naturally defines a matching $M'_i\subseteq A'_i\times B'_i$, where for every path $Q\in \qset_i$ whose endpoints are $a\in A'_i,b\in B'_i$, we add an edge $e=(a,b)$ to the matching. We also let $P(e)=Q$ be its embedding into $G$. Similarly, the set $\qset''_i$ of paths defines a matching $M''_i\subseteq B'_i\times A'_i$, together with an embedding $P(e)\in \qset''_i$ of each edge $e\in M''$. Recall that matching $M'_i$ has cardinality at least $n/2-z$, while $|A'|=|B'|=n/2$. Therefore, by adding a set $F'_i$ of at most $z$ additional edges (that we refer to as \emph{fake edges}), we can turn matching $M'_i$ into a perfect one. Similarly, we add a collection $F''_i$ of at most $z$ fake edges to $M''_i$ in order to turn it into a perfect matching. The edges of $M'_i\cup M''_i$ are then added to graph $W$, and the current iteration terminates.

Finally, assume that the current iteration $i$ is of type 2. In this case, the algorithm for the Cut Player from \Cref{thm: cut player outer} computed a graph 
$W'$ with $|V(W')|\geq n/4$ and maximum vertex degree at most $18$, such that $W'$ is an $\alpha_0$-expander, together with a subset $\hat F\subseteq E(W')$ of at most $\frac{\alpha_0\cdot n}{1000}$ edges of $W'$, and an embedding $\pset$ of $W'\setminus \hat F$ into $W$ via paths of length at most $O(\log^2 n)$, that cause congestion at most $c\log^2 n$.

Consider the current graph $W$. Denote $F^*=\bigcup_{i'=1}^{i-1}(F'_{i'}\cup F''_{i'})$ -- the collection of the fake edges that we have computed so far. Since for all $1\leq i'<i$, $|F'_{i'}\cup F''_{i'}|\leq 2z$ and $i\leq \cCMG\log n$, we get that $|F^*|\leq 2\cCMG z \log n$.
We also denote $\qset^*=\bigcup_{i'=1}^{i-1}(\qset'_i\cup \qset''_i)$. Note that $\qset^*$ defines an embedding of $W\setminus F^*$ into $G$. All paths in $\qset^*$ have lengths at most $2d'+1$. Each set $\qset'_i\cup \qset''_i$ causes congestion at most $O(d'\log n)$, and so the total congestion caused by the paths in $\qset^*$ is bounded by $\eta^*=O( d'\log^2n)$.

Next, we will construct a set $F\subseteq E(W')$ of fake edges, and an embedding of the remaining edges of $W'$ into $G$, by combining the embedding $\pset$ of $W'\setminus \hat F$ into $W$ with the embedding $\qset^*$ of $W\setminus F$ into $G$. We start with $F=\hat F$. Next, we process every edge $e\in E(W')\setminus \hat F$ one by one. When an edge $e=(u,v)$ is processed, we consider the path $P(e)\in \pset$, that serves as an embedding of $e$ into graph $W$. Recall that the length of the path is at most $O(\log^2 n)$. Let $P(e)=(e_1,e_2,\ldots,e_r)$, where $r\leq O(\log^2 n)$. 
If path $P(e)$ contains any edges of $F^*$, then we add $e$ to the set $F$ of fake edges. Otherwise,
for all $1\leq j\leq r$, we consider the path $Q(e_j)\in \qset^*$, that serves as the embedding of the edge $e_j$ into graph $G$. Recall that the length of $Q$ is at most $O(d')$. By concatenating the paths $Q(e_1),\ldots,Q(e_r)$, we obtain a path $P'(e)$ connecting $u$ to $v$ in $G$, whose length is at most $O(d'\log^2 n)$.  We turn $P'(e)$ into a simple path in time $O(d'\log^2 n)$, obtaining the final path $P^*(e)$, that will serve as an embedding of the edge $e$.

Once every edge of $E(W')\setminus \hat F$ is processed, we obtain a final set $F$ of fake edges, and, for every remaining edge $e\in E(W')\setminus F$, a simple path $P^*(e)$ of length $O(d'\log^2 n)$, that serves as the embedding of $e$ in $G$. Recall that $|V(W')|\geq n/4$. If the original value $|V(G)|$ was odd, then $|V(W')|\geq (|V(G)|-1)/4\geq |V(G)|/8$ holds. As required, graph $W'$ is an $\alpha_0$-expander with maximum vertex degree at most $18$. It remains to bound the congestion caused by the paths in set $\pset^*=\set{P^*(e)\mid e\in E(W')\setminus F}$, and the cardinality of the set $F$.

Consider any edge $e\in E(G)$. Recall that $e$ may participate in at most $\eta^*= O(d'\log^2n)$ paths in set $\qset^*=\set{Q(e')\mid e'\in E(W)\setminus F^*}$. Every edge $e'\in E(W)\setminus F^*$, whose corresponding path $Q(e')$ contains the edge $e$, may in turn participate in at most $O(\log^2 n)$ paths of $\pset$. Therefore, overall, every edge $e$ of $G$ participates in at most $O(d'\log^4n)$ paths of $\pset^*$.

In order to bound the final number of the fake edges in $W'$,
 recall that $|F^*|\leq 2\cCMG z \log n$. Since the paths in $\pset$ cause congestion at most $c\log^2 n$, we get that every edge of $F^*$ may participate in at most $c\log^2 n$ paths of $\pset$. Therefore, $|F\setminus \hat F|\leq (2\cCMG z \log n)\cdot (c\log^2 n)\leq 2c\cdot \cCMG\cdot z\cdot \log^3n\leq \frac{\alpha_0\cdot n}{1000}$,
since $z=\frac{\alpha_0}{2000c\cdot \cCMG\cdot \log^3n}\cdot n$.  
Since $|\hat F|\leq \frac{\alpha_0\cdot n}{1000}$, we get that $|F|\leq \frac{\alpha_0\cdot n}{500}$.

It now remains to bound the running time of the algorithm. The algorithm has $O(\log n)$ iterations. In every iteration we use the algorithm for the Cut Player from \Cref{thm: cut player outer}, whose running time is  $\otilde(n^2)$, and an algorithm for the Matching Player from \Cref{thm: alg for matching player}, whose running time is $\otilde(md')$.

Additionally, when computing the embedding of $W'\setminus F$ into $G$, we spend time at most $\tilde O(d')$ per edge of $W'$, and, since $|E(W')|\le O(n)$, the computation of the embedding takes time $\tilde O(nd')$.
Overall, the running time of the algorithm is $\otilde(md'+n^2)$.
\subsection{Algorithm for the Cut Player -- Proof of \Cref{thm: cut player outer}}
\label{subsec: algorithm for the cut player}

The proof of \Cref{thm: cut player outer} easily follows from the following slightly weaker theorem.

\begin{theorem}\label{thm: cut player inner}
	There is a deterministic algorithm, that, given an $n$-vertex directed graph $W$ with maximum vertex degree at most $O(\log n)$, returns one of the following:
	
	\begin{itemize}
		\item either a cut $(A,B)$ in $W$ with $|E_W(A,B)|\leq  \frac{\min\set{|A|,|B|}}{128}$ and $|A|,|B|\geq \frac{\alpha_0}{2^{14}}\cdot n$; or
		
		\item another directed $n$-vertex graph $W'$ with maximum vertex degree at most $18$, such that $W'$ is an $\alpha_0$-expander, together with a subset $F\subseteq E(W')$ of at most $k=\frac{\alpha_0\cdot n}{1000}$ edges of $W'$, and an embedding $\pset$ of $W'\setminus F$ into $W$ via paths of length at most $\td=O(\log^2 n)$, that cause congestion at most $\eta=O(\log^2 n)$.
	\end{itemize}
	
	The running time of the algorithm is $\otilde(n^2)$.
\end{theorem}

We prove \Cref{thm: cut player inner} below, after we complete the proof of \Cref{thm: cut player outer} using it. We start with the graph $\hat W=W$, and then perform iterations, as long as $|V(\hat W)|\geq n/4$. We will also construct a collection $\xset$ of subsets of vertices of $\hat W$, that is initialized to $\emptyset$. In every iteration, we may add a single set $X$ to $\xset$ (that defines a sparse cut in the current graph $\hat W$), and then delete the vertices of $X$ from $\hat W$.

In order to execute a single iteration, we apply the algorithm from \Cref{thm: cut player inner} to the current graph $\hat W$. Let $n'\geq n/4$ denote the number of vertices that currently belong to $\hat W$. If the algorithm returned a directed $n'$-vertex  $\alpha_0$-expander graph $W'$ with maximum vertex degree at most $18$, together with a subset $F\subseteq E(W')$ of at most $k=\frac{\alpha_0\cdot n'}{1000}\leq \frac{\alpha_0 n}{1000}$ edges of $W'$, and an embedding $\pset$ of $W'\setminus F$ into $\hat W$ via paths of length at most $\td$, that cause congestion at most $\eta$, then we terminate the algorithm and return $W'$, $F$ and $\pset$. 

Assume now that the algorithm from \Cref{thm: cut player inner}  returned a cut $(A',B')$ in $\hat W$ with $|E_{\hat W}(A',B')|\leq  \frac{\min\set{|A'|,|B'|}}{128}$ and $|A'|,|B'|\geq \frac{\alpha_0}{2^{14}}\cdot n'\geq \frac{\alpha_0}{2^{16}}n$.
Denote the sets $A',B'$ by $X$ and $Y$ respectively, where $X$ is the set among $\set{A',B'}$ that is of smaller cardinality, and $Y$ is the other set (if the cardinalities of the two sets are equal, break the tie arbitrarily). We then add $X$ to $\xset$, delete the vertices of $X$ from $\hat W$, and continue to the next iteration. 

If the above algorithm does not terminate with an expander $W'$ that is embedded into $W$, then it must terminate once $|V(\hat W)|\leq n/4$ holds. Let $\xset=\set{X_1,X_2,\ldots,X_r}$ be the final collection of vertex subsets, where the sets are indexed in the order in which they were added to $\xset$. We also denote by $X_{r+1}$ the set $V(\hat W)$ obtained at the end of the algorithm, and we add $X_{r+1}$ to $\xset$ as well. For all $1\leq j\leq r$, denote $X'_j=X_{j+1}\cup \cdots\cup X_{r+1}$. From our algorithm, for each such index $j$, either $|E_W(X_j,X'_j)|\leq  |X_j|/128$, or $|E_W(X'_j,X_j)|\leq |X_j|/128$. Moreover, for all $1\leq j\leq r+1$, $|X_j|\leq n/2$, while $|X_{r+1}|\geq \frac{n}{8}$. Therefore, if we let $\alpha$ be such that $\max_{X_j\in \xset}\set{|X_j|}=\alpha\cdot n$, then $\frac 1 8\leq \alpha\leq \frac 1 2$.

We can now use the algorithm from \Cref{obs: from many sparse to one balanced} to compute, in time $O(|E(W)|)\leq \otilde(n)$, a cut $(A,B)$ in $W$, with   $|E_{W}(A,B)|\leq n/128$, and  $|A|,|B|\geq \min \set{\frac{1-\alpha}{2}\cdot n, \frac n 4}\geq \frac{n}{10}$.

The running time of every iteration of the algorithm is $\otilde(n^2)$, and the number of iterations is bounded by a constant. The running time of the algorithm from \Cref{obs: from many sparse to one balanced}  is bounded by $\otilde(n)$, and so the total running time of our algorithm is $\otilde(n^2)$.

In order to complete the proof of \cref{thm: cut player outer}, it is now enough to prove \Cref{thm: cut player inner}.

\begin{proofof}{\Cref{thm: cut player inner}}
	Since we can afford a running time that is as high as $\otilde(n^2)$, we can use a straightforward algorithm, that is significantly simpler than other current deterministic algorithms for the cut player (see e.g. Theorem 7.3 in \cite{SCC}, that generalizes the results of \cite{detbalanced} to directed graphs). 
	
We let $c$ be a sufficiently large constant, and we set the parameters
$\td=c\log^2n$ and $\eta=c^2\log^2n$.

We start by using the algorithm from \Cref{thm:explicit expander} to construct, in time $O(n)$, a directed  $n$-vertex graph $W'$ that is an $\alpha_0$-expander, with maximum vertex degree at most $18$.

Next, we start with a graph $\hat W=W$, and then perform $|E(W')|$ iterations. In every iteration we will attempt to embed some edge $e\in E(W')$ into $\hat W$, and we will delete from $\hat W$ edges that participate in too many of the embedding paths. We will maintain a partition $(E',E'',F)$ of the edges of $E(W')$, where set $E'$ contains all edges $e\in E'$ for which an embedding path $P(e)$ was already computed by the algorithm; set $F$ contains all edges $e\in E'$ for which the algorithm failed to compute an embedding path; and $E''$ contains all edges of $W'$ that were not considered by the algorithm yet. We also maintain a set $\pset=\set{P(e)\mid e\in E'}$ of all embedding paths that the algorithm computed so far. For every edge $e\in E(W)$ we will maintain a counter $\mu(e)$ that counts the number of paths in $\pset$ containing $e$.

Initially, we set $E'=F=\emptyset$, $E''=E(W')$, and $\pset=\emptyset$.  We also set $\hat W=W$, and perform iterations as long as $E''\neq \emptyset$.

In order to perform an iteration, we consider any edge $e=(u,v)\in E''$. We compute the shortest path $P$ in the current graph $\hat W$ connecting $u$ to $v$, by peforming a BFS from vertex $u$, in time $O(|E(\hat W)|)\leq \otilde(n)$. If the length of $P$ is greater than $\td$, then we move edge $e$ from $E''$ to $F$. Otherwise, we set $P(e)=P$, move $e$ from $E''$ to $E'$, and add path $P(e)$ to set $\pset$. Additionally, for every edge $e'\in E(P)$ we increase the counter $\mu(e')$ by $1$, and if the counter reaches $\eta$, we delete $e'$ from $\hat W$.

Once $E''=\emptyset$ holds, the algorithm terminates. 
Notice that, since the algorithm spends at most $\otilde(n)$ time on processing every edge of $W'$, its running time so far is $\otilde(n^2)$.

We now consider two cases. Assume first that, at the end of the above procedure, $|F|\leq k$ holds. In this case, we terminate the algorithm, and return the $\alpha_0$-expander $W'$, the set $F$ of its edges, and the embedding $\pset$ of $W'\setminus F$ into $W$ that our algorithm has computed. It is easy to verify that every path in $\pset$ has length at most $\td$, and $\pset$ causes congestion at most $\eta$ in $W$.

From now on we assume that $|F|>k$. We now show an algorithm to compute a cut $(A,B)$ in $W$ of sparsity at most $\frac 1 {128}$, such that $|A|,|B|\geq \frac{\alpha_0}{2^{14}}\cdot n$.

Let $E^*=E(W)\setminus E(\hat W)$. Equivalently, set $E^*$ contains all edges that participate in $\eta$ paths in $\pset$. Recall that $|\pset|\leq |E(W')|\leq 9n$, and the length of every path in $\pset$ is at most $\td$. Therefore, $\sum_{P\in \pset}|E(P)|\leq 9\td n$, and $|E^*|\leq\frac{\sum_{P\in \pset}|E(P)|}{\eta}\leq \frac{9\td n}{\eta}$.

In the remainder of the algorithm, we will iteratively cut off sparse cuts (that may be quite small) from graph $\hat W$. We maintain a set $\xset$ of subsets of vertices of $\hat W$, that is initially empty. The algorithm is iterative, and in every iteration a single set $Z$ is added to $\xset$, with the vertices of $Z$ deleted from $\hat W$. The algorithm continues as long as $|V(\hat W)|\geq n-k/2$ holds.
We now describe the execution of a single iteration.

Since $|V(\hat W)|\geq n-k/2$, there exists an edge $e=(x,y)\in F$ with $x,y\in \hat W$. Moreover, since we were unable to embed $e$ into $\hat W$ when it was considered, and since we only deleted edges and vertices from $\hat W$ since then, it must be the case that $\dist_{\hat W}(x,y)>\td $ holds right now. Let $\Delta=O(\log n)$ denote the maximum vertex degree in $\hat W$. Since $\td=c\log^2 n$, where $c$ is a large enough constant, we can ensure  that $\td\geq 64\Delta\log n$ holds. We now use the algorithm from \Cref{obs: ball growing} to compute a cut $(Z,Z')$ in the current graph $\hat W$, whose sparsity is at most $\phi'=\frac{32\Delta\log n}{\td}\leq \frac{O(\log^2 n)}{\td}$. Since $\td=c\log^2n$, and we can set $c$ to be a large enough constant, we can ensure that $\phi'\leq \frac{\alpha_0}{2^{22}}$. Denote the sets $Z,Z'$ by $X$ and $Y$, such that $X$ is the set of smaller cardinality, and $Y$ is the other set (if the cardinalities of the two sets are equal, break the tie arbitrarily). We then add $X$ to $\xset$, delete the vertices of $X$ from $\hat W$, and continue to the next iteration. Note that, if $X$ is the set of vertices that was added to $\xset$ during the current iteration, then from \Cref{obs: ball growing}, the running time of the iteration is bounded by $\tilde O(|X|)$.

When the above algorithm terminates, we let $\xset=\set{X_1,X_2,\ldots,X_r}$ be the final collection of vertex subsets, where the sets are indexed in the order in which they were added to $\xset$. We also denote by $X_{r+1}$ the set $V(\hat W)$ obtained at the end of the algorithm, and we add $X_{r+1}$ to $\xset$ as well. For all $1\leq j\leq r$, denote $X'_j=X_{j+1}\cup \cdots\cup X_{r+1}$. From our algorithm, for each such index $j$, either $|E_{\hat W}(X_j,X'_j)|\leq \phi'\cdot |X_j|$, or $|E_{\hat W}(X'_j,X_j)|\leq \phi'\cdot |X_j|$, where $\phi'\leq \frac{\alpha_0}{2^{22}}$. 
From our discussion, the total time that the algorithm spent on constructing the sets of vertices in $\xset$ is bounded by $\sum_{j=1}^r\tilde O(|X_j|)\leq \tilde O(n)$.
Let $\alpha$ be chosen such that the cardinality of largest set of vertices in $\xset$ is $\alpha\cdot n$. From our algorithm, $\alpha\leq 1-\frac{k}{2n}\leq 1-\frac{\alpha_0}{2000}$, since $k=\frac{\alpha_0\cdot n}{1000}$.

We can now use the algorithm from \Cref{obs: from many sparse to one balanced} to compute, in time $O(|E(\hat W)|)\leq \otilde(n)$, a cut $(A,B)$ in $\hat W$, with   $|E_{\hat W}(A,B)|\leq \phi'\cdot \sum_{i=1}^{r}|X_i|$, and  $|A|,|B|\geq \min \set{\frac{1-\alpha}{2}\cdot n, \frac n 4}\geq \frac{\alpha_0n}{4000}$.

Recall that graph $\hat W$ was obtained from $W$ by deleting the edges of $E^*$ from it, and that $|E^*|\leq \frac{9\td n}{\eta}$. Subtituting $\td=c\log^2 n$ and  $\eta=c^2\log^2 n$, where $c$ is a sufficiently large constant, we get that:
$|E^*|\leq \frac{\alpha_0 n}{2^{22}}$.

Overall:

\[|E_W(A,B)|\leq |E_{\hat W}(A,B)|+|E^*|\leq \phi' n+\frac{\alpha_0 n}{2^{22}}\leq \frac{\alpha_0 n}{2^{21}}.\]

(we have used the fact that $\phi'\leq \frac{\alpha_0}{2^{22}}$). Lastly, since $|A|,|B|\geq \frac{\alpha_0n}{4000}\geq \frac{\alpha_0n}{2^{14}}$, we get that $ |E_W(A,B)|\leq \frac{\min\set{|A|,|B|}}{128}$, as required.

From our discussion, the running time of the entire algorithm is bounded by $\otilde(n^2)$.
\end{proofof}
\subsection{Algorithm for the Matching Player -- Proof of \Cref{thm: alg for matching player}}
\label{sec: matching player alg}

As before, we call the edges of $E_G(R,L)$ \emph{special edges}, and the edges of $E_G(L,R)$ \emph{regular edges}. Notice that every vertex of $G$ may be incident to at most one special edge.

As a preprocessing step, we construct a maximal set $\qset_1$ of disjoint paths, where every path connects a distinct vertex of $A$ to a distinct vertex of $B$, and consists of a single regular edge. We do so by employing a simple greedy algorithm. We start with $\qset_1=\emptyset$. While there is a regular edge $e=(a,b)$ with $a\in A$ and $b\in B$, we add a path $Q=(e)$ to $\qset_1$, and we delete $a$ and $b$ from $A$ and $B$, respectively. Clearly, this preprocessing step can be executed in time $O(|E(G)|)$, and at the end of this step there is no regular edge in $G$ connecting a vertex of $A$ to a vertex of $B$.

Let $\hat G$ be the graph that is obtained from $G$ by replacing every edge $e$ with $d'$ parallel copies. Edges that are copies of special edges in $G$ are called special edges of $\hat G$, and similarly copies of regular edges in $G$ are called regular edges of $\hat G$. 

 At the beginning of the algorithm, every special edge $e\in E(\hat G)$ is assigned a length $\ell(e)=1/d'$, while every regular edge $e\in E(\hat G)$ is assigned length $0$. 
 Clearly, at the beginning of the algorithm, $\sum_{e\in E(\hat G)}\ell(e)\leq n$ holds, as there are at most $nd'$ special edges in $\hat G$, each of which has length $1/d'$.
 
 Over the course of the algorithm, the lengths of special edges may grow, while the lengths of regular edges remain $0$. Additionally, vertices may be deleted from sets $A$ and $B$. We will gradually construct a collection $\qset_2$ of paths in graph $\hat G$, where every path connects a distinct vertex of $A$ to a distinct vertex of $B$; the number of special edges on every path is at most $d'$; and the paths cause edge-congestion at most $O(\log n)$ in $\hat G$. Initially, we set $\qset_2=\emptyset$.

 We assume for now that we are given an oracle, that, when prompted, either returns a simple path $Q$ in the current graph $\hat G$ connecting some vertex of $A$ to some vertex of $B$, such that the length $\sum_{e\in E(Q)}\ell(e)$ of the path is at most $1$; or correctly establishes that 
such a path does not exist.

The algorithm proceeds in iterations. In every iteration, the oracle is asked to return a simple path $Q$ in the current graph $\hat G$ connecting some vertex of $A$ to some vertex of $B$, such that the length $\sum_{e\in E(Q)}\ell(e)$ of the path is at most $1$. If it successfully does so, then we add $Q$ to $\qset_2$, double the length of every special edge on path $Q$, remove the endpoints of the path from sets $A$ and $B$, respectively, and continue to the next iteration. Notice that in this case, $|\qset_2|$ grows by $1$, while $\sum_{e\in E(\hat G)}\ell(e)$ increases by at most $1$.
If the oracle establishes that no such path $Q$ exists, then we terminate the algorithm.
 
At the end of the algorithm, we set $\qset=\qset_1\cup \qset_2$.
If $|\qset|\geq n/2-z$, then we return the collection $\qset$ of paths. It is immediate to verify that the paths in $\qset$ connect vertices of $A$ to vertices of $B$, and that their endpoints are disjoint. Moreover, since the length of every special edge in $\hat G$ is at least $\frac{1}{d'}$, while the length of every path $Q\in \qset$ is at most $1$, each path in $\qset$ contains at most $d'$ special edges, and at most $2d'+1$ edges overall, since special and regular edges must alternate on every path. Lastly, due to the doubling of edge lengths on the paths that we add to $\qset$, every special edge of $\hat G$ may participate in at most $\log d' = O(\log n)$ paths in $\qset$, and so in total every special edge of $G$ may participate in at most $O(d'\log n)$ paths in $\qset$. Moreover, if $e=(u,v)$ is a regular edge, then every path $Q\in \qset_2$ that contains $e$ must also contain either the unique special edge incident to $u$, or the unique special edge incident to $v$, while the paths in $\qset_1$ are disjoint. We conclude that the paths in $\qset$ cause congestion at most $O(d'\log n)$ in $G$.

From now on we assume that, when the algorithm terminates, $|\qset|<n/2-z$ holds. Let $A'\subseteq A$, $B'\subseteq B$ be the sets of vertices that remain in $A$ and $B$, respectively at the end of the algorithm. Since the algorithm terminated, in the current graph $\hat G$, the length of the shortest path connecting a vertex of $A'$ to a vertex of $B'$ is at least $1$. Since $|\qset|<n/2-z$ and $|A'|,|B'|=n/2-|\qset|$, we get that $|A'|,|B'|\geq z$.

Recall that, at the beginning of the algorithm, $\sum_{e\in E(\hat G)}\ell(e)\leq n$ held. Over the course of the algorithm, in every iteration, $\sum_{e\in E(\hat G)}\ell(e)$ grew by at most $1$, and the number of iterations is bounded by $|A|\leq n$. Therefore,  at the end of the algorithm, $\sum_{e\in E(\hat G)}\ell(e)\leq 2n$ holds. 

Let $e$ be any edge of graph $G$, and let $S(e)$ be the collection of $d$ parallel edges representing $e$ in graph $\hat G$. Let $e'\in S(e)$ be the copy of $e$ with smallest current length $\ell(e')$. Notice that, if we set the lengths of all edges in $S(e)$ to $\ell(e')$, then the length of the shortest path connecting a vertex of $A'$ to a vertex of $B'$ in $\hat G$ will remain at least $1$. Therefore, we can assume without loss of generality, that, for every edge $e\in E(G)$, the lengths of all edges in set $S(e)$ are identical; we denote this length by $\ell'(e)$.
Notice that $\sum_{e\in E(\hat G)}\ell(e)\leq 2n$ continues to hold after this transformation.

Consider now the graph $G'$ that is obtained from $G$ by
setting the length of every edge $e\in E(G)$ to $\ell'(e)$,
 unifying all vertices of $A'$ into a source vertex $s$ and all vertices of $B'$ into a destination vertex $t$. Since graph $\hat G$ contained $d'$ parallel copies of every edge $e\in E(G)$, that all had the same length equal to $\ell'(e)$, we now have that $\sum_{e\in E(G)}\ell'(e)\leq \frac{\sum_{e\in E(\hat G)}\ell(e)}{d'}\leq \frac{2n}{d'}$. At the same time, since $\hat G$ did not contain an $A'$-$B'$ path of length at most $1$, we get that $\dist_{G'}(s,t)\geq 1$.
 
 Consider the following thought experiment: we choose a threshold $\rho\in (0,1)$ uniformly at random, and then define two sets $S=\set{v\in V(G')\mid \dist_{G'}(s,v)\leq \rho}$ and $T=V(G')\setminus S$ of vertices. Clearly, $s\in S$ and $t\in T$ must hold.
 By replacing $s$ with the vertices of $A'$ in $S$ and replacing $t$ with the vertices of $B'$ in $T$, we obtain a cut $(S',T')$ in graph $G$ with $|E_G(S',T')|=|E_{G'}(S,T)|$. It is easy to verify that, for every edge $e=(u,v)\in E(G')$, the probability that $u\in S$ and $v\in T$ is bounded by $\ell'(e)$. Therefore, the expected value of $|E_{G'}(S,T)|$ is bounded by $\sum_{e\in E(G)}\ell'(e)\leq \frac{2n}{d'}$. By performing a BFS search from vertex $s$, we can then compute a cut $(S',T')$ in $G'$ with $|E_{G'}(S',T')|\leq \frac{2n}{d'}$. Using the procedure outlined above, we can convert this cut into a cut $(S,T)$ in $G$ with $|E_G(S,T)|\leq \frac{2n}{d'}$. Since $A'\subseteq S$ and $B'\subseteq T$, we are guaranteed that $|S|,|T|\geq z$. We return the cut $(S,T)$ and terminate the algorithm. In order to analyze the algorithm we need to provide an implementation of the oracle, which we do next.

 Notice that the paths in $\qset$ cause congestion at most $O(d'\log n)$ in $G$, so $\sum_{Q\in \qset}|E(Q)|)\leq \tilde O(md')$. We will use this fact later.

 \paragraph{Implementing the Oracle.}
 In order to implement the oracle efficiently, we construct a new graph $H$, as follows. Initially, $H$ is obtained from $G$ by adding a source vertex $s$ that connects with an edge to every vertex of $A$, and a destination vertex $t$, to which every vertex of $B$ connects. The edges that leave $s$ or enter $t$ are called regular edges.
 
 Let $q=\ceil{\log d'}$.
 Consider any special edge $e=(u,v)\in E(G)$. We replace $e$ with $q$ parallel edges $e_1,\ldots,e_q$, whose lengths are $1,2,4,\ldots,2^q$, respectively. 
 Additionally, for every regular edge $e'=(x,u)$ that enters $u$, we replace $e'$ with $q$ parallel edges $e'_1,\ldots,e'_q$, whose lengths are $1,2,4,\ldots,2^q$, respectively (the latter transformation is also applied to the regular edge $(s,u)$, if such an edge is present). The latter step is done in order to ensure that the length of every edge in $G$ is greater than $0$, since this is a requirement in \Cref{thm: directed weighted ESTree} that we will use. 
If $(x,y)$ is a regular edge with $y\neq t$, such that no special edge leaves $y$, then we set $\ell(y)=1$. 
%If $(x,y)$ is a regular edge of $G$ with $y\neq t$, such that no special edge leaves $y$, then we delete $y$ and its incident edges from $G$, since vertex $y$ does not lie on any $s$-$t$ path.
 So far we have assigned lengths to all edges of $H$, except for some edges that leave $s$, and all edges that enter $t$. Initially, for every edge $e$ that leaves $s$, we set the length of $e$ to be $0$. If $e=(v,t)$ is an edge that enters $t$, then we set its length to be $1$ if $v\in L$ and to $0$ otherwise. We then increase the length of every edge that is incident to $s$ or $t$ by $1$, thus obtaining the final lengths of all edges in $H$.

 Over the course of the algorithm, for every special edge $e=(u,v)\in E(G)$, for every index $1\leq j\leq q$, we monitor the number of copies of $e$ in graph $\hat G$, whose length is $\frac{2^j}{d'}$. Once this number becomes $0$, the corresponding copy $e_j$ of $e$ is deleted from graph $H$. Additionally, for every regular edge $e'=(x,u)$, we also delete the corresponding copy $e'_j$ from $H$. Lastly, whenever a vertex $a$ is deleted from set $A$, we delete edge $(s,a)$ from $H$, and whenever a vertex $b$ is deleted from set $B$, we delete edge $(b,t)$ from $H$. Therefore, the only types of updates that graph $H$ undergoes is the deletion of edges.

  Throughout the algorithm, we will ensure that, if $Q$ is an $s$-$t$ path that the oracle returns, then exactly one vertex of $Q$ lies in $B$ -- the penultimate vertex on the path. This ensures that, throughout the algorithm, for every vertex $v$ that currently lies in $B$, for every edge $e\in E(G)$ that enters $v$, a copy of $e$ whose length is $1$ is present in $H$.

 Consider now an $s$-$t$ path $Q$ in $H$ at any time during the algorithm. From the definition of $H$, the penultimate vertex on $Q$ must lie in $B$; we denote that vertex by $b$. We assume that $Q$ is a simple path, that $b$ is the only vertex of $Q$ that lies in $B$, and that, for every edge $e\in E(Q)$, $e$ is the shortest-length edge among its parallel edges in $H$. Note that the first edge on $Q$ must be regular, while the second edge may be either regular or special. Similarly, the last edge on $Q$ must be regular, while the penultimate edge may be regular or special. We partition $Q$ into three subpaths, $Q_1,Q_2$, and $Q_3$, as follows. Path $Q_1$ is the longest subpath of $Q$ that starts with $s$ and only contains regular edges. Since no regular edge connects a vertex of $A$ to a vertex of $B$, it is easy to verify that $Q_1$ contains either one or two edges. Similarly, we let $Q_3$ be the longest subpath of $Q$ that terminates at $t$ and only contains regular edges. As before, path $Q_3$ contains either one or two edges. We let $Q_2$ be the subpath of $Q$ that starts at the last endpoint of $Q_1$, and terminates at the first endpoint of $Q_3$. From our definition, the first and the last edges on $Q_2$ must be special edges, and, from the properties of graph $G$, special and regular edges alternate on $Q_2$. Therefore, we can denote the sequence of the edges on path $Q_2$ by $(e_1,e'_2,e_2,\ldots,e'_z,e_z)$, where $e_1,e_2,\ldots,e_z$ are special edges, and $e'_2,e'_3,\ldots,e'_z$ are regular edges. For all $1\leq i\leq z$, let $e^*_i$ denote the edge of $\hat G$ corresponding to $e_i$. Then for all $1<i\leq z$, the lengths of the edges $e'_i$ and $e_i$ in $H$ are both equal to $\ell(e^*_i)\cdot d'$, where $\ell(e^*_i)$ is the length of $e^*_i$ in $\hat G$. Therefore, the length of $Q_2$ in $H$ is $\ell(e^*_1)\cdot d'+\sum_{i=2}^z2\ell(e^*_i)\cdot d'$. Consider now the path $Q_1$. If this path consists of a single edge (that must be incident to $s$), then, from our construction, the length of that edge is $\ell(e^*_1)\cdot d'+1$. Otherwise, the first edge on $Q_1$ has length $1$, and the second has length $\ell(e^*_1)\cdot d'$. In either case, the length of $Q_1$ in $H$ is $\ell(e^*_1)\cdot d'+1$. Lastly, consider the path $Q_3$. 
 Notice that the first endpoint of $Q_3$, that we denote by $v$, has a special edge entering it, and so $v$ lies in the set $L$ of the bipartition. 
 If path $Q_3$ consists of a single edge, then that edge connects a vertex of $L$ to $t$, so, from our construction, its length is $2$. Assume now that path $Q_3$ consists of two edges, and  denote the sequence of the vertices on $Q_3$ by $(v,b,t)$. Then $b\in R$ must hold, and so the length of the edge $(b,t)$ is $1$. Additionally, since $b\in B$, the length of the edge $(v,b)$ is $1$ as well. Therefore, in either case, the length of $Q_3$ is $2$. Altogether, we get that the length of $Q$ is $2d'\cdot\sum_{i=1}^z\ell(e^*_i)+3$. 
 
 It is now easy to verify that, throughout the algorithm, if $Q$ is an $s$-$t$ path in $H$, whose length is $\delta+3$, and exactly one vertex of $Q$ lies in $B$, then the length of the corresponding path $Q'$ in $\hat G$ (where for every edge $e\in E(Q)$, we choose its cheapest copy in $\hat G$ to participate in the path) is at most $\frac{\delta}{2d'}$. Similarly, if $Q'$ is an $s$-$t$ path in graph $\hat G$, that contains exactly one vertex of $B$, whose length is $\delta'$, then the length of the corresponding $s$-$t$ path in graph $H$ (where again for every edge $e\in E(Q')$ we use its cheapest copy available in $H$) is at most $2\delta'\cdot d'+3$.
 
 It is now enough to design an algorithm for the following problem. We are given a graph $H$, with two special vertices $s$ and $t$, and integral positive edge lengths, together with a parameter $d'>0$. The algorithm consists of a number of iterations. In every iteration, we need to either return the shortest $s$-$t$ path in the current graph $H$, provided that its  length is at most $2d'+2$; or correctly establish that the length of the shortest $s$-$t$ path in $H$ is greater than $2d'+2$. In the latter case, the algorithm terminates, while in the former case, the algorithm is given a batch of edge deletions from graph $H$, together with the deletion of the penultimate vertex on $Q$ from $B$, and we proceed to the next iteration.
 It is easy to verify that, if $Q$ is the shortest $s$-$t$ path in $H$, then it may contain at most one vertex that lies in $B$ -- the penultimate vertex on $Q$. Otherwise, if another vertex $b'\in B$ lies in $Q$, then the subpath of $Q$ between $s$ and $b'$, combined with the edge $(b',t)$, provides a shorter $s$-$t$ path in $H$.

We apply the algorithm from \Cref{thm: directed weighted ESTree} to graph $H$ and parameter $2d'+3$. Recall that the algorithm maintains an \EST for $H$, rooted in $s$, up to depth $2d'+3$. We denote the corresponding tree by $T$. In every iteration, we check whether vertex $t$ lies in the tree $T$, and if so, whether $\dist_H(s,t)\leq 2d'+3$. If this is not the case, then we terminate the algorithm. Otherwise, we use the tree $T$ to compute a shortest $s$-$t$ path $P$  in $H$, whose length must be at most $2d'+3$, in time $O(|E(P)|)$, by simply retracing the unique $s$-$t$ path in the tree $T$. 
By deleting the first and the last vertices from $P$, we obtain a path connecting a vertex of $A$ to a vertex of $B$ in $\hat G$, whose length is at most $1$.

It is easy to verify that the total time needed to implement the oracle is $O(|E(H)|\cdot d'+\sum_{Q\in \qset}|E(Q)|)$, where $\qset$ is the set containing all paths that the algorithm ever returned. As shown already, 
$\sum_{Q\in \qset}|E(Q)|)\leq \tilde O(md')$, while it is easy to see that $|E(H)|\leq \otilde(m)$. The total running time of the oracle is then bounded by $\otilde(md')$. 

We note that in general, the algorithm for \Cref{thm: alg for matching player} does not need to compute and maintain the graph $\hat G$ explicitly, this graph is only used for the sake of analysis. Instead, for every special edge $e\in E(G)$, for every integer $0\leq j\leq \ceil{\log d'}$, we only need to keep track of the number of copies of $e$ that lie in graph $\hat G$ and have length $2^j/d'$. Therefore, the total running time of the remainder of the algorithm is bounded by $O(m+\sum_{Q\in \qset}|E(Q)|)\leq \otilde(md')$. This includes the time required to perform a BFS in $G'$,  in order to compute a spase cut $(S,T)$ if needed. The total running time of the whole algorithm  for \Cref{thm: alg for matching player} is then bounded by $\otilde (md')$.

\newpage
\appendix

\section{Proof of \Cref{obs: ball growing}}
\label{sec: appx ball growing}
%\begin{proof}
	For all integers $i\geq 0$, let $S_i=B_G(x,i)$ be the set of all vertices $u\in V(G)$ with $\dist_G(x,u)\leq i$.
	
	Let $L=d/4$. We claim that, if $|S_L|< n/2$, then there is some index $1\leq i<L$ for which $|E_G(S_i,V\setminus S_i)|\leq  \phi\cdot |S_i|$ holds.
	
	Indeed, assume otherwise. Then for all $1\leq i<L$, $|E_G(S_i,V\setminus S_i)|\geq \phi\cdot |S_i|$. Denote $R_i=S_{i+1}\setminus S_i$. Since every edge in $E_G(S_i,V\setminus S_i)$ must enter a vertex of $R_i$, and the degree of every vertex in $G$ is at most $\dmax$, we get that $|R_i|\geq \frac{|E_G(S_i,V\setminus S_i)|}{\dmax}\geq \frac{\phi\cdot |S_i|}{\dmax}$ must hold.
	
	Therefore, for all $1\leq i<L$, $|S_{i+1}|\geq \left(1+\frac{\phi}{\dmax}\right )|S_i|$, and $|S_L|\geq \left(1+\frac{\phi}{\dmax}\right )^L\geq \left(1+\frac{\phi}{\dmax}\right )^{d/4}$.

Since $d=\frac{32\dmax\log n}{\phi}$, we get that $|S_L|>n$, a contradiction.

	For all integers $i\geq 0$, we also let $S'_i=\set{u\in V(G) \mid \dist(u,y)\leq i}$. Using the same reasoning, if $|S'_L|\leq n/2$, then there is some index $1\leq i<L$ for which $|E_G(V\setminus S'_i,S'_i)|\leq  \phi\cdot |S'_i|$.

We are now ready to describe our algorithm. We run in parallel two processes. The first process simply performs a BFS from vertex $x$, while the second process performs a reversed BFS from vertex $y$. The first process, whenever it computes a set $S_i$, checks whether $|E_G(S_i,V\setminus S_i)|\leq \phi\cdot |S_i|$. If so, we terminate both processes and return the cut $(A,B)=(S_i,V\setminus S_i)$. Similarly, the second process, whenever it computes a set $S'_j$, checks whether $|E_G(V\setminus S'_j,S'_j)|\leq  \phi\cdot |S'_j|$ holds. If so, we terminate both processes and return the cut $(A,B)=	(V\setminus S'_j,S'_j)$. Since we set $L=d/4$, it must be the case that either $|S_L|\leq n/2$ or $|S'_L|\leq n/2$, so one of the two processes must terminate with a sparse cut. Since we run the two processes in parallel, it is easy to verify that the running time is bounded by the number of edges that the process that successfully terminated explored. Therefore, the running time of the algorithm is at most $O(\dmax\cdot \min\set{|A|,|B|})$.
%	\end{proof}

\section{Proof of~\Cref{cor: one iteration main}}
\label{appsec:cor_one_iteration_main}

		We start by applying the algorithm from \Cref{thm: one iteration main} to graph $H$, to obtain a collection $\pset$ of $\Omega(\Delta/\poly\log n)$  $s$-$t$ paths in $H$, that cause edge-congestion at most $O(\log n)$. 
	We construct a directed graph $H'\subseteq H$, that consists of all vertices and edges that participate in the paths of $\pset$.
	The capacity of every edge in $H'$ remains unit -- the same as its capacity in $H$.
	
	Since the paths in $\pset$ cause edge-congestion $\eta=O(\log n)$, due to the special structure of graph $H$, they also cause vertex-congestion  $\eta$. Therefore, $|E(H')|\leq O(n\log n)$. Additionally, by sending $1/\eta=\Omega(1/\log n)$ flow units on every path in $\pset$, we can obtain a valid $s$-$t$ flow in $H'$, that obeys edge capacities, whose value is $\Omega(|\pset|/\log n)$. 
	
	Next, we compute an integral maximum $s$-$t$ flow in $H'$ that obeys edge capacities in $H'$, using the standard Ford-Fulkerson algorithm. We start by setting the flow $f(e)$ on every edge $e\in E(H')$ to $0$, and then perform iterations. In every iteration, we use the current flow $f$ in $H'$ in order to compute a residual flow network $H''=H'_f$. If there is no $s$-$t$ path in $H''$, then the algorithm terminates. We are then guaranteed that the value of the current flow $f$ is optimal, so $\val(f)\geq \Omega(|\pset|/\log n)$. Otherwise, we compute an arbitrary $s$-$t$ path $P$ in $H''$, and we augment the flow $f$ along the path $P$. Notice that every iteration of the algorithm can be implemented in time $O(|E(H'')|)=\otilde(n)$, and the number of iterations is bounded by $n$. Therefore, the total running time of this algorithm is $\otilde(n^2)$. The final integral flow $f$ can be decomposed, in time $O(|E(H')|)=\otilde(n)$, into a collection $\pset'$ of $\Omega(|\pset|/\log n)\ge \Omega(\Delta/\poly\log n)$ edge-disjoint paths connecting $s$ to $t$ in $H'$. Due to the special structure of the graph $H$, the paths in $\pset'$ are also internally vertex-disjoint. 
	It now remains to analyze the running time of the algorithm. The running time of the algorithm from \Cref{thm: one iteration main} is $\otilde\left(\frac{n^{2.5}}{\sqrt{\Delta}}\right ) $; the running time of the Ford-Fulkerson algorithm in graph $H'$ is $\otilde(n^2)$, and the time required to decompose the resulting flow into a collection $\pset'$ of paths is bounded by $\otilde(n)$, as discussed above. Since $\Delta \le n$, the overall running time of the algorithm is bounded by $\otilde\left(\frac{n^{2.5}}{\sqrt{\Delta}}\right ) $.
	
	We note that we could have used the blocking flow idea used in Hopcroft-Karp algorithm to compute the set $\pset'$ of paths in $H'$ in time $\otilde(n^{1.5})$. Alternatively, we could also directly round the initial fractional $s$-$t$ flow in $H'$, in expected time $\otilde(|E(H)|)=\otilde(n)$, e.g. by using the algorithm from Theorem 5 in~\cite{LRS13}, that builds on the results of \cite{GKK10}. The latter would however result in a randomized algorithm, and since the bottlenecks in the running time of our algorithm lie elsewhere we use this simpler albeit somewhat slower method.

%-------------------------
%-------------------------
%-------------------------
%-------------------------
%-------------------------
\section{Proof of \Cref{thm: almost DAG routing}}
\label{subsec: proof of DAG-like SSSP}
%-------------------------
%-------------------------
%-------------------------
%-------------------------
%-------------------------

\paragraph{Initial Transformation of  Edge Lengths.} We start by performing a simple transformation of the lengths of edges, whose goal is to ensure that $\Gamma\cdot \log n\leq d\leq \max\set{4n/\eps,4\Gamma\cdot \log n}$.
 
We perform the transformation in two steps. The first step is only performed if $d>\ceil{n/\eps}$. For every edge $e\in E(G)$, we let the new length  $\tilde \ell(e)$ be obtained by rounding the length $\ell(e)$ up to the next integral multiple of $\ceil{\eps d/n}$, and we set $\tilde d=(1+2\eps)d$. Consider any simple $s$-$t$ path $P$ in $G$. Since $|E(P)|\leq n$, we get that $\sum_{e\in E(P)}\ell(e)\leq \sum_{e\in E(P)}\tilde \ell(e)\leq 2\eps d+\sum_{e\in E(P)}\ell(e)$. Notice that, if $P$ is a simple $s$-$t$ path in $G$, with $\sum_{e\in E(P)}\tilde \ell(e)\leq (1+6\eps)\tilde d$, then $\sum_{e\in E(P)} \ell(e)\leq (1+10\eps)d$.
Moreover, if $G'$ is a subgraph that is obtained from $G$ by deleting some edges from it, and $G'$ does not contain an $s$-$t$ path $P$  with $\sum_{e\in E(P)}\tilde \ell(e)\leq \tilde d$, and $\sum_{e\in E(P)}w(e)\leq \Gamma$, then $G$ may not contain an $s$-$t$ path $P$  with $\sum_{e\in E(P)}\ell(e)\leq d$, and $\sum_{e\in E(P)}w(e)\leq \Gamma$. Therefore, we can now replace the lengths $\ell(e)$ of edges $e\in E(G)$ with new  lengths $\tilde \ell(e)$. The problem remains unchanged, except that now, in response to  $\pquery$, the algorithm needs to return a simple $s$-$t$ path $P$ of length at most $(1+6\eps)\td$. 

Since all edge lengths are now integral multiples of $\ceil{\frac {\eps d}{n}}$, we can divide the length of every edge $e$ by $\ceil{\frac {\eps d}{n}}$, obtaining a new length, that we still denote by $\ell(e)$, which remains an integer. We also update $\td$ by dividing its value by $\ceil{\frac {\eps d}{n}}$. This scaling of the lengths of edges and of parameter $d$ do not change the problem, but now we are guaranteed that $\td\leq 4n/\eps$. For convenience, we will call $\tilde d$ by $d$ from now on.

If $d\geq \Gamma\cdot \log n$ now holds, no further transformations are necessary, since we are now guaranteed that $\Gamma\cdot \log n\leq d\leq \max\set{4n/\eps,4\Gamma\cdot \log n}$. Assume now that $d<\Gamma\cdot \log n$.
Then we multiply the lengths of all edges by factor $\ceil{\frac{\Gamma\log n}{d}}$, and replace $d$ with value $d\cdot \ceil{\frac{\Gamma\log n}{d}}$; this transformation does not change the problem, but it ensures that $\Gamma\cdot \log n\leq d\leq 4\Gamma\cdot \log n$. From now on we assume that $\Gamma\cdot \log n\leq d\leq \max\set{4n/\eps,4\Gamma\cdot \log n}$ holds. The problem remains the same, except that now, in response to  $\pquery$, the algorithm needs to return a simple $s$-$t$ path $P$ of length at most $(1+6\eps)d$.

\paragraph{Modified Edge Lengths.}
We define, for every edge $e\in E(G)$, a \emph{modified edge length} $\ell'(e)$, as follows: $ {\ell'(e)=\ell(e)+\eps\cdot \ceil{\frac{d}{\Gamma\cdot \log n}\cdot w(e)}}$. Notice that, since $1/\eps$ is an integer, for every edge $e\in E(G)$, its modified length $\ell'(e)$ is an integral multiple of $\eps$. Throughout, for a pair of vertices $u,v$ in $G$, we denote by $\dist'(u,v)$ the distance from $u$ to $v$ in the current graph $G$, with respect to edge lengths $\ell'(\cdot)$. We need the following simple observation.

\begin{observation}\label{obs: short path translation}
	Assume that at time $\tau\in \tset$,  there is an $s$-$t$ path $P$ in $G$ with $\sum_{e\in E(P)}\ell(e)\leq d$, and $\sum_{e\in E(P)}w(e)\leq \Gamma$. Then at time $\tau$, $\dist'(s,t)\leq (1+3\eps)\cdot d$.
\end{observation}
\begin{proof}
	For all $0\leq i\leq \ceil{\log n}$, let $E_i$ be the set of all edges $e\in E(P)$ with $w(e)=2^i$. Since $\sum_{e\in E(P)}w(e)\leq \Gamma$, for all such $i$, $|E_i|\leq \frac{\Gamma}{2^i}$. We now get that:
	
	\[\begin{split}
	\sum_{e\in E(P)}\ell'(e)&\leq \sum_{e\in E(P)}\left(\ell(e)+\frac{\eps\cdot d}{\Gamma\cdot \log n}\cdot w(e)+\eps\right )\\
	&\leq \sum_{e\in E(P)}(1+\eps)\ell(e)+\sum_{i=0}^{\ceil{\log n}}\sum_{e\in E_i}\frac{\eps\cdot d}{\Gamma\cdot \log n}\cdot 2^i\\
		&\leq (1+\eps) d+\sum_{i=0}^{\ceil{\log n}} |E_i|\cdot \frac{\eps\cdot d}{\Gamma\cdot \log n}\cdot 2^i\\
			&\leq  (1+\eps)\cdot d+\sum_{i=0}^{\ceil{\log n}} \frac{\Gamma}{2^i}\cdot \frac{\eps\cdot d}{\Gamma\cdot \log n}\cdot 2^i\\
			&\leq (1+3\eps)d.
	\end{split}\]
	(we have used the fact that $\ell(e)\geq 1$ for every edge $e\in E(G)$).
\end{proof}

Throughout the algorithm, we maintain, for every vertex $v\in V(G)$, an estimate $\tdist(v)$ on $\dist'(s,v)$. If $\tdist(v)>2d$ holds, we replace $\tdist(v)$ with $\infty$. Additionally, for every edge $e=(v,u)$, vertex $u$ is given an estimate $\tdist_u(v)$ on $\dist'(s,v)$ (that may be outdated). More specifically, if $w(e)=2^i$, then vertex $v$ only communicates its estimate $\tdist(v)$ to $u$ in increments of ${\eps^2\cdot \ceil{\frac{d\cdot 2^i}{\Gamma\cdot \log n}}}$ (it is this modification of the standard Even-Shiloach tree algorithm that leads to the faster total update time). In other words, every time $\tdist(v)$ passes a threshold of $ {z\cdot \eps^2\cdot \ceil{\frac{d \cdot 2^i}{\Gamma\cdot \log n}}}$ for some integer $z$, vertex $v$ communicates its current estimate $\tdist(v)$ to $u$, and $\tdist_u(v)$ is replaced with this amount. 
Additionally, we maintain a directed tree $T\subseteq G$, that is rooted at vertex $s$, with all edges directed away from the root. The tree contains all vertices $v\in V(G)$ with $\tdist(v)\neq\infty$. Whenever we refer to lengths of paths or distances in the tree $T$, it is with respect to edge lengths $\ell'(e)$ for $e\in E(T)$.

\paragraph{Additional Data Structures.} In addition to the tree $T$, the estimates $\tdist(v)$ for vertices $v\in V$, and estimates $\tdist_u(v)$ for edges $(v,u)\in E$, we maintain the following data structures. For every vertex $v\in V$, we maintain a heap $H(v)$ that contains all vertices of $\set{u\in V\mid (u,v)\in E}$, with the key associated with each such vertex $u$ being $\tdist_v(u)+\ell'(u,v)$. Additionally, for all $0\leq i\leq \ceil{\log n}$, we maintain a set $U_i(v)$ of vertices, that contains every vertex $u\in V(G)$, such that there is an edge  $e=(v,u)$ with $w(e)=2^i$ in $G$.

\paragraph{Invariants.}
We ensure that, throughout the algorithm, the following invariants hold.

\begin{properties}{I}
	\item For every vertex $v$, the value $\tdist(v)$ is an integral multiple of $\eps$, and it may only grow over the course of the algorithm; the same holds for values $\tdist_u(v)$ for all edges $e=(v,u)$; \label{inv: estimates grow}
	\item If $e=(x,y)$ is an edge in the tree $T$, then vertex $x$ is the vertex with the lowest key $\tdist_y(x)+\ell'(x,y)$ in $H(y)$, and moreover $\tdist(y)=\tdist_y(x)+\ell'(x,y)$; \label{inv: keys via tree}
	\item For every edge $e=(u,v)$, if $w(e)=2^i$, then $\tdist(u)-\eps^2\cdot \ceil{\frac{d \cdot 2^i}{\Gamma\cdot \log n}}\leq \tdist_v(u)\leq \tdist(u)$; \label{inv: updates of neighbors}

	\item For all $v\in V(G)$, if $\dist'_G(s,v)\leq 2d$, then $\tdist(v)\neq \infty$ and $\tdist(v)\leq \dist'(s,v)$; and\label{inv: dist estimate upper bound}
	\item For all $v\in V(G)$, $\tdist(v)\geq (1-\eps)\cdot \dist'(s,v)$ holds at all times; moreover, if $v\in V(T)$, then the length of the $s$-$v$ path in the current tree $T$ is at most $\frac{\tdist(v)}{1-\eps}$. \label{inv: dist estimate lower bound} 

\end{properties}

\paragraph{Initialization.}
At the beginning of the algorithm, we compute a shortest-path out-tree $T$ rooted at vertex $s$, with respect to edge lengths $\ell'(\cdot)$. For every vertex $v\in V$, we then set $\tdist(v)=\dist'(s,v)$ -- the distance from $s$ to $v$ in the tree $T$.
Next, we scan all vertices in the bottom-up order in the tree $T$. For each such vertex $v$, if $\tdist(v)>2d$, then we set $\tdist(v)=\infty$, and delete $v$ from the tree. For every edge $e=(v,u)\in E$, we  set $\tdist_u(v)=\tdist(v)$. Lastly, for every vertex $v\in V$, we initialize the heap $H(v)$ with the vertices of $\set{u\in V\mid (u,v)\in E}$, and the sets $\set{U_i(v)}_{0\leq i\leq \ceil{\log n}}$ of its out-neighbors.  It is easy to verify that the initialization can be performed in time $O(m\log m)$, and that all invariants hold at the end of initialization.

Next we describe the algorithms for handling vertex-splits and edge deletions.

\paragraph{Algorithm for Handling Vertex Splits.}

Consider a vertex split operation, in which we are given a vertex $v\in V$, and 
 a collection $U=\set{u_1,\ldots,u_k}$ of new vertices that need to be inserted into $G$. Recall that we are also given a collection $E'$ of new edges that must be inserted into $G$, and for every edge $e\in E'$, we are given a length $\ell(e)$ and a weight $w(e)$. For each such edge $e\in E'$, either both its endpoints lie in $\set{u_1,\ldots,u_k, v}$; or one endpoint lies in $\set{u_1,\ldots,u_k}$, and the other in $V(G)\setminus\set{v,u_1,\ldots,u_k}$. In the latter case, if $e=(u_i,x)$, then edge $e'=(v,x)$ must currently lie in $G$, and both $w(e')=w(e)$ and $\ell(e')=\ell(e)$ must hold, so $\ell'(e')=\ell'(e)$. Similarly, if $e=(x,u_i)$, then edge $e'=(x,v)$ must currently lie in $G$, and both $w(e')=w(e)$ and $\ell(e')=\ell(e)$ must hold, so $\ell'(e')=\ell'(e)$.
 
 If $\tdist(v)\neq \infty$, then let $e=(x,v)$ be the unique edge that is incoming into $v$ in the tree $T$. We \emph{temporarily} insert, into set $E'$, edges $e_i=(x,u_i)$ with $\ell'(e_i)=\ell'(e)$ and $w(e_i)=w(e)$ for all $u_i\in U$.
 
 Our algorithm consists of two steps.
 In the first step,
we set, for every vertex $u_i\in U$, $\tdist(u_i)=\tdist(v)$. We then inspect the edges of $E'$ one by one. 
Consider any such edge $e=(a,b)$.

\begin{itemize}
\item If both endpoints of $e$ lie in $U\cup\set{v}$, then we set $\tdist_b(a)=\tdist(a)$, and update the corresponding heap $H(b)$ by inserting $a$ into it.
 
\item If the second endpoint $b$ of $e$ lies in $U$ and the first endpoint $a$ lies in $V(G)\setminus (U\cup\set{v})$, then edge $e'=(a,v)$ with $\ell'(e')=\ell'(e)$ must be present in $G$. In this case, we set $\tdist_b(a)=\tdist_v(a)$, and update heap $H(b)$ by inserting $a$ into it.

\item Similarly, if the first endpoint $a$ of $e$ lies in $U$, and the second endpoint $b$ lies in $V(G)\setminus (U\cup\set{v})$, then edge $e'=(v,b)$ with $\ell'(e')=\ell'(e)$ must be present in $G$. In this case, we set $\tdist_b(a)=\tdist_b(v)$, and update heap $H(b)$ by inserting $a$ into it.
\end{itemize}

Lastly, if $w(e)=2^i$, and $b\not \in U_i(a)$, then we insert $b$ into $U_i(a)$.
We also update the graph $G$ with the insertion of the vertices of $U$ and edges of $E'$.
 Notice that all these updates can be done in time $O((|U|+|E'|)\cdot \log m)$. 
 
The second step only needs to be executed if $\tdist(v)\neq\infty$. In this step we delete from $G$ the edges $e_1,\ldots,e_k$ that were temporarily inserted. Our data structure is then updated using the algorithm that we describe below, for handling edge-deletion updates.

We now establish that all invariants hold after the first step. 
Invariant \ref{inv: estimates grow} continues to hold, as we did not update any $\tdist(a)$ or $\tdist_b(a)$ estimates except for the newly inserted vertices. 

 In order to establish Invariant \ref{inv: keys via tree}, assume that $\tdist(v)\neq \infty$, and so edge $(x,v)$ lies in tree $T$. Observe that for every vertex $u_i\in U$, all vertices in $H(u_i)$ lie in $H(v)$, except possibly for the vertices of $\left(U\cup v\right )\setminus\set{u_i}$. But for every vertex $a\in \left(U\cup v\right )\setminus\set{u_i}$,  $\tdist(a)=\tdist(v)>\tdist(x)$ currently holds, so the key of such a vertex $a$ in the heap $H(u_i)$ is greater than or equal to $\tdist_{u_i}(v)= \tdist(v)=\tdist_{v}(x)+\ell'(x,v)=\tdist_{u_i}(x)+\ell'(x,u_i)$.
Since $u_i$ is connected to $x$ in the tree $T$, we get that Invariant \ref{inv: keys via tree} holds for $u_i$. Consider now any vertex $a\not\in U$, and assume that some edge $(b,a)$ lies in tree $T$. It is possible that we have inserted new vertices into $H(a)$: if an edge $e=(u_i,a)$ for some $u_i\in U$ lies in $E'$, then $u_i$ was inserted into $H(a)$. However, in this case, edge $e'=(v,a)$ with $\ell'(e')=\ell'(e)$ belonged to $G$ before the current update, and so $v$ already lied in $H(a)$.  Since $\ell'(e)=\ell'(e')$, and $\tdist_a(u_i)=\tdist_a(v)$, the key of $u_i$ in heap $H(a)$ is equal to the key of $v$, that remained unchanged. Therefore, Invariant \ref{inv: keys via tree} continues to hold.

We now turn to prove that Invariants \ref{inv: dist estimate upper bound} and  \ref{inv: dist estimate lower bound} continue to hold. Consider any vertex $a\in V(G)$. If $a\not \in U$, then it is easy to see that $\dist'(s,a)$ did not change after the current update. This is because for any path $P$ connecting $s$ to $a$ in the current graph $G$, we can obtain a path $P'$ connecting $s$ to $a$ in the original graph $G$, whose length with respect to the modified edge lengths $\ell'(\cdot)$ is no higher than that of $P$, by replacing every vertex of $U$ on $P$ with the vertex $v$. Moreover, path $P'$ still belongs to the current graph since we did not delete any vertices or edges during the current update. Since $\tdist(a)$ did not change, and the path connecting  $s$ to $a$ in $T$ did not undergo any changes either, both invariants continue to hold for $a$.

Assume now that $a\in U$. Let $P$ be the path connecting $s$ to $v$ in the current tree $T$. Since path $P$ did not change during the current update, from Invariant \ref{inv: dist estimate lower bound} , the length of this path is at most $\frac{\tdist(v)}{1-\eps}$. Notice that the last edge on this path is $e=(x,v)$, and recall that we have inserted an edge $e'=(x,a)$ into $G$, with $\ell'(e')=\ell'(e)$. By replacing edge $e$ with $e'$ on path $P$, we obtain a path $P'$, connecting $s$ to $a$, which is the unique $s$--$a$ path in tree $T$. Moreover, the length of $P'$ is equal to the length of $P$, which, in turn, is at most $\frac{\tdist(v)}{1-\eps}=\frac{\tdist(a)}{1-\eps}$. We conclude that $\dist'(s,a)\leq \frac{\tdist(a)}{1-\eps}$, establishing Invariant \ref{inv: dist estimate lower bound}  for vertex $a$.

In order to establish Invariant \ref{inv: dist estimate upper bound} for vertex $a$, let $Q$ be the shortest $s$-$a$ path in graph $G$, with respect to edge lengths $\ell'(\cdot)$. By replacing every vertex of $U$ on this path with vertex $v$, and removing loops, we obtain an $s$-$v$ path $Q'$ in $G$, whose length is no higher than the length of $Q$. Therefore, $\dist'(s,v)\leq \dist'(s,a)$. Since we set $\tdist(a)=\tdist(v)$, we get that, if $\dist'(s,a)\neq \infty$, then $\dist'(s,v)\neq \infty$, and $\tdist(a)=\tdist(v)\leq \dist'(s,v)\leq \dist'(s,a)$. We conclude that Invariants \ref{inv: dist estimate upper bound} and \ref{inv: dist estimate lower bound}  continue to hold.

It now remains to establish Invariant \ref{inv: updates of neighbors}. 
Consider any edge $e=(a,b)$. Notice that none of the values $\tdist(a)$ of vertices $a\in V$ changed, except for the vertices of $U$ that were inserted into $G$. Similarly, value $\tdist_b(a)$ remained unchanged, except if $a$ or $b$ (or both) are vertices that were newly inserted into $G$. In other words, we only need to consider edges $e=(a,b)\in E'$.

Let $e=(a,b)$ be any such edge. As before, we need to consider three cases. The first case happens if both $a,b\in U\cup \set{v}$. In this case, we set $\tdist(a)=\tdist(b)=\tdist(v)$ and $\tdist_b(a)=\tdist(v)$, so the invariant clearly holds. Assume now that $a\in U$ and $b\not\in U\cup\set{v}$. In this case, edge $e'=(v,b)$ was present in $G$ before the current update, and we set $\tdist(a)=\tdist(v)$ and $\tdist_b(a)=\tdist_b(v)$. Since the invariant held for edge $(v,b)$ before the current update, and since $w(a,b)=w(v,b)$, the invariant holds for edge $(a,b)$ at the end of the current update. Lastly, assume
that  $a\not\in U\cup\set{v}$ and $b\in U$. In this case, edge $(a,v)$ was present in $G$ before the current update, we set $\tdist_b(a)=\tdist_v(a)$ and $\tdist(a)$ remains unchanged. Since Invariant \ref{inv: updates of neighbors} held for edge $(a,v)$ at the beginning of the current update, it now holds for edge $(a,b)$ at the end of the update.

\paragraph{Algorithm for Handling Edge Deletions.}
We now provide an algorithm for updating the data structure with the deletion of some edge $e=(x,y)$. First, we may need to update the heap $H(y)$ (by either deleting $x$ from it, or changing the key of $x$ in the heap, if another $(x,y)$ edge is present in $G$). We may also need to delete $y$ from the appropriate set $U_i(x)$; all this can be done in time $O(\log n)$. If $e\not\in T$, then no further updates are needed. We assume from now on that $e\in E(T)$. We start by deleting the edge $e$ from the tree $T$. 

Throughout the algorithm, we let $S$ be the set of vertices that are not currently connected to the tree $T$. We will not maintain the set $S$ explicitly, but it will be useful  for the analysis. At the beginning of the algorithm, the set $S$ contains all vertices that lied in the original subtree of $T$ rooted at $y$.

Our algorithm maintains a heap $Q$ that contains a subset of vertices of $S$ that need to be examined. The key associated with each vertex $v$ in $Q$ is $\tdist(v)$. We will ensure that the following invariant holds throughout the update:

\begin{properties}{J}
	\item If $a$ is the vertex with minimum value $\tdist(a)$ in $S$, then $a\in Q$ holds. \label{inv: smallest key in Q}
\end{properties}

We start with the heap $Q$ containing only the vertex $y$. Then we perform iterations, as long as $Q\neq\emptyset$. In order to execute a single iteration, we let $a$ be the vertex with smallest key $\tdist(a)$ in the heap $Q$, and we let $b$ be the vertex in heap $H(a)$ with smallest key.
If $\tdist(a)>2d$, then we set $\tdist(a)$ to $\infty$ and remove it from $Q$. We then set value $\tdist(x)$ of every vertex that is currently not connected to the tree $T$ to $\infty$ and terminate the update algorithm. We assume from now on that $\tdist(a)\leq 2d$. We then consider two cases.

The first case happens if $\tdist_a(b)+\ell'(b,a)\leq \tdist(a)$. Note that from Invariant \ref{inv: smallest key in Q}, vertex $b$ is currently attached to the tree $T$. In this case, we remove $a$ from $Q$ and attach it to the tree $T$ as a child of vertex $b$, via the edge $(b,a)$. Every vertex that is currently a descendant of $a$ also becomes automatically attached to the tree $T$ and removed from the set $S$; vertex $a$ itself is removed from $S$ as well.

Consider now the second case, where $\tdist_a(b)+\ell'(b,a)> \tdist(a)$. In this case, we increase $\tdist(a)$ by (additive) $\eps$, and keep vertex $a$ in the heap $Q$. Additionally, for every integer $0\leq i\leq \ceil{\log n}$, such that, for some integer $q$, threshold $q\cdot \eps^2\cdot \ceil{\frac{d \cdot 2^i}{\Gamma\cdot \log n}}$ lies between the old and the new value $\tdist(a)$, we consider every vertex $z\in U_i(a)$. We set $\tdist_z(a)=\tdist(a)$ and update the key of vertex $a$ in the heap $H(z)$. Additionally, if $z$ is currently a child-vertex of $a$, then we disconnect $z$ from $a$ (so it stops being a child of $a$), and add $z$ to the heap $Q$. This finishes the algorithm for handling edge-deletion updates.

We start by showing that Invariant \ref{inv: smallest key in Q} holds throughout the algorithm.

\begin{claim}\label{claim: update property holds}
Invariant \ref{inv: smallest key in Q} holds throughout the algorithm.
\end{claim}
\begin{proof}
 Consider any time during the execution of the update procedure, and
let $a$ be the vertex of $S$ with smallest value $\tdist(a)$. Assume for contradiction that $a\not\in Q$. Observe first that, whenever the value $\tdist(v)$ changes for any vertex of $S$, we add $v$ to $Q$. Vertex $v$ can only leave the heap $Q$ when it is connected to the tree $T$, or if we set $\tdist(v)$ to $\infty$. Therefore, if $\tdist(a)$ ever changed over the course of the current update operation, then $a\in Q$ must hold. We assume from now on that value $\tdist(a)$ did not change over the course of the current update operation. Let $b$ be the parent-vertex of $a$ at the beginning of the current update. 
From our algorithm, if value $\tdist_b(a)$ ever changed over the course of the current update operation, then vertex $a$ was immediately added to heap $Q$. Therefore, value $\tdist_b(a)$ remained unchanged throughout the current update. 

Denote the weight of edge $(b,a)$ by $2^i$.
In the next observation, we show that, at the beginning of the current update procedure, $\tdist(b)$ was significantly lower than $\tdist(a)$. We will then conclude that $\tdist(b)$ must have grown significantly over the course of the current update procedure, and so $\tdist_b(a)$ must have grown as well, leading to a contradiction.

\begin{observation}\label{obs: small tdist of b}
At the beginning of the current update procedure, $\tdist(b)\leq \tdist(a)-\frac{\eps\cdot d}{\Gamma\cdot \log n}\cdot 2^{i-1}$ held.	
\end{observation}
\begin{proof}
	From Invariant \ref{inv: keys via tree}, at the beginning of the current update procedure, $\tdist(a)=\tdist_a(b)+\ell'(b,a)$ held.
Recall that $\ell'(b,a)=\ell(b,a)+\eps\cdot \ceil{\frac{d}{\Gamma\cdot \log n}\cdot 2^i}$. Moreover, from Invariant \ref{inv: updates of neighbors}, 
$\tdist_a(b)\geq \tdist(b)-\eps^2\cdot\ceil{\frac{ d \cdot 2^i}{\Gamma\cdot \log n}}$ held at the beginning of the procedure for processing the update.
Altogether, we get that, at the beginning of the current update procedure:

\[ \begin{split}
\tdist(a)&=\tdist_a(b)+\ell'(b,a)\\
&\geq \tdist(b)-\eps^2\cdot\ceil{\frac{d \cdot 2^i}{\Gamma\cdot \log n}}+\ell(b,a)+\eps\cdot \ceil{\frac{d}{\Gamma\cdot \log n}\cdot 2^i}\\
&\geq \tdist(b)+\frac{\eps\cdot d}{\Gamma\cdot \log n}\cdot 2^{i-1}.
\end{split} \]
%
%Therefore, at the beginning of the current update procedure, $\tdist(b)\leq \tdist(a)-\frac{\eps\cdot d}{\Gamma\cdot \log n}\cdot 2^{i-1}$ held.  
\end{proof}

Since vertex $a$ is the vertex of $S$ with smallest value $\tdist(a)$, either vertex $b$ is currently connected to the tree $T$, or the value $\tdist(b)$ has grown to become at least $\tdist(a)$ over the course of the current update procedure.

If vertex $b$ remained attached to the tree $T$ over the course of the entire update procedure, then $a=y$ must hold; in other words, the edge that was deleted from $G$ is $(b,a)$. In this case, $a$ must have been added to $Q$ immediately at the beginning of the update procedure.
Otherwise, if vertex $b$ became attached to tree $T$ at any time during the update procedure, then each of its children was either attached to tree $T$ together with it, or added to $Q$. Since $a$ was never added to $Q$, it must be the case that $b$ was never attached to $T$, and so $\tdist(b)$ increased over the course of the current update procedure to become at least $\tdist(a)$. We claim that in this case, value $\tdist_b(a)$ must have been updated, leading to a contradiction.

Indeed,  from the above calculations, over the course of the current update operation, value $\tdist(b)$ must have increased by at least $\frac{\eps\cdot d}{\Gamma\cdot \log n}\cdot 2^{i-1}>\eps^2\cdot\ceil{\frac{ d \cdot 2^i}{\Gamma\cdot \log n}}$ units, (since $d\ge \Gamma\log n$ and $\eps\leq 1/8$), and so value $\tdist_a(b)$ must have been updated, with vertex $a$ added to $Q$, a contradiction.
\end{proof}

We now verify that Invariants \ref{inv: estimates grow}--\ref{inv: updates of neighbors} continue to hold at the end of the update procedure, assuming that they held before the procedure started.

First, it is immediate to verify that values $\tdist(v)$ for vertices $v\in V$ and values $\tdist_u(v)$ for edges $(v,u)$ are never decreased, and they only increase by intergral multiples of $\eps$, so Invariant \ref{inv: estimates grow} continues to hold.

Consider now some edge $e=(v,u)$ with $w(e)=2^i$. Notice that our algorithm updates $\tdist_u(v)$ to the current value $\tdist(v)$ whenever the latter passes a threshold that is an integral multiple of $\eps^2\cdot \ceil{\frac{ d \cdot 2^i}{\Gamma\cdot \log n}}$. This ensures that  $\tdist(u)-\eps^2\ceil{\frac{d \cdot 2^i}{\Gamma\cdot \log n}}\leq \tdist_v(u)\leq \tdist(u)$ holds at all times and establishes Invariant \ref{inv: updates of neighbors}.

In order to establish Invariant \ref{inv: keys via tree}, we consider the time when some vertex $v\in S$ reconnects to the tree $T$ via an edge $(u,v)$. Notice that at this time, if $z$ is a descendant of $v$, and no vertex on the $v$--$z$ path in $T$ was ever added to $Q$, then $z$ becomes connected to the tree $T$ as well. If $z'$ is the parent of $z$, then $\tdist_{z}(z')$ and $\tdist(z)$ remained unchanged over the course of the current update procedure, so $\tdist(z)=\tdist_z(z')+\ell'(z',z)$ continues to hold, and $z'$ is the vertex with the lowest key $\tdist_z(z')+\ell'(z',z)$ in the heap $H(z)$. Consider now vertex $v$ itself, and recall that it reconnects to the tree $T$ via the edge $(u,v)$. From our algorithm, it must be the case that $u$ is the vertex with the smallest key $\tdist_v(u)+\ell'(u,v)$ in $H(v)$, and moreover, 
$\tdist_v(u)+\ell'(u,v)\leq \tdist(v)$. It now remains to show that 
$\tdist_v(u)+\ell'(u,v)= \tdist(v)$. If the value $\tdist(v)$ never increased over the course of the current update procedure, then, since Invariant \ref{inv: keys via tree} held at the beginning of the update procedure, and since value $\tdist_v(u)$ may only grow, it is impossible that $\tdist_v(u)+\ell'(u,v)< \tdist(v)$. Therefore, in this case, $\tdist_v(u)+\ell'(u,v)= \tdist(v)$ must hold at the end of the update procedure. Assume now that the value $\tdist(v)$ grew over the course of the current update procedure, and let $\tau$ be the time when it was increased last (recall that each increase is by exactly $\eps$ units). Then, at time $\tau$, $\tdist_v(u)+\ell'(u,v)> \tdist(v)$ held. Since value $\tdist_v(u)$ may not decrease, and since $\tdist_v(u)+\ell'(u,v)$ is an integral multiple of $\eps$, it must be the case that $\tdist_v(u)+\ell'(u,v)= \tdist(v)$ at the end of the procedure.

In order to establish Invariant \ref{inv: dist estimate lower bound}, it is enough to show that, if $v\in V(T)$, then the length of the $s$-$v$ path in the current tree $T$ is at most $\frac{\tdist(v)}{1-\eps}$. The proof is by induction on the number of edges on the path. The claim clearly holds for the root vertex $s$. Let $j>0$ be any integer. We assume that the claim holds for all vertices $v\in V(T)$ for which the $s$-$v$ path in $T$ contains at most $j-1$ edges, and prove the claim for the case where the path contains $j$ edges. Consider now some vertex $v\in V(T)$, and let $P$ be the $s$-$v$ path in tree $T$, that contains $j$ edges. Let $u$ be the penultimate vertex on the path, and let $P'$ be the subpath of $P$ from $s$ to $u$. Then $\ell'(P)=\ell'(P')+\ell'(u,v)$. Let $i$ be the integer for which $w(u,v)=2^i$. From the induction hypothesis:

\[ \ell'(P')\leq \frac{\tdist(u)}{1-\eps}. \]

From Invariant \ref{inv: keys via tree},  $\tdist(v)=\tdist_v(u)+\ell'(u,v)$, and from Invariant \ref{inv: updates of neighbors}, $\tdist_v(u)\geq \tdist(u)-\eps^2\cdot\ceil{\frac{ d \cdot 2^i}{\Gamma\cdot \log n}}$. Altogether, we get that:

\[ \tdist(v)=\tdist_v(u)+\ell'(u,v)\geq \tdist(u)-\eps^2\cdot\ceil{\frac{d \cdot 2^i}{\Gamma\cdot \log n}}+\ell'(u,v)\geq (1-\eps)\cdot \ell'(P')-\eps^2\cdot \ceil{\frac{d \cdot 2^i}{\Gamma\cdot \log n}}+\ell'(u,v). \] 

Since $\ell'(u,v)=\ell(u,v)+\eps\cdot \ceil{\frac{d\cdot 2^i}{\Gamma\cdot \log n}}$, we get that $\eps^2\cdot\ceil{\frac {d \cdot 2^i}{\Gamma\cdot \log n}}\leq \eps\cdot \ell'(u,v)$.
Altogether, we get that:

\[\tdist(v)\geq (1-\eps)\cdot \ell'(P')+(1-\eps)\ell'(u,v)=(1-\eps)\ell'(P),\]
as required.

Lastly, it remains to establish Invariant \ref{inv: dist estimate upper bound}. We first prove that, for every vertex $v$ that lies in the tree $T$,  $\tdist(v)\leq \dist'(s,v)$ holds. The proof is by induction on the number of edges on the shortest path (with respect to edge lengths $\ell'(\cdot)$) connecting $s$ to $v$. Induction base is when the path consists of a single vertex and no edges, so $v=s$, in which case the claim clearly holds.  Assume now that for some integer $j>0$, the claim holds for all vertices $v'$, for which there exists a shortest $s$-$v'$ path that contains fewer than $j$ edges. Consider some vertex $v$, such that there exists a shortest $s$-$v$ path $P^*$ (with respect to edge lengths $\ell'(\cdot)$), containing exactly $j$ edges. Let $u$ be the penultimate vertex on this path. From the induction hypothesis, $\tdist(u)\leq \dist'(s,u)$.
From Invariant \ref{inv: updates of neighbors}, $\tdist_v(u)\leq \tdist(u)$ holds, and, from Invariant \ref{inv: keys via tree}, $\tdist(v)\leq \tdist_v(u)+\ell'(u,v)$. Altogether, we get that $\tdist(v)\leq \tdist(u)+\ell'(u,v)\leq \dist'(s,u)+\ell'(u,v)=\dist'(s,v)$.
From the analysis of the algorithm it is easy to verify that, once $\tdist(v)$ reaches $2d$, it must be the case that $\dist'_G(s,v)\geq 2d$. This establishes Invariant \ref{inv: dist estimate upper bound}.

\paragraph{The Remainder of the Algorithm.} 
We continue the algorithm as long as $\tdist(t)\leq (1+3\eps)\cdot d$ holds. Once  $\tdist(t)> (1+3\eps)\cdot d$, we terminate the algorithm and return ``FAIL''. From Invariant  \ref{inv: dist estimate upper bound}, at this time $\dist'_G(s,t)>(1+3\eps)\cdot d$ must hold, and, from  \Cref{obs: short path translation}, graph $G$ may no longer contain an $s$-$t$ path $P$ with 
$\sum_{e\in E(P)}\ell(e)\leq d$, and $\sum_{e\in E(P)}w(e)\leq \Gamma$.

Whenever a \pathquery arrives, we simply return the unique $s$-$t$ path $P$ in the current tree $T$, in time $O(|E(P)|)$. Since we are guaranteed that $\tdist(t)\leq (1+3\eps)\cdot d$, from Invariant \ref{inv: dist estimate lower bound}, $\sum_{e\in E(P)}\ell'(e)\leq 
\frac{\tdist(t)}{1-\eps}\leq (1+6\eps)d$, and so $\sum_{e\in E(P)}\ell(e)\leq \sum_{e\in E(P)}\ell'(e)\leq (1+6\eps)d$.

\paragraph{Analysis of Total Update Time.}
From the discussion above, the initialization takes time $O(m\log m)$, and the total time spent on vertex-splitting updates (excluding the time spent on edge-deletion updates generated by them) is bounded by $O(m\log m)$.

We now turn to analyze the total time spent on edge-deletion updates. Consider the algorithm for processing the deletion of some edge $e=(x,y)$, that lies in the tree $T$. Initially, we add vertex $y$ to the heap $Q$, and then inspect this vertex. Whenever any vertex $v\in Q$ is inspected, we spend time $O(\log n)$ in order to decide whether it needs to be reconnected to the tree $T$, and, if $\tdist(v)$ increases, to calculate all indices $0\leq i\leq \ceil{\log n}$ for which an integral multiple of  $\eps^2\cdot \ceil{\frac{ d \cdot 2^i}{\Gamma\cdot \log n}}$ lies between the old and the new value $\tdist(v)$. We refer to this as \emph{vertex inspection time}. Additionally, for every integer $i$, for which an integral multiple of $\eps^2\cdot \ceil{\frac{ d \cdot 2^i}{\Gamma\cdot \log n}}$ lies between the old and the new value $\tdist(v)$, we consider every vertex $u\in U_i(v)$, update $\tdist_u(v)$, and possibly insert $u$ into $Q$. We call this procedure the \emph{update of out-neighbor $u$}, and the time required to perform this update, including the possible subsequent inspection of vertex $u$ when it becomes the lowest-key vertex in $Q$, is counted as the time required to update the out-neighbor $u$. However, if such a vertex $u$ is inspected several times (in which case its value $\tdist(u)$ grows with every inspection), then only the first inspection is charged to the time that $v$ spends on updating its out-neighbors. All subsequent times are charged to the increase of value $\tdist(u)$.

To summarize, the total update time due to edge deletion updates can be bounded as follows.

\begin{itemize}
	\item Whenever an edge $e=(x,y)$ is deleted from $G$, we may spend $O(\log m)$ time on inspecting vertex $y$. In total this may contribute up to $O(m\log m)$ to total update time.
	
	\item Whenever, for some vertex $a\in V$, $\tdist(a)$ increases, we may spend $O(\log m)$ time on a subsequent inspection of $a$. Since, when $\tdist(a)$ reaches $2d$, it is set to be $\infty$, and since every increase is by $\eps$ units, there can be at most $O(d/\eps)$ such increases for each vertex $a$. In total this may contribute up to $O\left(\frac{nd\log m}{\eps}\right )$ to total update time.

		\item The remaining time contribution is due to out-neighbor update operations. Consider a vertex $v\in V$,
	and consider an index $0\leq i\leq \ceil{\log n}$. Whenever $\tdist(v)$ passes a threshold that is an integral multiple of 
	$\eps^2\cdot \ceil{\frac{d \cdot 2^i}{\Gamma\cdot \log n}}$, we need to update all vertices in $U_i(v)$. Since $\tdist(v)\leq 2d$, and $\tdist(v)$ only grows throughout the algorithm, there are in total at most $O\left(\frac{\Gamma\cdot\log n}{\eps^2\cdot 2^i}\right )$ such update operations of vertices in $U_i(v)$. 
	Let $n_i(v)$ be the largest number of vertices that are ever present at any given time in set $U_i(v)$. Then each  such update operation takes time $O(n_i(v)\cdot \log n)$, and the total time spent on updating all out-neighbors of $v$ in $U_i(v)$ over the course of the algorithm is bounded by $O\left(\frac{n_i(v)\cdot \Gamma\cdot\log^2 n}{2^i\cdot \eps^2}\right )$. 
	Summing this over all indices $0\leq i \leq\ceil{\log m}$, the total time required for all updates of out-neighbors of $v$ is bounded by:
	
	\begin{equation}\label{eq: updating out-neighbors}
	O\left(\frac{\Gamma\cdot\log^2 n}{\eps^2}\right )\cdot\sum_{i=0}^{\ceil{\log m}}\frac{n_i(v)}{2^i}.
		\end{equation}
		
	From Property \ref{prop: few close neighbors}, if $v\neq t$, then $n_i(v)\leq 2^{i+2}$, and so the total time required to update all out-neighbors of $v$ is bounded by:

\[	O\left(\frac{\Gamma\cdot\log^3 m}{\eps^2}\right ).\]

	The total time required for all updates of out-neighbors of all vertices in $V(G)\setminus\set{t}$ is bounded by $O\left(\frac{\Gamma\cdot n\cdot \log^3 m}{\eps^2}\right )$.
	
	Lastly, if $v=t$, then for all $0\leq i\leq \ceil{\log m}$, $n_i(v)\leq n$ must hold. From Equation \ref{eq: updating out-neighbors}, the total time required for updating all out-neighbors of $t$ is bounded by $O\left(\frac{\Gamma\cdot n\cdot \log^3 m}{\eps^2}\right )$. Altogether, the total time required for all updates of out-neighbors of all vertices of $G$ is bounded by:
	
	\[O\left(\frac{\Gamma\cdot n\cdot \log^3 m}{\eps^2}\right ).\]

\end{itemize}

We conclude that the total update time of the algorithm is bounded by:

\[O\left(\frac{nd\log m}{\eps}+ m\log m + \frac{\Gamma\cdot n\log^3 m}{\eps^2}\right )\leq \tilde O\left(\frac{n^2+m+\Gamma\cdot n}{\eps^2}\right),  \]

since $d\leq \max\set{4n/\eps,4\Gamma\cdot \log n}$ holds.

There are two minor points that we have ignored in our analysis so far. Consider a time when vertex $v$ is split, and a collection $U=\set{u_1,\ldots,u_r}$ of new vertices is inserted into $G$. If $x$ is the parent-vertex of $v$ in the tree $T$, then we may temporarily insert edges $(x,u_i)$ into $G$, for all $u_i\in U$. These new edges are then immediately deleted from $G$, but between the time that they are inserted and until the time they are deleted, Property \ref{prop: few close neighbors} may not hold for $x$. However, it is easy to verify that, during the same time interval, we never update $\tdist(x)$, as the insertions and the deletions of these new edges do not affect the $s$-$x$ path in the current tree, and so we do not need to update out-neighbors of $x$ during this time.
Additionally, these edges that were temporarily inserted into $G$ 
are not included in the total number of edges $m$ that ever belong to $G$. However, every edge that ever lied in $G$ contributes at most one such extra edge, so the total number of edges ever present in $G$ is at most $n+m\leq 2m$, and the above asymptotic bound on total update time remains unchanged.

\section{Proof of \Cref{obs: ball growing2}}
\label{sec: appx: ball growing 2}
%\begin{proof}
Let $\phi'=\frac{64\log n}{d}$.
	For all $i\geq 1$, let $S_i=B_H(v,i)=\set{u\in V(H)\mid \dist_H(v,u)\leq i}$ -- the ball of radius $i$ around $v$. Assume first that $v\in R$. Then for all odd indices $i\geq 1$, $S_i\setminus S_{i-1}$ only contains vertices of $L$, while for all even indices $i>1$, $S_i\setminus S_{i-1}$ only contains vertices of $R$. Therefore, for every even index $i>1$, all edges in $E_H(S_i,V(H)\setminus S_i)$ are special edges. We say that an even index $i>1$ is \emph{acceptable} if $|E_H(S_i,V(H)\setminus S_i)|\leq \phi'\cdot |S_i|$. We now claim that some even index $1<i\leq d-1$ must be acceptable. Indeed, assume for contradiction that each such even index $i$ is unacceptable. Consider any such even index $i$. 
	Then $|E_H(S_i,V(H)\setminus S_i)|>\phi'\cdot |S_i|$. Since all edges in $E_H(S_i,V(H)\setminus S_i)$ are special, every vertex in $S_{i+1}\setminus S_i$ is incident to exactly one such edge, and so $|S_{i+1}\setminus S_i|\geq |E_H(S_i,V(H)\setminus S_i)|>\phi'\cdot |S_i|$. In particular, $|S_{i+2}|\geq |S_{i+1}|\geq (1+\phi')\cdot |S_i|$. Let $r$ be the largest even integer that is at most $d-1$, so $r\geq d-2\geq d/2$. Then, since 	$\phi'=\frac{64\log n}{d}$:
	
	\[|S_r|\geq (1+\phi')^{r/2}\geq (1+\phi')^{d/4}=\left(1+\frac{64\log n}{d}\right )^{d/4}>n,\]
	
	a contradiction. 
	
	Our algorithm performs a BFS starting from $v$, until it reaches the first even index $i$  that is acceptable. It then returns a cut $(A,B)$ with $S_i=A$ and $B=V(H)\setminus A$. Clearly, $|E_H(A,B)|\leq \phi'\cdot |A|\leq \phi\cdot |A|$ holds from the definition of acceptable index. Moreover, since $|E_H(A,B)|\leq \phi'\cdot n= \frac{64\log n}{d}\cdot n$, while $|B|\geq |X|\geq \frac{n}{32}$, we get that $|E_H(A,B)|\leq \frac{2^{11}\log n}{d}|B|=\phi\cdot |B|$. Therefore, the sparsity of cut $(A,B)$ is at most $\phi$. It is immediate to verify that this cut has all other required properties. It is also immediate to verify that the running time of the algorithm is $O(\vol_H(A))$.
	
	If $v\in L$, then for all odd indices $i\geq 1$, $S_i\setminus S_{i-1}$ only contains vertices of $R$, while for all even indices $i>1$, $S_i\setminus S_{i-1}$ only contains vertices of $L$. Therefore, for every odd index $i\geq 1$, all edges in $E_H(S_i,V(H)\setminus S_i)$ are special edges.
	The algorithm is the same as before, except that we use odd indices $i$ instead of even indices.
%\end{proof}

%===========================
%===========================
%===========================
%===========================
%===========================
%===========================
%===========================

\bibliographystyle{alpha}

\bibliography{faster-classical-matching}

\end{document}